\documentclass[11pt]{article}
\usepackage[utf8]{inputenc}

\usepackage{geometry}
\usepackage[dvipsnames]{xcolor}
\usepackage{amsmath}
\usepackage{amsfonts}
\usepackage{amssymb}
\usepackage{graphicx}
\usepackage{url}
\usepackage[ruled, algosection]{algorithm2e}

\newcommand{\RR}{\mathbb{R}}

\newcommand{\EE}{\mathbb{E}}
\newcommand{\one}{\mathbf{1}}
\DeclareMathOperator{\poly}{poly}
\DeclareMathOperator{\argmax}{argmax}
\DeclareMathOperator{\spann}{span}
\DeclareMathOperator{\tr}{tr}

\usepackage{amsthm}
\newtheorem{theorem}{Theorem}
\newtheorem{lemma}[theorem]{Lemma}
\newtheorem{corollary}[theorem]{Corollary}

\newtheorem{definition}[theorem]{Definition}

\numberwithin{theorem}{section}
\numberwithin{conjecture}{section}
\numberwithin{problem}{section}

\newcommand\numberthis{\addtocounter{equation}{1}\tag{\theequation}}

\begin{document}

\author{Orestis Plevrakis\footnote{Princeton University, \texttt{or.plevrakis@gmail.com}.}
\and
Seyoon Ragavan\footnote{MIT, \texttt{sragavan@mit.edu}.}
\and
S. Matthew Weinberg\footnote{Princeton University, \texttt{smweinberg@princeton.edu}.}
}
\title{On the cut-query complexity of approximating max-cut}
\date{\today}
\maketitle
\thispagestyle{empty}



\abstract{We consider the problem of query-efficient global max-cut on a weighted undirected graph in the value oracle model examined by \cite{rubinstein}. Graph algorithms in this cut query model and other query models have recently been studied for various other problems such as min-cut, connectivity, bipartiteness, and triangle detection. Max-cut in the cut query model can also be viewed as a natural special case of submodular function maximization: on query $S \subseteq V$, the oracle returns the total weight of the cut between $S$ and $V \backslash S$.


Our first main technical result is a lower bound stating that a deterministic algorithm achieving a $c$-approximation for any $c > 1/2$ requires $\Omega(n)$ queries. This uses an extension of the cut dimension to rule out approximation (prior work of~\cite{graur} introducing the cut dimension only rules out exact solutions). Secondly, we provide a randomized algorithm with $\tilde{O}(n)$ queries that finds a $c$-approximation for any $c < 1$. We achieve this using a query-efficient sparsifier for undirected weighted graphs (prior work of~\cite{rubinstein} holds only for unweighted graphs).

To complement these results, for most constants $c \in (0,1]$, we nail down the query complexity of achieving a $c$-approximation, for both deterministic and randomized algorithms (up to logarithmic factors). Analogously to general submodular function maximization in the same model, we observe a phase transition at $c = 1/2$: we design a deterministic algorithm for global $c$-approximate max-cut in $O(\log n)$ queries for any $c < 1/2$, and show that any randomized algorithm requires $\Omega(n/\log n)$ queries to find a $c$-approximate max-cut for any $c > 1/2$. Additionally, we show that any deterministic algorithm requires $\Omega(n^2)$ queries to find an exact max-cut (enough to learn the entire graph).

}

\section{Introduction}\label{introduction}


For the most part, the field of graph algorithms has focused on extensive study in the ``standard'' model where the entire graph is given to the algorithm as input e.g. in the form of an adjacency list or an adjacency matrix. However, another growing body of work focuses on the case where the algorithm does not know the graph; rather, it has query access to an oracle that knows the graph, and aims to make as few queries as possible. Different models restrict the algorithm to make different types of queries e.g. individual edge queries, linear measurements~\cite{agm12a, AssadiCK21}, matrix-vector products~\cite{swyz19}, etc.

One model that has attracted particular attention is the \emph{cut query model}~\cite{gk00, Choi13, rubinstein, graur, MukhopadhyayN20, LeeLSZ21, ApersEGLMN22, ChakrabartyL23}. In this model, the algorithm has query access to the cut function of a graph $G$ with vertex set $[n]$: it provides subsets $S \subseteq [n]$ as queries and the oracle returns the total weight of all edges of $G$ with exactly one endpoint in $S$. An additional reason for the wide interest in this model is that the graph's cut function is in particular a submodular function;~\cite{graur} and~\cite{LeeLSZ21} obtained new lower bounds for the query complexity of submodular function minimization by considering the special case of graph min-cut in the cut query model.

A range of graph problems has been studied in the cut query model and related models. Perhaps the most fundamental is that of learning the entire graph~\cite{gk00, ChoiK08, BshoutyM12, Choi13, BshoutyM15}; once the entire graph is known to the algorithm, it has all the standard graph algorithm literature at its disposal to solve the problem of interest. A more subtle question is whether the given problem can be solved with fewer queries than one would need to learn the graph. This was first answered affirmatively for the case of min-cut~\cite{rubinstein, MukhopadhyayN20, ApersEGLMN22}: an undirected, unweighted graph requires $\Omega(n^2/\log n)$ queries to learn~\cite{gk00} but $O(n)$ queries suffice to find its min-cut. Since then, novel algorithms have been developed in the cut query model and related models for problems such as connectivity~\cite{agm12a, swyz19, AssadiCK21, ApersEGLMN22, ChakrabartyL23}, testing bipartiteness~\cite{agm12a}, and triangle detection~\cite{swyz19}.

In this work, we initiate the study of max-cut in the cut-query model. We aim to understand how many queries are necessary and sufficient for an algorithm to find a $c$-approximate global max-cut in an undirected, weighted graph.



\subsection{Our Results}\label{sec:results}

Our two main results are stated below. Both are concerned with the setting where $c \in (1/2, 1)$ i.e. one wants to do better than a straightforward greedy algorithm or guessing a random cut, but does not want an exact max-cut. The first result is a lower bound against deterministic algorithms, and the second result is a randomized algorithm.

\begin{theorem}\label{thm:informaldethalfhardness}
    (See Corollary \ref{cor:dethalfhardness} for a precise statement) For $c > 1/2$, any deterministic algorithm achieving a $c$-approximation requires $\Omega(n)$ queries.
\end{theorem}

\begin{theorem}\label{thm:informalrandhalfalgo}
    (See Corollary \ref{cor:randhalfalgo} for a precise statement) For $c < 1$, there exists a randomized algorithm with query complexity $\tilde{O}(n)$ that achieves a $c$-approximation.
\end{theorem}



These results naturally beg the question of how the query complexity behaves for other ranges of $c$, for both deterministic and randomized algorithms. We also answer this question in most cases. The content of all of our results is summarized in the following three theorems, and also tabulated in Figure \ref{summarytable}. Soon after, we provide context for these results. 

\begin{theorem}\label{thm:informalexact}
(See Theorems \ref{thm:exactdethardness} and \ref{thm:randhalfhardness} for precise statements) For $c = 1$, the query complexity for a deterministic algorithm to achieve a $c$-approximation is $\Theta(n^2)$ and the query complexity for a randomized algorithm to do the same is between $\tilde{\Omega}(n)$ and $O(n^2)$.
\end{theorem}

\begin{theorem}\label{thm:informalhalf}
(See Corollary \ref{cor:randhalfalgo}, Corollary \ref{cor:dethalfhardness}, and Theorem \ref{thm:randhalfhardness} for precise statements) For $c \in (1/2, 1)$, the query complexity for a deterministic algorithm to achieve a $c$-approximation is between $\Omega(n)$ and $O(n^2)$ and the query complexity for a randomized algorithm to do the same is $\tilde{\Theta}(n)$.
\end{theorem}

\begin{theorem}\label{thm:informalzero}
(See Corollary \ref{corollary:randalgorepetition}, Theorem \ref{thm:detlesshalfalgorithm}, and Corollary \ref{cor:detzerohardness} for precise statements) For $c \in (0, 1/2)$, the query complexity for a deterministic algorithm to achieve a $c$-approximation is $\Theta(\log n)$ and the query complexity for a randomized algorithm to do the same is $\Theta(1)$.
\end{theorem}


\begin{figure}
\begin{center}
    \begin{tabular}{|c|c|c|}
    \hline
    $c$ & Deterministic & Randomized \\
    \hline\hline
    $1$ & $\Theta(n^2)$ & $(\tilde{\Omega}(n), O(n^2))$ \\
    \hline
    $(1/2, 1)$ & $(\Omega(n), O(n^2))$ & $\tilde{\Theta}(n)$ \\
    \hline
    $(0, 1/2)$ & $\Theta(\log n)$ & $\Theta(1)$ \\
    \hline
    \end{tabular}
    \caption{Summary of our results. For each range of $c$, we state the query complexity (up to constant and logarithmic factors) that we show for achieving a $c$-approximation in both the deterministic and randomized settings. $(n, n^2)$ indicates settings where we have a lower bound of $\tilde{\Omega}(n)$ and an upper bound of $O(n^2)$.}
    \label{summarytable}
\end{center}
\end{figure}

We now provide context for each of our results, beginning with Theorem~\ref{thm:informaldethalfhardness}. Observe that there is a straightforward non-adaptive $O(n^2)$-query deterministic algorithm to find the max-cut exactly (observed in~\cite{rubinstein}, as it applies to the min-cut as well). The algorithm can query all singletons and sets of size 2 and from this learn $w_{i, j} = \frac{1}{2}(F(\left\{i\right\}) + F(\left\{j\right\}) - F(\left\{i, j\right\}))$ for all $i, j$ and thus the entire graph. The algorithm can then find the max cut exactly using brute force, since no more queries need to be made. This addresses all of the $O(n^2)$ upper bounds stated in the theorems. Theorem~\ref{thm:informalexact} shows that this trivial algorithm is optimal among deterministic algorithms: up to constant factors, any deterministic algorithm must learn the entire graph in order to find the global max-cut with cut queries.

Next, observe that there is a phase transition at $c = 1/2$. Theorem~\ref{thm:informalzero} establishes that a $(1/2-\varepsilon)$-approximation can be achieved deterministically with $O(\log n)$ queries (which happen to also be necessary). On the other hand, even a randomized algorithm needs $\tilde{\Omega}(n)$ queries to guarantee a $(1/2+\varepsilon)$-approximation.

Finally, observe that we resolve the asymptotic query complexity for most cases. For randomized algorithms, the only unresolved cases are $c= 1/2$ and $c = 1$.\footnote{For the $c = 1/2$ case, simply outputting a random cut will achieve a $(1/2)$-approximation in expectation in $O(1)$ queries. But we are concerned with the algorithm's ability to achieve a $c$-approximation with some probability $p > 0$. In this setting, there exists a deterministic $O(n)$-query algorithm (see Section \ref{sec:detgreedyalgo}) but we do not have a matching lower bound.} For deterministic algorithms, the range $[1/2, 1)$ remains unresolved. 




Before continuing, we note several facts about our results. In the positive direction, we note that all of our algorithms find a set $S$ that is a $c$-approximation to the global max-cut (for undirected, weighted graphs with arbitrarily large weights), rather than simply estimating the value of the max-cut within a factor of $c$. We refer to these two settings as the \emph{cut finding} and \emph{value estimation} settings respectively. On the other hand, most of our lower bounds hold even against algorithms in the value estimation setting. The only exception is our $\tilde{\Omega}(n)$ lower bound on randomized algorithms for $c > 1/2$, which we only prove in the cut finding setting. Similarly, all algorithms we discuss, except for our $\tilde{O}(n)$-query randomized algorithm for $c \in (1/2,1)$ and the well-known $O(n)$-query deterministic greedy algorithm for $c = 1/2$, are non-adaptive. On the other hand, all of our lower bounds hold even against adaptive algorithms.



It may appear at first glance that Theorem \ref{thm:informaldethalfhardness} is subsumed by the randomized lower bound in Theorem \ref{thm:informalhalf}. Both results are lower bounds for $c > 1/2$: the former states that a deterministic algorithm in this setting requires $\Omega(n)$ queries, while the latter says that a randomized algorithm requires $\tilde{\Omega}(n)$ queries. However, the deterministic lower bound is interesting on its own for two reasons. The first reason is that the randomized lower bound is only $\Omega(n/\log n)$, which is weaker than the $\Omega(n)$ we are able to prove in the deterministic case. Secondly, and more critically, the deterministic lower bound holds even for the value estimation setting while the randomized lower bound is only for the cut finding setting. Our argument for the randomized lower bound does not appear adaptable to the value estimation setting; indeed, the hard distribution used in our proof is actually very easy in the value estimation setting.\footnote{In more detail, the hard distribution is a random complete bipartite graph with edge weights 1. For this distribution in the value estimation setting, an algorithm could just immediately output $n^2/4$ without making any queries at all.} 

In the negative direction, we note here that most of our algorithms have exponential time complexity even though they are query-efficient. For example, the $O(n^2)$-query algorithm we just described learns the graph in polynomial time, but then uses exponential brute force search to actually find its max-cut. This illustrates that, by focusing on query complexity, we are primarily concerned with the information theoretic question of how much the algorithm must learn about the graph to be able to estimate its max-cut, even with unlimited computation. Note that the most query-efficient algorithms for submodular function minimization are exponential time as well~\cite{jiang21}, so this is not uncommon when prioritizing query complexity.

\subsection{Technical Highlights}
Our results follow from a wide array of techniques. Some results (e.g.~our $\tilde{\Omega}(n)$ lower bound on randomized algorithms for $c$-approximate max-cut and $c > 1/2$) follow from direct constructions of hard distributions --- we defer further discussion of these constructions to the corresponding technical sections and appendices. Some of our results follow from novel techniques that are likely of independent interest for subsequent work --- we quickly highlight these below.\footnote{Additionally, some of our results are reasonably straightforward (e.g.~our $O(\log n)$ deterministic algorithm for $c$-approximate max-cut when $c < 1/2$) --- we include discussion of these results in Appendices \ref{sec:randhalfhardness}, \ref{sec:dethardnesszero}, and \ref{sec:easyalgoresults} for completeness.}

The key ingredient in our $\Omega(n^2)$-lower bound for exact max-cut is the \emph{cut dimension} introduced by~\cite{graur} for min-cuts (which has also been studied in follow-up work~\cite{LeeLSZ21}). The lower bound follows immediately after establishing that the complete graph on $n$ vertices has max-cut dimension $\Omega(n^2)$. More interestingly, we extend the concept of the cut-dimension to what is roughly a notion of ``$c$-approximate cut-dimension" using strong LP duality. We show that this technique provides a lower bound on the number of queries needed by a deterministic algorithm to find a $c$-approximate max-cut for $c > 1/2$, and that for the complete graph on $n$ vertices this gives a bound of $\Omega(n)$. This approach should also be of independent interest, given recent interest in the exact cut dimension. Completing this analysis also requires a technical lemma (Lemma \ref{lemma:ell1}) stating a simple geometric property of the Boolean hypercube, which may additionally be of independent interest.

Our $\tilde{O}(n)$-query randomized $c$-approximation for $c \in (1/2, 1)$ follows from a query-efficient sparsifier for global cuts in undirected, weighted graphs (once we have learned a sparsifier for the graph, we can just exhaust to find its max-cut, which is a $(1-\varepsilon)$-approximation). A query-efficient sparsifier for unweighted graphs appears in~\cite{rubinstein}, based on ideas by~\cite{benczurkarger}. We first give a natural extension of their sparsifier that accommodates graphs with weights in $[1,\poly(n)]$. We then go a step further and provide a novel query-efficient sparsifier for \emph{all weighted graphs}, using other ideas from~\cite{benczurkarger} relating edge strengths to maximum spanning trees. For completeness, we show in Appendix \ref{sec:rswmaxcut} that in fact a slight modification of the~\cite{rubinstein} sparsifier also suffices for approximate max-cut for all weighted graphs (however, this modification may not produce a sparsifier). Our stronger sparsifier is hence not ``necessary'' to resolve approximate max-cut, but is of independent interest as not all weighted graph problems can be reduced to one with weights in $[1,\poly(n)]$ (e.g. exact min-cut). In fact, our stronger sparsifier is already used in recent work~\cite{MukhopadhyayN20}.\footnote{\cite[Theorem~5.2]{MukhopadhyayN20} seems to credit~\cite{rubinstein} for handling weighted graphs (whereas their query-efficient sparsifier only accommodates unweighted graphs). They use this result as part of their $\tilde{O}(n)$ query algorithm for \emph{weighted} min-cut.}




\subsection{Related Work}


\textbf{Learning graphs in the cut-query model.} The first natural question when an algorithm has restricted query access to a graph is how many queries it needs to learn the entire graph. For the case of cut queries, this was first studied by~\cite{gk00}, who show that the query complexity of deterministically learning an unweighted graph is $\Theta(n^2/\log n)$, and that this can be attained non-adaptively. A later line of work~\cite{ChoiK08, BshoutyM12, Choi13, BshoutyM15} answered this question for weighted graphs while also being sensitive to the number $m$ of edges in the graph with nonzero weight, finding that the query complexity in this case is $\Theta(m \log n/\log m)$.

\textbf{Min-cut in the cut-query model.} A very active line of work is that of understanding what properties an algorithm can learn about a graph that it only has query access to, without using up enough queries to learn the entire graph. The particular case of most relevance to this work is that of min-cut in the cut-query model. The first progress on this problem was the algorithm by~\cite{rubinstein} for \emph{unweighted} graphs using $\tilde{O}(n)$ queries. This was later improved to $O(n)$ queries by~\cite{ApersEGLMN22}. For \emph{weighted} graphs, an algorithm using $\tilde{O}(n)$ cut queries was presented by~\cite{MukhopadhyayN20} (as previously noted, their result leverages a query-efficient sparsifier for weighted graphs, which our paper is the first to provide). 

On the lower bound side,~\cite{graur} show using a linear algebraic argument that $\Omega(n)$ queries are necessary for weighted graphs (the constant-factors have since been improved by~\cite{LeeLSZ21}), thus resolving the query complexity of exact min-cut in weighted graphs up to logarithmic factors. As mentioned earlier, our deterministic lower bounds arise from adapting and extending the concept of the cut dimension to approximation.

The cut dimension was introduced by~\cite{graur}, and further studied in~\cite{LeeLSZ21}.~\cite{LeeLSZ21} further nail down that undirected graphs have a min-cut dimension of no more than $2n-3$ (and there exist graphs realizing this bound). They also introduce the $\ell_1$-approximate cut dimension. The $\ell_1$-approximate cut dimension is motivated by a similar thought experiment as our ``$c$-approximate cut dimension'' technique --- both consider an adversary changing the original graph in a way that respects all queries made, and both consider the magnitude of the changes made to the max/min-cut (rather than just whether or not the cut is changed). However, their $\ell_1$-approximate cut dimension does not imply lower bounds on the query complexity of finding approximately-optimal cuts. Finally, they further consider ``the dimension of approximate min-cuts'', and prove a linear upper bound on this. This concept does not have any relation to our approximate cut dimension technique (and in particular, only our technique implies lower bounds on the deterministic query complexity). In summary, the cut dimension is an active concept of study in related work, but our work is the first to use this concept to lower bound the query complexity of approximately-optimal cuts.

\textbf{Other graph algorithms in other query models.} Many works have considered other graph problems in other query models as well. These query models include cut queries~\cite{ChakrabartyL23}, linear measurements~\cite{agm12a, AssadiCK21}, matrix-vector products~\cite{swyz19}, OR queries~\cite{AssadiCK21, BlikstadBEMN22}, XOR queries~\cite{BlikstadBEMN22}, and AND queries~\cite{BlikstadBEMN22}. Some of these models e.g. linear measurements and matrix-vector products generalize the cut query model, while the others do not appear directly related. These works also consider a variety of graph problems, including connectivity~\cite{agm12a, swyz19, AssadiCK21, ApersEGLMN22, ChakrabartyL23}, bipartiteness~\cite{agm12a}, triangle detection~\cite{swyz19}, minimum spanning tree~\cite{agm12a}, and maximum matching~\cite{agm12a, BlikstadBEMN22}.

\textbf{Graph sparsification.}~\cite{rubinstein} also provide an algorithm for graph sparsification in $\tilde{O}(n)$ queries for unweighted graphs, based on the idea of \emph{edge strength-based sampling} used by \cite{benczurkarger} to construct sparsifiers in the classical computational model. There are also other lines of work that adapt the ideas from \cite{benczurkarger} to construct sparsifiers in other limited-access computational models, particularly (semi-)streaming \cite{agm12a, agm12b, agm13, klmms14, swyz19}. Some of these algorithms are adaptable to the cut-query model. However, these algorithms either only support unweighted graphs or incur a performance cost for weighted graphs. For example, the natural extension of the algorithm in \cite{rubinstein} to weighted graphs requires $\tilde{O}(n \log W)$ queries, where $W$ is the ratio between the largest and smallest nonzero edge weight in the graph. Our novel sparsifier avoids this pitfall by using results from \cite{benczurkarger} connecting edge strengths to maximum spanning trees to remove the $\log W$ factor.


While these results and ours all use randomization to construct a sparsifier, there are also classical algorithms for constructing sparsifiers deterministically \cite{bss14, silvahs16}. Implementing such an algorithm directly would require being able to learn the spectrum of the graph's Laplacian. It may indeed be possible to do this (exactly or approximately) in a query-efficient way, but doing so would require significantly new ideas. A deterministic query-efficient implementation of these algorithms would also yield a deterministic query-efficient algorithm for a $(1-\epsilon)$-approximation to max cut.

\textbf{Max-cut in the classical computational model.} Max-cut is extremely well-studied in the computational model, but algorithms in the cut-query model differ significantly. For example,~\cite{GW} present a randomized polynomial-time algorithm using semidefinite programming that achieves a $\approx 0.878$ approximation in expectation. However, this algorithm does not imply anything in our model; even formulating the necessary SDP requires access to the entire graph.

It has since been shown by~\cite{hastad} that it is NP-hard to beat a $16/17 \approx 0.941$ approximation. Moreover,~\cite{khot} show that it is NP-hard to beat the constant achieved by Goemans and Williamson assuming the unique games conjecture~\cite{khotugc}. These lower bounds do not imply anything in our model either. In fact, many of our algorithms (and those in prior work) involve exponential computation even if they use an efficient number of queries. Once we have either learned the entire graph ($c = 1$) or constructed a sparsifier of the graph ($c \in (1/2, 1)$), we use brute force search to find the exact max cut on the graph we have learned.

\textbf{Submodular function maximization.} Our problem can be viewed as a special case of symmetric submodular function maximization, by taking $f(\cdot)$ to be the cut function of a graph. However, the additional structure imposed by the cut function of a graph opens up possibilities for new algorithms and requires new lower bound arguments. For example, when $c > 1/2$, \cite{feige} show that achieving a $c$-approximation to symmetric submodular function maximization requires $\exp(\Omega(n))$ queries to the oracle. For global max-cut, the situation is significantly different: $O(n^2)$ queries suffice to trivially learn the entire function (and therefore, there seems to be little hope of embedding hard SFM instances as graphs). Still, we show that this trivial algorithm is optimal among deterministic algorithms when $c=1$, and show that $\tilde{\Omega}(n)$ queries are necessary for even a randomized algorithm to achieve a $c$-approximation for $c > 1/2$. 

For $c \leq 1/2$, randomized~\cite{buchbinder} and deterministic~\cite{buchbinderdet} algorithms are known that achieve a $1/2$-approximation with $O(n)$ and $O(n^2)$ queries respectively, even when the submodular function is not symmetric. \cite{feige} show in the symmetric case that picking a random set ($O(1)$ queries) achieves a $1/2$-approximation in expectation. They also provide a deterministic algorithm based on local search that achieves a $c$-approximation for any $c < 1/2$ in $\tilde{O}(n^3)$ queries. When restricting to max-cut,~\cite{feige} immediately implies a randomized algorithm for all $c \leq 1/2$ with $O(1)$ queries (which is optimal). For deterministic algorithms, however, our bound of $\Theta(\log n)$ is again an exponential improvement over the general case of submodular function maximization.


In summary, global max-cut demonstrates a similar phase transition at $c =1/2$ as general submodular function maximization: the query complexity for $c < 1/2$ is exponentially smaller than the query complexity for $c > 1/2$. However, the distinction for max-cut is between logarithmic and polynomial queries, whereas the distinction for general submodular functions is between polynomial and exponential.


\textbf{Submodular function minimization.} Just as submodular function maximization generalizes max-cut in the cut query model, submodular function minimization generalizes min-cut. Unlike the maximization case which requires exponentially many queries, general submodular function minimization can be solved exactly with polynomially many queries; state-of-the-art algorithms use $\tilde{O}(n^2)$ queries~\cite{jiang21, LeeSW15}.

However, until relatively recently, most query lower bounds for submodular function minimization were only $\Omega(n)$~\cite{Harvey08, graur, LeeLSZ21}. The last two results proceed using the aforementioned cut dimension technique. This was recently improved to $\Omega(n \log n)$ by~\cite{ChakrabartyGJS22}, using different ideas unrelated to graph cuts.

\section{Preliminaries}
$G$ is an undirected, weighted graph with weights $w_{i,j}$ on the edge $(i,j)$. $G$ induces a cut function $F$. We denote the cut function by $F(\cdot)$, so $F(S):=\sum_{i \in S, j \notin S} w_{i,j}$. Additionally, for any cut $S \subseteq [n]$, we define the indicator vector $v_S \in \RR^{\binom{n}{2}}$ by $(v_S)_{i, j} = 1$ if $(i, j)$ crosses the cut defined by $S$ and 0 otherwise.

We consider algorithms that have black-box access to $F(\cdot)$, and aim to find $\arg\max_{S \subseteq [n]}\{F(S)\}$ (exact maximization) or a set $T$ such that $F(T) \geq c\cdot \max_{S \subseteq [n]}\{F(S)\}$ ($c$-approximation). We refer to the query complexity as the number of queries to $F(\cdot)$ that the algorithm makes (the algorithm has no other access to $G$ or $F$, and may perform unlimited computation). 



\section{Lower Bound for Deterministic $(1/2 + \epsilon)$-approximation}\label{sec:dethalfhardness} 

In this section, we extend the cut dimension technique by~\cite{graur} using linear programming and the strong duality theorem to show the deterministic hardness part of Theorem \ref{thm:informalhalf}. (We refer the reader to Appendix \ref{sec:detexacthardness} for a discussion of this technique and its direct application to show the deterministic lower bound in Theorem \ref{thm:informalexact}.) Our exact result is as follows:

\begin{theorem}\label{thm:dethalfhardness}
Suppose we have $c \in (1/2, 1), \epsilon \in (0, 1), \alpha > 0$ such that $c > \frac{1}{1+\epsilon^2}$ and $\alpha < \frac{(1 - \epsilon)^3}{108(1+\epsilon)}$. Then for $n$ sufficiently large, $\alpha n$ queries do not suffice for a deterministic algorithm to estimate the max cut value within a factor of $c$.
\end{theorem}

Before proving this theorem, we state a cleaner lower bound as a corollary:
\begin{corollary}\label{cor:dethalfhardness}
For $c \in (1/2, 1)$, a deterministic algorithm that estimates the max cut value within a factor of $c$ requires at least $n(\frac{(\sqrt{c} - \sqrt{1-c})^4}{108c(2c-1)} - o(1))$ queries.

This implies that the query complexity for a deterministic algorithm to achieve a $c$-approximation for global max-cut on a weighted undirected graph in the value estimation setting is $\Omega(n)$.
\end{corollary}
\begin{proof}
See Appendix \ref{lemmaproof:cordethalfhardness}.
\end{proof}


The first step is conceptually similar to the cut dimension argument from Theorem \ref{thm:exactdethardness}. We consider an adversary that answers all queries as if the graph were $K_n$ (there is an edge of weight 1 between all pairs of vertices). Then we would like to find a perturbation $z$ to the weight vector of $K_n$ such that $w, w+z$ agree on all queries but have differing max cut values. The difference here is that we would like to show that the algorithm cannot even achieve a $c$-approximation, so we require the perturbation to be so large that the max cut value of $w+z$ is a multiplicative factor greater than the max cut of $w$. To do this, we will have to go beyond linear algebraic tools and consider linear programming instead. Our argument comprises the following steps:

\begin{enumerate}
    \item Write the conditions we require of the perturbation $z$ as linear constraints and thus formulate a linear program LP1, which we would like to show has high value.
    \item LP1 works with vectors in $\RR^{\binom{n}{2}}$ that represent cuts, which are unwieldy and unnatural. Rewrite this in terms of indicator vectors in $\RR^n$.
    \item Define another linear program LP2 and show that a high value for LP2 implies that LP1 must also have high value.
    \item Show that LP2 has high value by taking its dual, and showing that the dual has high value. This comes down to showing a key technical lemma, which essentially states that the $n$-dimensional hypercube cannot be covered by an $\ell_1$ neighborhood of an $\alpha n$-dimensional subspace of $\RR^n$. We show this using an $\ell_1$ $\epsilon$-net argument.
\end{enumerate}

We now work through each step in detail. We retain all notation used in Appendix \ref{sec:detexacthardness} for the cut dimension argument, and introduce additional notation as necessary.

\subsection{Step 1: Formulating LP1}\label{sec:step1}

Throughout these proofs, we use $\one$ to denote a vector with all entries equal to 1 in either $\RR^n$ or $\RR^{\binom{n}{2}}$. It will be clear from context which of these we are referring to at any given time, but for now we are taking $\one \in \RR^{\binom{n}{2}}$ to denote the weight vector of $K_n$. Let $q = \alpha n$ and $Q_1, Q_2, \ldots, Q_q \subseteq [n]$ be the $q$ queried cuts. As in Appendix \ref{sec:detexacthardness}, they have the corresponding 0/1 indicator vectors $v_{Q_1}, v_{Q_2}, \ldots, v_{Q_q} \in \RR^{\binom{n}{2}}$. We are interested in finding a perturbation $z$ such that $\one, \one + z$ are both weighted, undirected graphs (with non-negative edge weights) and agree on all queries but have max cut values differing by a multiplicative factor.

First, we require $\one + z$ to define a valid graph i.e. its entries should all be non-negative since these correspond to edge weights:
\begin{align}\label{eqn:posconstraint}
    \one + z \geq 0 &\Leftrightarrow z \geq -\one.
\end{align}
Next, we need $\one, \one + z$ to agree on all queries. This guarantees that the algorithm cannot tell the difference between $\one$ and $\one+z$ based only on the queries made so far.
\begin{align}\label{eqn:agreementconstraint}
    \one^Tv_{Q_i} = (\one + z)^T v_{Q_i} ,\ \forall i &\Leftrightarrow z^T v_{Q_i} = 0,\ \forall i.
\end{align}
Finally, we would like the graph corresponding to $\one + z$ to have a much larger max cut value than the graph corresponding to $\one$. We capture this in the following definition and lemma:

\begin{definition}
Define a \emph{near-max cut} to be any cut $C \subseteq [n]$ such that $n/2 - \sqrt{n \log n} \leq |C| \leq n/2 + \sqrt{n \log n}$.
\end{definition}

Note that a near-max cut is nearly a max-cut in $K_n$. As hinted at earlier, we will show that we can find a near-max cut $C$ and a perturbation $z$ to the graph that will preserve the value of all queries while blowing up the value of $C$ by a factor of $c$. In this case, the algorithm cannot distinguish between $K_n$ and the perturbed graph and thus cannot achieve a $c$-approximation.

\begin{lemma}\label{lemma:cutduality}
To prove Theorem \ref{thm:dethalfhardness}, it suffices to show that there exists a near-max cut $C$ such that the following linear program has value at least $\epsilon^2n^2/4$. We call this linear program LP1.
\begin{align*}
    \text{Maximize }& z^Tv_C \\
    \text{subject to }
    &z^Tv_{Q_j} = 0 \quad \text{ for all } j \in [q], \text{ and } \\
    &z \geq -\one.
\end{align*}
\end{lemma}
\begin{proof}
We have already explained how the constraints arise. To justify that the objective corresponds to a bound on the approximation ratio, see Appendix \ref{lemmaproof:cutduality}.
\end{proof}

\subsection{Step 2: Rewriting LP1}

As already mentioned, the vectors $v_S \in \RR^{\binom{n}{2}}$ are unnatural and difficult to work with. Intuitively, the reason for this is that a cut only has $n$ degrees of freedom (each vertex can be included or not included in $S$), but we are representing it with a vector with $\binom{n}{2}$ entries, thereby creating unwanted dependencies between entries of these vectors.

We would thus like to find a more natural parametrization of these cuts that can still be connected naturally to the vectors $v_S$. To construct such a parametrization, we assign to each cut $S$ a $\pm 1$ indicator vector $u_S \in \RR^n$. Entries are indexed by vertices in $[n]$, and for each $i \in [n]$ we have $(u_S)_i = 1$ if $i \in S$ and $-1$ otherwise.

Now we connect these new indicator vectors to $v_S$ as follows. Define the matrix $M_S \in \RR^{n \times n}$ by $M_S = \frac{\one \one^T - u_S u_S^T}{2}$. (Note that from here onwards, $\one$ now refers to the vector in $\RR^n$ with all entries equal to 1.) $M_S$'s entries are indexed by ordered pairs of vertices. Now observe that:
\begin{align*}
    (M_S)_{i, j} &= \frac{\one_i \one_j - (u_S)_i (u_S)_j}{2} \\
    &= \frac{1 - (u_S)_i (u_S)_j}{2} \\
    &= \begin{cases}
       0,&  i, j \in S\text{ or }i, j \notin S,\\
       1,& \text{ otherwise}
    \end{cases} \\
    &= \begin{cases}
       (v_S)_{i, j},&  i \neq j,\\
       0,& i = j.
    \end{cases}
\end{align*}
Thus if we flatten $M_S$ into a vector, it consists of two copies of $v_S$ (each unordered vertex pair $(i, j)$ in $v_S$ appears twice in $M_S$ since the vertex pairs indexing $M_S$ are ordered) and $n$ 0's (corresponding to vertex pairs $(i, i)$ for $i \in [n]$). This allows us to rewrite LP1 in terms of the $u_S$'s as stated in the following lemma. For matrices $A, B$ of the same shape, we use $\langle A, B \rangle$ to denote the matrix inner product $\tr(A^TB) = \sum_{i, j} A_{i, j} B_{i, j}$.

\begin{lemma}\label{lemma:step2}
For any $C$, LP1 has value $\geq \epsilon^2n^2/4$ if and only if the following LP has value at least $\epsilon^2n^2/2$. We call this the ``matrix LP". Here, $Z \in \RR^{n \times n}$.
\begin{align*}
    \text{Maximize }& \langle Z, M_C \rangle \\
    \text{subject to }
    &\langle Z, M_{Q_j} \rangle = 0 \quad \text{ for all } j \in [q], \text{ and } \\
    &Z \geq -\one.
\end{align*}
\end{lemma}

\begin{proof}
See Appendix \ref{lemmaproof:step2}.
\end{proof}

\subsection{Step 3: Defining LP2 and Connecting LP2 to LP1}
For a near-max cut $C$, define a new linear program which we call LP2 as follows. Here $y \in \RR^n$.
\begin{align*}
    \text{Maximize }& y^Tu_C \\
    \text{subject to }&
    y^Tu_{Q_j} = 0 \quad \text{ for all } j \in [q], \\
    &y^T\one = 0, \text{ and } \\
    &-\one \leq y \leq \one.
\end{align*}
We claim that it suffices to show that LP2 has value at least $\epsilon n$:

\begin{lemma}\label{lemma:step3}
If there exists a near-max cut $C$ such that LP2 has value at least $\epsilon n$ then the matrix LP for $C$ has value at least $\epsilon^2n^2/2$, which would imply Theorem \ref{thm:dethalfhardness}.
\end{lemma}
\begin{proof}[Proof Sketch]
Take such a near-max cut $C$ and $y \in \RR^n$ such that $y$ is in the feasible region of LP2 and $y^T u_C \geq \epsilon n$. Then we claim that $Z = -yy^T$ is feasible for the matrix LP and attains a value $\geq \epsilon^2n^2/2$. {Intuitively, $y$ can be thought of as a vector representing a ``pseudo-cut" in the same way that $u_S$ represents $S$. LP2 having high value means that $y$ aligns non-trivially with $C$, and what we are claiming is that perturbing towards the ``pseudo-cut" corresponding to $y$ will align the graph's weights with the cut corresponding to $C$.} We provide details in Appendix \ref{lemmaproof:step3}.
\end{proof}

\subsection{Step 4: Showing that LP2 has High Value}

Finally, we show that we can find a near-max cut $C$ such that LP2 has value at least $\epsilon n$, which by Lemma \ref{lemma:step3} will complete the proof of Theorem \ref{thm:dethalfhardness}. We do this by taking the dual of LP2, which has a simple characterization captured by the following lemma:

\begin{lemma}\label{lemma:duality}
Consider vectors $w, w_1, w_2, \ldots, w_k \in \RR^d$, and the following LP:
\begin{align*}
    \text{Maximize }& z^Tw \\
    \text{subject to }&
    z^Tw_i = 0 \quad \text{ for all } i \in [k], \text{ and } \\
    &-\one \leq z \leq \one.
\end{align*}
Let $W = \spann(w_1, w_2, \ldots, w_k)$. Then the value of this LP is $\min_{v \in W} ||v - w||_1$. (If there is no such $v$ then by this minimum we mean $\infty$.)
\end{lemma}
\begin{proof}
See Appendix \ref{lemmaproof:duality}.
\end{proof}

To use this lemma, let $V = \spann(u_{Q_1}, \ldots, u_{Q_q}, \one)$. Then Lemma \ref{lemma:duality} tells us that LP2 has value equal to $\min_{u \in V} ||u - u_C||_1$.

Now note that $V$ depends only on the set of queries and not at all on $C$. Thus we would like to show that there exists a near-max cut $C$ such that $\min_{u \in V} ||u - u_C||_1 \geq \epsilon n$. We will show the strict version of this inequality i.e. that $\min_{u \in V} ||u - u_C||_1 > \epsilon n$. Denote by $B_r$ the $\ell_1$ ball of radius $r$ in $\RR^n$. Then what we want to show is that there exists a near-max cut $C$ such that $u_C \notin V + B_{\epsilon n}$. This brings us to our key technical lemma, which has little to do with graphs and may be of independent interest:

\begin{lemma}\label{lemma:ell1}
Let $\epsilon \in (0, 1)$ and $d \leq \alpha'n$ for $\alpha' < \frac{(1-\epsilon)^3}{108(1+\epsilon)}$. Suppose $D$ is a $d$-dimensional subspace of $\RR^n$. Denote by $B_r$ the $\ell_1$ ball of radius $r$ in $\RR^n$. Then there exists $p \in \left\{-1, 1\right\}^n$ such that $p \notin D + B_{\epsilon n}$ and $|\one^T p| \leq 2\sqrt{n\log n}$.
\end{lemma}
\begin{proof}
We show this by a volume argument. Specifically, we estimate the size of $(D + B_{\epsilon n}) \cap \left\{-1, 1\right\}^n$ using an $\ell_1$ $\epsilon$-net argument and show that this must be much less than $2^n$. We provide details in Appendix \ref{lemmaproof:ell1}.
\end{proof}

With this lemma, we can complete this step and thus the proof of Theorem \ref{thm:dethalfhardness}:

\begin{corollary}\label{cor:endofstep4}
For $n$ sufficiently large, there exists a near-max cut $C$ such that $u_C \notin V + B_{\epsilon n}$.
\end{corollary}
\begin{proof}
See Appendix \ref{lemmaproof:corendofstep4}.
\end{proof}

\section{Sparsifier-based Randomized Algorithms for $(1-\epsilon)$-approximation}\label{sec:randhalfalgo}


Here we address the randomized upper bound part of Theorem \ref{thm:informalhalf}, namely that a $(1-\epsilon)$-approximation can be achieved in the cut finding setting with $\tilde{O}(n)$ queries. Our algorithms are adaptive. The key notion is that of a \emph{sparsifier}:

\begin{definition}
Given weighted graphs $G, H$ on the same set of $n$ vertices with non-negative weights, we say $H$ is an $\epsilon$-sparsifier of $G$ if all of the following conditions hold:

\begin{enumerate}
    \item $H$ has $\tilde{O}_\epsilon(n)$ edges with nonzero weight.
    \item For any cut $S \subseteq [n]$, we have $(1-\epsilon) F(S; G) \leq F(S; H) \leq (1+\epsilon) F(S; G)$.
\end{enumerate}
Here $F(S; G)$ denotes the value of the cut defined by $S$ on the graph $G$, and similarly for $F(S; H)$.
\end{definition}

\begin{lemma}\label{lemma:sparsifiersuffices}
If an algorithm can compute an $\epsilon$-sparsifier of $G$ in $\tilde{O}(n/\epsilon^2)$ queries with high probability, then it can also find a $(1-2\epsilon)$-approximate max cut with no additional queries.
\end{lemma}
\begin{proof}
Once the algorithm has a sparsifier $H$, it can just try all possible cuts and output the cut $U$ maximizing $F(U; H)$. Indeed, for any other cut $S$, we would have $F(U; G) \geq \frac{F(U; H)}{1+\epsilon} \geq \frac{F(S; H)}{1+\epsilon} \geq \frac{(1-\epsilon)F(S; G)}{1+\epsilon} \geq (1-2\epsilon)F(S; G)$, so $U$ is indeed a $(1-\epsilon)$-approximate max cut.
\end{proof}

Throughout this section, let $d > 5$ be constant, and let $\delta \in (1/2, 1)$ be a constant sufficiently close to 1. Then let $c_0$ and $c_1$ be sufficiently large positive constants. Also, let $W_\text{tot}$ be the total weight of all edges in the graph, and $W$ the ratio between the largest and smallest nonzero edge weights. Finally, for a vertex set $S \subseteq V(G)$, let $G[S]$ denote the vertex-induced subgraph of $G$ on $S$. Note that we do not require that $G$ be connected. We organize the remainder of this section as follows:

\begin{enumerate}
    \item In Section \ref{sec:sparsifiersetup}, we set up and analyze the algorithmic tools necessary to adapt \cite{rubinstein}'s algorithm to weighted graphs.
    \item In Section \ref{sec:rswadaptation}, we present the direct adaptation of \cite{rubinstein}'s algorithm to weighted graphs and show that it constructs a sparsifier in $\widetilde{O}(n\log W)$ queries.
    \item In Section \ref{sec:fws}, we use ideas introduced by \cite{benczurkarger} relating edge strengths to maximum spanning trees in order to construct a sparsifier for weighted graphs in $\widetilde{O}(n)$ queries, thus eliminating the dependence of the query complexity on $W$. 
    \item Additionally, in Appendix \ref{sec:rswmaxcut}, we show that we can achieve a $(1-\epsilon)$-approximation for max-cut in $\widetilde{O}(n)$ queries without needing the optimizations in Section \ref{sec:fws}. This is achieved by essentially stopping our adaptation of \cite{rubinstein}'s algorithm early. This will not construct a sparsifier, but it will construct something ``close enough" to suffice for max-cut. Intuitively, we do this by discarding edges of weight $< W_\text{tot}/\text{poly}(n)$ since these will not have much effect on the max cut, thereby reducing the problem to one where $W = \text{poly}(n)$.

    Given this result, our sparsifier for weighted graphs in Section \ref{sec:fws} is not necessary for max-cut specifically, but it makes for a conceptually cleaner algorithm for max-cut and may be applicable to other problems.
\end{enumerate}


\subsection{Setup and Algorithmic Tools}\label{sec:sparsifiersetup}

The key idea is the notion of edge strength introduced by \cite{benczurkarger}.

\begin{definition}
(\cite{benczurkarger}, as stated in \cite{rubinstein}) The \emph{strong connectivity} of $G$, denoted $K(G)$, is the value of $G$'s min cut.
\end{definition}

\begin{definition}
(\cite{benczurkarger}, as stated in \cite{rubinstein}) Given an edge $e$ in $G$, the strong connectivity or \emph{edge strength} $k_e$ of $e$ is the maximum min cut over all vertex-induced subgraphs of $G$ containing $e$:
$$k_e = \max_{S \subseteq V: u, v \in S} K(G[S]).$$
\end{definition}


The idea introduced by \cite{benczurkarger} and used by \cite{rubinstein} is that subsampling edges of $G$ with probabilities inversely proportional to their strength will give a sparsifier. We cannot do exactly this in the cut query model, but we can subsample in a way that is ``close enough" to independent. We need some basic primitives to support our algorithm, and we capture all of them in the following lemma:

\begin{lemma}\label{lemma:allalgoprimitives}
There exists a data structure supporting the following operations.
\begin{enumerate}
    \item $\text{InitializeDS}(H)$: Initialize the structure's state and carry out any preprocessing needed with the starting graph $H$.
    \item $\text{Contract}(S)$: Contract a given supernode set $S$.
    \item $\text{GetEdge}()$: Find and return an edge from $H$ (that has not been contracted) with weight at least $\frac{2W(H)}{n^2}$. Here $W(H)$ is the total weight of not-yet-contracted edges.
    \item $\text{GetTotalWeight}(S)$: Return the total weight of all edges with both endpoints in $S$ (that have not yet been contracted). $S$ must once again be a set of supernodes.
    \item $\text{Sample}(S)$: Sample a random edge with both endpoints in $S$ (that has not yet been contracted), with probability proportional to its weight. $S$ must be a set of supernodes here as well. 
\end{enumerate}
It takes $O(n)$ queries for each call to InitializeDS, $O(1)$ queries for each call to Contract, $O(\log n)$ queries for each call to GetEdge, $O(1)$ queries for each call to GetTotalWeight, and $O(\log n)$ queries for each call to Sample.
\end{lemma}
\begin{proof}
The pre-processing and queries when contracting are due to having to keep track of (super)node degrees. GetEdge and Sample can be handled using recursive bisection procedures similar to that used to prove Corollary 2.2 in \cite{rubinstein}. We provide details in Appendix \ref{lemmaproof:allalgoprimitives}.
\end{proof}

Before we go any further, we make an important comment about how we regard contraction in our algorithms (including Lemma \ref{lemma:allalgoprimitives}). When we contract a set of vertices, we regard that set of vertices as one supernode as usual, but we do not merge any edges that have now become parallel. So the set of edges will always be a subset of the edges from the original graph.

It will be convenient to regard our algorithms as having two separate stages, although the two stages share some ingredients. In the first stage, the algorithm iteratively subsamples and contracts the graph to estimate the strengths of all edges within a constant factor. In the second stage, the algorithm uses these edge strength estimates to construct the sparsifier, following the ideas of \cite{benczurkarger}. We address these two stages in the next two subsections respectively.


We note that these algorithms are very similar to those presented by \cite{rubinstein}; the key difference is that whenever the sparsification algorithm in \cite{rubinstein} subsamples the graph, it does so independently for each edge. This works because \cite{rubinstein} focuses on unweighted graphs. With weighted graphs, we would like to sample edges proportionately to their weight, and this cannot be done independently without knowledge of the graph. We modify their subsampling procedures to obtain algorithms that do not sample completely independently, but still have the concentration properties that we need.

\subsubsection{Constant-Factor Edge Strength Estimates}

We next describe the main tool of our edge strength estimation, which we call EstimateAndContract. It takes in the input graph $G$ where some vertex sets $S_i$ have already been contracted and a strength parameter $\kappa$, and further contracts $G$ while also estimating the strength of any edges that get contracted. For any graph $H$, let $W(H)$ denote the total edge weight of $H$.

\begin{algorithm}
\KwData{Initial graph $G$ with vertex set $V(G) \subseteq [n]$, disjoint collection $\mathcal{C}$ of contracted sets, strength parameter $\kappa$, partial list $X$ of strength estimates}
\KwResult{Updated collection of contracted sets $\mathcal{C}$, updated list $X$ of strength estimates}
 Let $G'$ be $G$ with all sets from $\mathcal{C}$ contracted and $\lambda = \frac{c_0\log^2 n}{\kappa} W(G')$. (We can find $W(G')$ using GetTotalWeight from Lemma \ref{lemma:allalgoprimitives}.)\\
 \begin{enumerate}
     \item Sample $\lambda$ edges from $G'$, proportionally to their weights, with replacement. If an edge $e$ is sampled at least once, assign weight $\frac{w(e)}{p(e)}$ to it, where $p(e) = 1 - (1 - \frac{w_e}{W(G')})^\lambda$. These newly weighted edges form a new graph $G''$.
     \item While there exists a connected component of $G''$ with a cut of size $\leq (1-\delta)\kappa$, delete all edges of that cut from $G''$. Let the resulting connected components of $G''$ be $C_1, C_2, \ldots, C_r$.
     \item For each $i \in [r]$, append the tuple $(C_i, \kappa/2)$ to $X$. (Here, what we are saying is ``assign a strength estimate of $\kappa/2$ to any edge with both endpoints in $C_i$ that has not already been contracted", but we phrase it differently to account for the fact that the algorithm may not actually know these edges.)
     \item Add $C_i$ to $\mathcal{C}$ (and remove any subsets of $C_i$ to maintain disjointness) for all $i \in [r]$.
 \end{enumerate}
 \caption{$\text{EstimateAndContract}(G, \mathcal{C}, \kappa, X)$}\label{algo:contract}
\end{algorithm}



We state some key properties of Algorithm \ref{algo:contract} and defer their proofs to Appendix \ref{lemmaproofs:subsampleandcontract}:

\begin{lemma}\label{lemma:claim1-2}
(Analogous to claim 3.6 from \cite{rubinstein}) With probability $1 - O(n^{1-d})$, for all $e$ such that $k_e \geq \kappa$ and $e$ is not contracted by any of the sets in $\mathcal{C}$, it will be assigned $k'_e = \kappa/2$ and contracted.
\end{lemma}

\begin{lemma}\label{lemma:claim1-3}
(Analogous to claim 3.7 from \cite{rubinstein}) Assume that any edge contracted by $\mathcal{C}$ has strength $\geq \kappa/2$. Then with probability $1 - O(n^{1-d})$, no edge $e$ such that $k_e < \kappa/2$ is assigned a strength estimate or contracted.
\end{lemma}




\subsubsection{Sparsifier Construction}

Next, we describe our algorithm that will construct a sparsifier if provided constant-factor strength estimates for all edges in the graph, which we call ConstructSparsifier.

\begin{algorithm}
\KwData{Graph $G$ with vertex set $[n]$, list $X = [(C_1, \beta_1), (C_2, \beta_2), \ldots, (C_r, \beta_r)]$ of strength estimates. We assume that $\left\{C_1, \ldots, C_r\right\}$ is a laminar family of vertex sets and that if $C_i \subset C_j$ then $i < j$. (This is because we will be contracting $C_1, \ldots, C_r$ in that order.) Note that we do \textbf{not} assume here that $X$ includes a strength estimate for every edge.}
\KwResult{Potential sparsifier $H$}
 Initialize $G' = G$ and $H$ to be empty.\\
 $\text{InitializeDS}(G')$. (This is just to reset and ignore any previous contractions we may have done in EstimateAndContract.)\\
 \For{$i \gets 1$ \KwTo $r$} {
    \begin{enumerate}
        \item Let $\mu_i = \frac{c_1 \log^2 n}{\epsilon^2 \beta_i} W(C_i)$. (Here $W()$ is with respect to $G'$, and once again can be found using GetTotalWeight from Lemma \ref{lemma:allalgoprimitives}.)
        \item Sample $\mu_i$ edges from $C_i$ proportionally to their weights, with replacement. If an edge $e$ is sampled at least once, add it to $H$ with weight $\frac{w_e}{p_e}$, where $p_e = 1 - (1 - \frac{w_e}{W(C_i)})^{\mu_i}$ is the probability that $e$ is sampled at least once. (Note that this sampling would be done by calling $\text{Sample}(C_i)$ from Lemma \ref{lemma:allalgoprimitives}.)
        \item Contract $C_i$ in $G'$.
    \end{enumerate}
 }
 \caption{$\text{ConstructSparsifier}(G, X)$}\label{algo:sparsify}
\end{algorithm}
We capture the desired sparsification properties in the following lemma:
\begin{lemma}\label{lemma:sparsifierconcentration}
Fix $X = [(C_1, \beta_1), (C_2, \beta_2), \ldots, (C_r, \beta_r)]$ as in ConstructSparsifier ($X$ might be random, but we condition on a particular list of values for now). Then for each $i$, define:
\begin{align*}
    E(C_i) &= \left\{e \in G: \text{ both endpoints of $e$ are in $C_i$}\right\} \\
    \widetilde{E}(C_i) &= E(C_i) \backslash \bigcup_{j: C_j \subset C_i} E(C_j)
\end{align*}
(Thus $\widetilde{E}(C_i)$ is the set of edges that will be contracted at the time that $C_i$ is contracted.)

Assume each of the following conditions:
\begin{enumerate}
    \item Each connected component of $G$ is contained in at least one $C_i$ (every edge in $G$ gets contracted), and
    \item For all $i$ and edges $e \in \widetilde{E}(C_i)$, we have $\beta_i \in [k_e/4, k_e]$ (edge strength estimates are correct within a constant factor).
\end{enumerate}

Then with probability $1 - O(n^{-d})$, ConstructSparsifier will output a sparsifier $H$ that approximates the cuts of $G$ within a factor of $1 \pm 2\epsilon$. (Thus this probability is only considering the randomness of ConstructSparsifier.)
\end{lemma}
\begin{proof}
This follows using similar ideas to \cite{benczurkarger}, but we need to take some extra care because the subsampling we use to construct the sparsifier is not independent. We provide details in Appendix \ref{lemmaproof:sparsifierconcentration}.
\end{proof}


\subsection{Sparsification with $\widetilde{O}(n\log W)$ Queries}\label{sec:rswadaptation}

Here we analyze the naive generalization of \cite{rubinstein}'s sparsifier to weighted graphs. The procedure is described in Algorithm \ref{algo:nws}.

\begin{algorithm}
\KwData{Graph $G$ on $n$ vertices, positive real parameter $T$}
\KwResult{A potential $(2\epsilon)$-sparsifier $H$ of $G$.}
Initialize $\mathcal{C} = \emptyset$ and $X$ as an empty list.\\ 
$\text{InitializeDS}(G)$. \\
Find $W_\text{tot}$ by running $\text{GetTotalWeight}(G)$. \\
Initialize $\kappa$ to be the smallest power of 2 that is at least $W_\text{tot}$.\\
\While{$\text{GetTotalWeight}(G) > 0$ and $\kappa > W_\text{tot}/T$}{
    $\text{EstimateAndContract}(G, \mathcal{C}, \kappa, X)$ \\
    $\kappa \gets \kappa/2$ \\
}
$H \gets \text{ConstructSparsifier}(G, X).$\\
 \caption{$\text{NaiveWeightedSubsample}(G, T)$}\label{algo:nws}
\end{algorithm}

\begin{theorem}\label{thm:nwsworks}
$\text{NaiveWeightedSubsample}(G, \infty)$ runs in $O(n\log^3n(\log n + \log W + \frac{1}{\epsilon^2}))$ queries and outputs a $(2\epsilon)$-sparsifier $H$ with probability $1 - O(n^{1-d}(\log n + \log W))$.
\end{theorem}
\begin{proof}
We outline the proof here and provide details in Appendices \ref{lemmaproofs:nwscorrect} (correctness) and \ref{sec:sparsifierspeed} (efficiency). The proof proceeds in two parts.

First, we address the calls to EstimateAndContract. Given Lemmas 4.6 and 4.7, EstimateAndContract can be thought of as estimating the strengths of and then contracting edges whose strengths are within a window that is a factor of 4 wide. The lemmas tell us that these strength estimates are accurate within a factor of 4. Then NaiveWeightedSubsample essentially ``slides" this window across all possible edge strengths so that all edges of $G$ are assigned a strength estimate. This part has query complexity $O(n\log^3 n(\log n + \log W))$; the $\log W$ term is because the outer loop could run for $O(\log n + \log W)$ iterations. This is also why the success probability depends on $W$, as this is obtained by taking a union bound over all iterations.

Secondly, because all edges are assigned an accurate strength estimate (within a constant factor), Lemma \ref{lemma:sparsifierconcentration} tells us that ConstructSparsifier will output a $(2\epsilon)$-sparsifier with high probability. This part has query complexity $O(n\log^3n/\epsilon^2)$.
\end{proof}

\subsection{Sparsification with $\widetilde{O}(n)$ Queries}\label{sec:fws}

We now show how to eliminate the dependence of the query complexity and success probability on $W$, thus constructing sparsifiers in $\widetilde{O}(n)$ queries. We begin by setting up the necessary ideas.


\subsubsection{Crude Edge Strength Estimates}\label{sec:crudestrengthestimates}


Recall that the key problem with Algorithm \ref{algo:nws} was that our ``sliding window" for edge strength estimation could potentially repeat $O(\log n + \log W)$ times. The key idea is to mitigate this by finding very crude (within a factor of $n^4$) estimates for the edge strengths for all edges in $G$, before refining these estimates using EstimateAndContract. We do this using the idea of \cite{benczurkarger} to estimate edge strengths from the maximum spanning forest (MSF) of $G$. Fix an MSF $\mathcal{T}$ of $G$. Then for any edge $e$ with endpoints $i, j$, define $d_e = d_{i, j}$ to be the minimum weight of an edge on the MSF path between the endpoints of $e$. We first state a lemma shown in \cite{benczurkarger}:

\begin{lemma}\label{lemma:MSTapprox}
(\cite{benczurkarger}) For all edges $e$, we have $d_e \leq k_e \leq n^2d_e$.
\end{lemma}

This would immediately give us sufficient edge strength estimates but we do not know of a way to find the MSF efficiently in the cut query model. In fact, we can adapt the max cut dimension argument from Lemma \ref{lemma:gprw} to a ``max tree dimension" argument to show that deterministically finding the MSF requires $\Omega(n^2)$ queries; see Appendix \ref{sec:MSThardness} for details. So instead, we run what we call ``approximate Kruskal" using the primitives we have available from Lemma \ref{lemma:allalgoprimitives}.

\begin{algorithm}
\KwData{Graph $G$ on $n$ vertices}
\KwResult{A forest $\mathcal{\widetilde{T}}$ using the edges of $G$.}
Initialize $\mathcal{\widetilde{T}}$ to be the empty graph on $n$ vertices. \\
$\text{InitializeDS}(G)$. \\
\While{$\text{GetTotalWeight}(G) > 0$}{
    $e = \text{GetEdge}()$\\
    $\text{Contract}(e)$\\
    Add $e$ to $\mathcal{\widetilde{T}}$.
}
 \caption{$\text{ApproximateKruskal}(G)$}\label{algo:apxkruskal}
\end{algorithm}

It follows from Lemma \ref{lemma:allalgoprimitives} that ApproximateKruskal can be run in $O(n \log n)$ queries. Once we run ApproximateKruskal, we will have a spanning forest of $G$ so we will also know its connected components. Moreover, the following lemma tells us that even this crude approximation to the MSF suffices to give us edge strength estimates. We defer the proofs of this lemma and its straightforward corollary to Appendix \ref{lemmaproof:apxkruskal}.

\begin{lemma}\label{lemma:apxkruskal}
For any distinct $i, j$, define $\tilde{d}_{i, j}$ as follows:
\begin{itemize}
    \item If $i, j$ are connected in $G$, then let $\tilde{d}_{i, j}$ be the minimum weight of an edge on the path in $\mathcal{\widetilde{T}}$ connecting $i$ and $j$.
    \item Otherwise, let $\tilde{d}_{i, j} = 0$.
\end{itemize}
Then for any $i, j$ that are connected in $G$, we have $\frac{2}{n^2} d_{i, j} \leq \tilde{d}_{i, j} \leq d_{i, j}$ for all $i, j$.

Note that we only assume here that $\mathcal{T}$ is maximal; any assumptions we make about $\widetilde{\mathcal{T}}$ are baked into GetEdge.
\end{lemma}

\begin{corollary}\label{cor:strengthestimate}
For all distinct $i, j$, we have $k_{i, j} \in [\tilde{d}_{i, j}, \frac{n^4}{2} \tilde{d}_{i, j}]$.
\end{corollary}

\subsubsection{Fast Weighted Subsampling}

We now describe how to use our crude edge strength estimates to construct a sparsifier in $\tilde{O}(n)$ queries. The idea is that by Lemma \ref{lemma:claim1-3}, any edges with strength $< \kappa/2$ will be deleted in Step 2 of Algorithm \ref{algo:contract}. But we can use our crude edge strength estimates to preemptively identify some edges that will definitely be deleted, and then delete these edges to disconnect the graph a bit \emph{before} running EstimateAndContract. We describe this procedure in Algorithm \ref{algo:fws}.

\begin{algorithm}
\KwData{Graph $G$ on $n$ vertices}
\KwResult{A potential $(2\epsilon)$-sparsifier $H$ of $G$.}
Run $\text{ApproximateKruskal}$ on a copy of $G$ and calculate $\tilde{d}_{i, j}$ for all $i, j$. \\
Initialize $L$ to a list of all nonzero values of $\tilde{d}_{i, j}$. \\
Initialize $\mathcal{C} = \emptyset$, $\kappa = \infty$, and $X$ as an empty list.\\ 
$\text{InitializeDS}(G)$. \\
\While{$L \neq \emptyset$}{
    Let $\tilde{D} = \max L$ and $C = \left\{(i, j): \tilde{d}_{i, j} < \tilde{D}/(2n^5)\right\}$. \\
    Let $S_1, \ldots, S_r$ be the connected components of $K_n$ after we remove all edges from $C$. \\ 
    Let $\kappa'$ be the smallest power of 2 that is at least $n^4 \tilde{D}/2$. \\
    $\kappa \gets \min(\kappa, \kappa')$ \\
    \While{$\kappa \geq \tilde{D}/(2n)$}{
        \For{$i \gets 1$ \KwTo $r$} {
            $\text{EstimateAndContract}(G[S_i], \mathcal{C}, \kappa, X)$ \\
            Remove all contracted edges from $L$.\\
        }
        $\kappa \gets \kappa/2$\\
    }
}
$H \gets \text{ConstructSparsifier}(G, X).$
 \caption{$\text{FastWeightedSubsample}(G)$}\label{algo:fws}
\end{algorithm}


We make one comment here about FastWeightedSubsample: for the algorithm as presented to even be well-defined, we need to check that each $S_i$ at any stage of the algorithm is the union of some collection of supernodes. If not, it does not make sense to run EstimateAndContract on each $G[S_i]$. This condition is also necessary to ensure the applicability of Lemma \ref{lemma:allalgoprimitives} to each call to EstimateAndContract; Contract, GetTotalWeight, and Sample all require that their input $S$ be a union of supernodes. We defer the verification of this and the proof of our final theorem to Appendices \ref{lemmaproofs:fwscorrect} (correctness) and \ref{sec:sparsifierspeed} (efficiency):

\begin{theorem}\label{thm:fwsworks}
FastWeightedSubsample runs in $O(n\log^3n(\log n + \frac{1}{\epsilon^2}))$ queries and outputs a $(2\epsilon)$-sparsifier $H$ with probability $1 - O(n^{5-d})$.
\end{theorem}

\begin{corollary}\label{cor:fastsparsifier}
For any $\epsilon > 0$, the query complexity for a randomized algorithm to construct an $\epsilon$-sparsifier with $1 - o(1)$ probability for a weighted undirected graph is $\widetilde{O}(n)$. This algorithm is adaptive.

\end{corollary}

By Lemma \ref{lemma:sparsifiersuffices}, this yields our desired algorithmic result for max-cut:

\begin{corollary}\label{cor:randhalfalgo}
For any $c < 1$, the query complexity for a randomized algorithm to achieve a $c$-approximation with $1 - o(1)$ probability for global max-cut on a weighted undirected graph in the cut finding setting is $\tilde{O}(n)$. This algorithm is adaptive.
\end{corollary}

We remind the reader that Theorem~\ref{thm:fwsworks} extends the query-efficient sparsifier of~\cite{rubinstein} to weighted graphs (as we saw in Section \ref{sec:rswadaptation}, a naive generalization of~\cite{rubinstein} requires $\widetilde{O}(n \log W)$ queries).

\paragraph{Acknowledgements.} SR would like to thank Georgy Noarov and Dmitry Paramonov for collaborating on a course project that led to this work. We would also like to thank Sepehr Assadi for helpful discussions and anonymous reviewers for feedback and suggested changes. SR was supported by an Akamai Presidential Fellowship. SMW was supported by NSF CCF-1955205.

\bibliographystyle{alpha}
\bibliography{main}

\appendix

\section{Lower Bound for Exact Deterministic Algorithms}\label{sec:detexacthardness} 

Here we address the deterministic hardness part of Theorem \ref{thm:informalexact}. The following theorem shows that $\Omega(n^2)$ queries are needed to deterministically find the value of the max cut. Note that this lower bound relies on the fact that the graph is weighted; \cite{gk00} show that $O(n^2/\log n)$ queries are necessary and sufficient to deterministically and non-adaptively learn an unweighted graph.

\begin{theorem}\label{thm:exactdethardness}
Suppose $n \geq 3$. Then a deterministic algorithm that finds the value of the max cut exactly requires at least $\binom{n}{2}$ queries when $n$ is odd, and at least $\binom{n-1}{2}$ queries when $n$ is even.

This implies that the query complexity for a deterministic algorithm to exactly solve global max-cut on a weighted undirected graph in the value estimation setting is $\Omega(n^2)$.
\end{theorem}

Our proof of this directly applies the linear algebraic \emph{cut dimension} technique developed by \cite{graur}. \cite{graur} use this technique to obtain query lower bounds for exact min-cut in the cut query model, but we will see that it works just as well for exact max-cut. We set up some notation. Let $E(G)$ be the set of edges in $G$ with positive weight, then we can stack the nonzero edge weights of $G$ into a vector $w \in \RR^{|E(G)|}$. We still index the coordinates in $\RR^{|E(G)|}$ by pairs of vertices $i, j$. Now any cut $S \subseteq [n]$, we have that:
\begin{align*}
    F(S) &= \sum_{i \in S} \sum_{j \in [n]\backslash S} w_{i, j} \\
    &= \sum_{i, j} (v_S)_{i, j} w_{i, j} \\
    &= w^Tv_S,
\end{align*}
where we define the indicator vector $v_S \in \RR^{|E(G)|}$ by:
\begin{align*}
    (v_S)_{i, j} &= \begin{cases}
       1,&  (i, j)\text{ crosses the cut defined by }S,\\
       0,& \text{ otherwise}.
    \end{cases}
\end{align*}
We call $v_S$ the \emph{query vector} corresponding to $S$. Note that we only include entries for vertex pairs in $G$ that share an edge of positive weight. With this language, we can intuitively summarize the cut dimension technique -- we refer the reader to Section 3 of \cite{graur} for a more comprehensive explanation. Before any queries have been made, the graph's weights can be described by a vector $w \in \RR^{|E(G)|}$ and thus exist in a subspace of dimension $|E(G)|$. Each query made imposes one linear constraint on $w$, so that if there are $q$ queries made, the algorithm only has the information that $w$ exists in a subspace of dimension at least $|E(G)| - q$. The idea is to show that this is still too many degrees of freedom for the algorithm to be able to identify the max cut value exactly. We now make this precise by defining the \emph{max cut dimension}:

\begin{definition}\label{defn:maxcut}
The \emph{max cut dimension} of a graph $G$ is equal to the dimension of the span of:
$$\left\{v_U: U \in \argmax_{S \text{ a cut of } G} F(S)\right\}.$$
Call this span the \emph{max cut space} of the graph.
\end{definition}

In other words, the max cut space is the span of the indicator vectors of all max cuts, and the max cut dimension is the dimension of this span. We now show that the max cut dimension of a graph provides a lower bound on the number of queries needed:

\begin{lemma}\label{lemma:gprw}
\cite{graur} Suppose there exists a graph $G$ on $n$ vertices with max cut dimension $d$. Then any deterministic algorithm that finds the value of the max cut exactly must use at least $d$ queries.

(Note importantly that even if all edges in $G$ have weight 0 or 1, this lemma still only yields a lower bound on deterministic max-cut for weighted graphs. This is because the proof involves making small non-integer perturbations to the edge weights of $G$.)
\end{lemma}
\begin{proof}

This exact result is not proven by \cite{graur} but they show the analogous result for min cuts and the proof is exactly the same. We present the proof below for completeness.


Suppose for the sake of contradiction that there is a deterministic algorithm that uses $q \leq d-1$ queries. Then let the algorithm run against an adversary that answers all queries as if the graph were the $G$ given in the statement of the lemma. Let $w \in \RR^{|E(G)|}$ be the weight vector corresponding to $G$, and $Q_1, Q_2, \ldots, Q_q$ be the queries made by the algorithm.

The idea is to find a nonzero ``perturbation direction" $z \in \RR^{|E(G)|}$ and small $\epsilon \neq 0$ such that $w, w+\epsilon z$ agree on all queries but have different max cut values, so that the algorithm cannot succeed. It suffices to show that such a perturbation exists.

To do this, let $X$ be the max cut space of $G$, so that we have $\dim X = d$. Also, let $Y$ be the span of $v_{Q_1}, \ldots, v_{Q_q}$ so that $\dim Y \leq q < d = \dim X$. It follows that $Y^\perp$ (the orthogonal complement of $Y$) and $X$ must have nontrivial intersection\footnote{To see this, note that $\dim Y^\perp = |E(G)| - \dim Y$ and we have $|E(G)| \geq \dim (X + Y^\perp) = \dim X + \dim Y^\perp - \dim(X \cap Y^\perp) = \dim X + |E(G)| - \dim Y - \dim(X \cap Y^\perp) \Leftrightarrow \dim(X \cap Y^\perp) \geq \dim X - \dim Y > 0$.}. Let us take $z$ to be any nonzero vector in $X \cap Y^\perp$, and $\epsilon$ positive but small enough that both $w - \epsilon z$ and $w + \epsilon z$ have only positive entries (we can do this since $w$ has only positive entries). Thus $w - \epsilon z$ and $w + \epsilon z$ will both correspond to valid graphs with non-negative weights.

Since $z \in X$ (which is the max cut space) and $z$ is nonzero, we have $z \notin X^\perp$ so that there exists some max cut $C$ such that $z^Tv_C \neq 0$. Assume WLOG that $z^Tv_C > 0$; if it is negative then we can just replace $z$ with $-z$. Then we claim that $w, w + \epsilon z$ agree on all queries but have differing max cut values. To see that they agree on all queries, note that $z \in Y^\perp \Rightarrow z^T v_{Q_i} = 0$ for all $i$. Then we also have for all $i$ that $(w+\epsilon z)^T v_{Q_i} = w^T v_{Q_i} + \epsilon z^T v_{Q_i} = w^T v_{Q_i}$, thus they indeed agree on the queries.

Finally, to see that the max cut values differ, note that:
\begin{align*}
    (w + \epsilon z)^T v_C &= w^T v_C + \epsilon z^T v_C \\
    &> w^T v_C.
\end{align*}
Now, the max cut value of the graph corresponding to $w+\epsilon z$ is at least the value of the cut defined by $C$, which is $(w + \epsilon z)^T v_C$. This is strictly greater than $w^T v_C$ which by definition of $C$ is exactly the max cut value of the graph corresponding to $w$. This establishes that the graphs have distinct max cut values and completes the proof of the lemma.

\end{proof}

We complete the proof of Theorem \ref{thm:exactdethardness} by using Lemma \ref{lemma:gprw} with the particular case of $G = K_n$. The following lemma shows that $K_n$ has max cut dimension corresponding to the lower bounds stated in the theorem:



\begin{lemma}\label{lemma:cutdimension}
Suppose $n \geq 3$, and let $G = K_n$ be the graph on $n$ vertices where all pairs of distinct vertices share an edge of weight 1. (Thus $|E(G)| = \binom{n}{2}$.) Then the max cut dimension of $G$ is $\binom{n}{2}$ if $n$ is odd, and $\binom{n-1}{2}$ if $n$ is even.
\end{lemma}

\begin{proof}
The idea is just that $K_n$ has exponentially many max cuts since any cut of size $\lfloor n/2 \rfloor$ or $\lceil n/2 \rceil$ is a max cut, so we can take appropriate linear combinations of their indicator vectors to obtain a large span. When $n$ is odd, we can obtain the indicator vector for any vertex pair and thus the entire space, giving $\binom{n}{2}$. When $n$ is even, we don't quite get the entire space since all vertices have the same degree of $n/2$ in a max cut and this imposes some linear constraints. Details are provided below.


Observe that any set $S \subseteq [n]$ has cut value $F(S) = |S|(n - |S|)$, so $S$ is a max cut if and only if $|S| \in \left\{\lfloor n/2 \rfloor, \lceil n/2 \rceil\right\}$. We need to show three things:

\begin{enumerate}
    \item When $n \geq 3$ is odd, the max cut space of $K_n$ is all of $\RR^{\binom{n}{2}}$ (and thus has dimension $\binom{n}{2}$).
    \item When $n \geq 4$ is even, the max cut space of $K_n$ has dimension at least $\binom{n-1}{2}$.
    \item When $n \geq 4$ is even, the max cut space of $K_n$ has dimension at most $\binom{n-1}{2}$.
\end{enumerate}

Note that the third claim is not necessary since we just need a lower bound on the max cut dimension, but we prove it for completeness. As we remark at the end, the proof of the third claim in Lemma \ref{lemma:cutdimevenupper} also tells us exactly which $\binom{n-1}{2}$-dimensional subspace gives the max cut space of $K_n$ when $n$ is even.

To set up our proofs, we introduce some helpful notation. For any disjoint sets $S, T \subseteq [n]$, let $E(S, T)$ be the set of vertex pairs/edges with one endpoint in $S$ and the other in $T$. Similarly, let $v_{S, T} \in \RR^{\binom{n}{2}}$ be the indicator vector defined as follows:
\begin{align*}
    (v_{S, T})_{i, j} &= \begin{cases}
       1,&  i \in S\text{ and }j \in T\text{ or vice versa},\\
       0,& \text{ otherwise}.
    \end{cases}
\end{align*}
So for example we have $v_{S, [n]\backslash S} = v_S$. When $S = \left\{a\right\}$ or $T = \left\{b\right\}$ are singletons, we also denote $v_{S, T}$ by $v_{a, b}$, $v_{S, U}$ by $v_{a, U}$ etc. Thus the $\binom{n}{2}$ vectors $v_{i, j}$ for all $i \neq j$ form the standard orthonormal basis for $\RR^{\binom{n}{2}}$.

Now we turn to proving these lemmas:

\begin{lemma}\label{lemma:cutdimodd}
For $n \geq 3$ odd, the max cut space of $K_n$ is all of $\RR^{\binom{n}{2}}$.
\end{lemma}
\begin{proof}
Let $X$ be the max cut space of the graph. We proceed in steps to show that $X$ contains $v_{a, b}$ for all pairs of distinct vertices $(a, b)$, which we do by just taking appropriate linear combinations of the indicator vectors of max cuts. This would imply the desired result since the $v_{a, b}$'s form a complete orthonormal basis.

\textit{Step 1.} $v_{a, b} - v_{a, c} \in X$ whenever $a \neq b$ and $a \neq c$.

If $b = c$ then we have $v_{a, b} - v_{a, c} = 0 \in X$ since $X$ is a subspace. Thus from now assume that $b \neq c$ also. Let $A, B$ be any partition of the vertices excluding $a, b, c$ into two sets of size $\frac{n-3}{2}$. (We can do this since $n \geq 3$; if $n = 3$ just let $A, B$ both be the empty set.) Let $S_1$ be the cut with $a, b, A$ on one side and $c, B$ on the other, and $S_2$ the cut with $b, A$ on one side and $a, c, B$ on the other. These are both max cuts, so $v_{S_1} - v_{S_2} = v_{a, c} + v_{a, B} - v_{a, b} - v_{a, A} \in X$.

But we can repeat this with $A$ and $B$ swapped to find that $v_{a, c} + v_{a, A} - v_{a, b} - v_{a, B} \in X$. Adding this with the vector from the above paragraph and dividing by 2 gives $v_{a, c} - v_{a, b} \in X \Rightarrow v_{a, b} - v_{a, c} \in X$ as desired.

\textit{Step 2.} $v_{a, b} - v_{c, d} \in X$ whenever $a \neq b$ and $c \neq d$.

If $\left\{a, b\right\}$ and $\left\{c, d\right\}$ intersect then this is a special case of Step 1, so assume they are disjoint. In this case, this follows since $v_{a, b} - v_{c, d} = [v_{a, b} - v_{a, c}] + [v_{a, c} - v_{c, d}] \in X$ (the term in each bracket is in $X$ by Step 1).

\textit{Step 3.} $v_{a, b} \in X$ whenever $a \neq b$.

Take an arbitrary max cut $S$. Then we have
\begin{align*}
    v_S = \sum_{(c, d) \in E(S, [n] \backslash S)} v_{c, d} \in X.
\end{align*}
We also have
\begin{align*}
    \sum_{(c, d) \in E(S, [n] \backslash S)} [v_{a, b} - v_{c, d}] \in X
\end{align*}
by Step 2. Adding these gives
\begin{align*}
    \sum_{(c, d) \in E(S, [n] \backslash S)} v_{a, b} = |S|(n - |S|) v_{a, b} \in X.
\end{align*}
Since $S$ is a max cut, we have $|S|(n - |S|) \neq 0$ so dividing by this scalar yields that $v_{a, b} \in X$ as desired.
\end{proof}

\begin{lemma}\label{lemma:cutdimevenlower}
For $n \geq 4$ even, the max cut space of $K_n$ has dimension at least $\binom{n-1}{2}$.
\end{lemma}
\begin{proof}
The idea is that we can roughly show that the max cut space of an induced $K_{n-1}$ within this graph is ``nested" within the max cut space of this graph. Since the former max cut space has dimension $\binom{n-1}{2}$, this would yield the lemma.

To make this precise, let $V \subseteq \RR^{\binom{n}{2}}$ be the subspace spanned by $v_{i, j}$ for all $i \neq j$ such that $i, j \in [n-1]$, and let $\mathcal{P}_V$ be the orthogonal projection onto $V$. Thus $V$ has dimension $\binom{n-1}{2}$ and $\mathcal{P}_V$ is the operator that zeroes out all coordinates corresponding to edges involving vertex $n$.

Let $X$ be the max cut space of $K_n$. Then it suffices to show that $\mathcal{P}_V(X) = V$. Then the conclusion would follow since $\dim X \geq \dim \mathcal{P}_V(X) = \dim V = \binom{n-1}{2}$.\footnote{To see why $\dim X \geq \dim \mathcal{P}_V(X)$, let $d = \dim X$ and take a basis $u_1, u_2, \ldots, u_d$ of $X$. Then $\mathcal{P}_V(X)$ is the span of $\mathcal{P}_V(u_1), \mathcal{P}_V(u_2), \ldots, \mathcal{P}_V(u_d)$ and thus has dimension at most $d$. This is an inequality since these $d$ vectors need not be linearly independent.}

To do this, let $T \subseteq [n-1]$ be any set of size $(n-2)/2$ and let $S = T \cup \left\{n\right\}$. Then $|S| = n/2$ so $S$ is a max cut, and we have:
\begin{align*}
    v_S &= \sum_{i \in S} \sum_{j \in [n] \backslash S} v_{i, j} \\
    &= \sum_{i \in T \cup \left\{n\right\}} \sum_{j \in [n-1] \backslash T} v_{i, j}.
\end{align*}
To find $\mathcal{P}_V(v_S)$, we just need to zero out all terms here that have one endpoint at $n$. This gives us:
\begin{align*}
    \mathcal{P}_V(v_S) &= \sum_{i \in T \cup \left\{n\right\}} \sum_{j \in [n-1] \backslash T} v_{i, j} \\
    &= \sum_{i \in T} \sum_{j \in [n-1] \backslash T} v_{i, j} \\
    &= v_{T, [n-1] \backslash T}.
\end{align*}
Since $X$ is the max cut space of $K_n$, we know that $v_S \in X \Rightarrow \mathcal{P}_V(v_S) \in \mathcal{P}_V(X)$. It follows that we have:
\begin{align*}
    \mathcal{P}_V(X) &\supseteq \spann(\left\{\mathcal{P}_V(v_S): |S| = n/2\right\}) \\
    &\supseteq \spann(\left\{v_{T, [n-1] \backslash T}: T \subseteq [n-1]\text{ and }|T| = (n-2)/2\right\}).
\end{align*}
Now note that this is exactly the max cut space of the induced subgraph of our $K_n$ on $[n-1]$, which is an instance of $K_{n-1}$. This is because any max cut of this induced subgraph will have one side with $(n-2)/2$ vertices and the other side with $n/2$ vertices, so we can take always take $T$ to be the side with $(n-2)/2$ vertices. Since $n-1$ is odd and $n-1 \geq 3$, Lemma \ref{lemma:cutdimodd} tells us that this span is in fact all of $V$. Thus $\mathcal{P}_V(X) \supseteq V$ and we have $\mathcal{P}_V(X) \subseteq V$ by definition of $\mathcal{P}_V$, so this proves that $\mathcal{P}_V(X) = V$ and completes the proof of the lemma.
\end{proof}

\begin{lemma}\label{lemma:cutdimevenupper}
For $n \geq 4$ even, the max cut space of $K_n$ has dimension at most $\binom{n-1}{2}$. Moreover, the max cut space of $K_n$ is exactly the space of all weight vectors $w \in \RR^{\binom{n}{2}}$ that have the same total weight at each vertex.
\end{lemma}
\begin{proof}
We first prove the stated upper bound on the dimension of the max cut space. The idea is that all vertices have degree $n/2$ in any max cut of $K_n$ for even $n$, and that this adds some linear constraints that restrict the dimensionality of the max cut space. This is because a max cut will split the vertices into two halves of size $n/2$ and each vertex is connected to all the vertices on the opposite half. Writing this in our linear algebraic notation, we have for any max cut $S$ and vertex $i \in [n]$ that:
\begin{align*}
    v_S^Tv_{i, [n] \backslash \left\{i\right\}} &= v_S^T(\sum_{j \neq i} v_{i, j}) \\
    &= \frac{n}{2}.
\end{align*}
Thus for any two distinct vertices $i, j$, we have:
\begin{align*}
    v_S^T(v_{i, [n] \backslash \left\{i\right\}} - v_{j, [n] \backslash \left\{j\right\}}) &= 0.
\end{align*}
Now let $X$ be the max cut space of $K_n$. For any max cut $S$, the indicator vector $v_S$ is orthogonal to $v_{i, [n] \backslash \left\{i\right\}} - v_{j, [n] \backslash \left\{j\right\}}$. Thus arbitrary linear combinations of the $v_S$'s are orthogonal to $v_{i, [n] \backslash \left\{i\right\}} - v_{j, [n] \backslash \left\{j\right\}}$, so that all vectors in $X$ are orthogonal to $v_{i, [n] \backslash \left\{i\right\}} - v_{j, [n] \backslash \left\{j\right\}}$. Thus let us define $Y$ to be the following subspace:
\begin{align*}
    Y &= \spann(\left\{v_{i, [n] \backslash \left\{i\right\}} - v_{j, [n] \backslash \left\{j\right\}}: i \neq j\right\}).
\end{align*}
Then what we have just argued is that $X \subseteq Y^\perp$. It follows that $\dim X \leq \dim Y^\perp = \binom{n}{2} - \dim Y$. Thus it suffices to show that $\dim Y \geq n-1$.

To do this, consider the following set of $n-1$ vectors $u_i$ for $i \in [n-1]$, where we define $u_i$ as follows:
\begin{align*}
    u_i &= v_{i, [n]\backslash\left\{i\right\}} - v_{n, [n-1]}.
\end{align*}
Clearly $u_i \in Y$ for all $i$ so it suffices to show that $u_1, u_2, \ldots, u_{n-1}$ are linearly independent. Indeed, suppose we have scalars $\lambda_i \in \RR$ such that:
\begin{align*}
    \sum_{i = 1}^{n-1} \lambda_i u_i &= 0 \\
    \Leftrightarrow \sum_{i = 1}^{n-1} \lambda_i(v_{i, [n]\backslash\left\{i\right\}} - v_{n, [n-1]}) &= 0.
\end{align*}
For any distinct $i, j \in [n-1]$, the coefficient of $v_{i, j}$ on the LHS must be 0 (equivalently, the inner product of the LHS with $v_{i, j}$ must be 0). This is equivalent to $\lambda_i + \lambda_j = 0$. Since this is true for all distinct $i, j \in [n-1]$ and $n \geq 4$, this forces $\lambda_i = 0$ for all $i$.\footnote{To see this, note that for any $i \in [n-1]$ we can find two more indices $j, k \in [n-1]$ such that $i, j, k$ are distinct (this is why we need $n \geq 4$) and then we have $\lambda_i + \lambda_j = \lambda_j + \lambda_k = \lambda_k + \lambda_i = 0 \Rightarrow \lambda_i = \lambda_j = \lambda_k = 0$.} This implies the stated linear independence and thus that $\dim X \leq \binom{n-1}{2}$. This addresses the first part of the lemma.

It follows by Lemma \ref{lemma:cutdimevenlower} that this argument must be tight. Thus we must have that $\dim Y = n-1$ and $X = Y^\perp$. Thus $X$ is exactly the orthogonal complement of $Y$. In other words, the max cut space is the space of all weight vectors in $\RR^{\binom{n}{2}}$ that have the same total weight at each vertex. This completes the proof of the lemma.
\end{proof}

We have now proven all three of our claims, thus the proof of Lemma \ref{lemma:cutdimension} is complete.

\end{proof}

\subsection{Finding an MSF Requires $\Omega(n^2)$ Queries}\label{sec:MSThardness}

Here, we similarly adapt the cut dimension technique to prove a lower bound for deterministically finding a maximum spanning forest (MSF) in the cut query model. (We are doing this to justify a remark made in Section \ref{sec:crudestrengthestimates}; this result does not imply anything about max-cut.)

Since we are proving a lower bound in this section, we can restrict attention to connected graphs and hence work with ``maximum spanning trees" (MSTs) instead of ``maximum spanning forests". We begin by adapting the notion of max cut dimension to spanning trees.

\begin{definition}\label{def:treevector}
Suppose $G$ is connected. For any tree $\mathcal{T} \subseteq G$, let the \emph{tree vector} $a(\mathcal{T}) \in \RR^{|E(G)|}$ be defined as follows:
\begin{align*}
    a(\mathcal{T})_{i, j} &= \begin{cases}
       1,&  (i, j)\text{ is included in }\mathcal{T},\\
       0,& \text{ otherwise}.
    \end{cases}
\end{align*}
\end{definition}
Note that once again, if we describe the graph's weights with a vector $w \in \RR^{|E(G)|}$, the total weight of a given tree $\mathcal{T}$ is equal to the inner product $w^T a(\mathcal{T})$.
\begin{definition}\label{def:maxtreedimension}
Suppose $G$ is connected. Then we define the \emph{max tree dimension} of $G$ to be equal to the dimension of the span of:
$$\left\{a(\mathcal{T}): \mathcal{T}\text{ is a max spanning tree of }G\right\}.$$
Call this span the \emph{max tree space} of the graph.
\end{definition}
The following lemma is the adaptation of the cut dimension lemma by \cite{graur} that we need:
\begin{lemma}\label{lemma:gprwtree}
Suppose there exists a connected graph $G$ on $n$ vertices with max tree dimension $d$. Then any deterministic algorithm that finds the total weight of the maximum spanning tree exactly must use at least $d$ queries.
\end{lemma}
\begin{proof}
This proceeds identically to the proof of Lemma \ref{lemma:gprw}. $X$ will now be the max tree space of $G$, and instead of a max cut $C$ we will now take an MST $\mathcal{T}$ such that $z^T a(\mathcal{T}) \neq 0$. At the final step, the graph $G'$ defined by $w+\epsilon z$ will agree with $G$ on all cut queries but the tree defined by $\mathcal{T}$ will now have value strictly greater than $w^T a(\mathcal{T})$, so the total weight of the max spanning tree in $G'$ will be strictly greater than the total weight of the max spanning tree in $G$.
\end{proof}
To finish, we will show that $K_n$ has max tree dimension $\Omega(n^2)$, which will in turn imply that $\Omega(n^2)$ queries are needed to deterministically find an MST:
\begin{lemma}
Suppose $n \geq 3$, and let $G = K_n$ be the graph on $n$ vertices where all pairs of distinct vertices share an edge of weight 1. Then the max tree dimension of $G$ is $\binom{n}{2}$.
\end{lemma}
\begin{proof}
Note that $|E(G)| = \binom{n}{2}$ so it just suffices to show that the max tree dimension of $G$ is at least $\binom{n}{2}$. Let $X$ be the max tree space of $K_n$. Note that any tree on $K_n$ has total weight $n-1$ and is hence a max spanning tree.

As in the proof of Lemma \ref{lemma:cutdimension}, for distinct vertices $a, b$ let $v_{a, b}$ be the indicator vector that is 1 for the edge between $a$ and $b$ and 0 everywhere else. We proceed in steps to show that $X$ contains $v_{a, b}$ for all pairs of distinct vertices $(a, b)$, which we do by just taking appropriate linear combinations of the indicator vectors of max cuts. This would imply the desired result since the $v_{a, b}$'s form a complete orthonormal basis.

\textit{Step 1.} $v_{a, b} - v_{a, c} \in X$ whenever $a \neq b$ and $a \neq c$.

If $b = c$ then we have $v_{a, b} - v_{a, c} = 0 \in X$ since $X$ is a subspace. Thus from now assume that $b \neq c$ also. Number the remaining vertices of $G$ by $1, 2, \ldots, n-3$. Consider the following two path graphs on $G$:
\begin{enumerate}
    \item $A$: passes through the vertices in order $1, 2, \ldots, n-3, a, b, c$ (i.e. its edges are $(i, i+1)$ for $i \in [n-4]$ and then $(n-3, a), (a, b), (b, c)$.
    \item $B$: passes through the vertices in order $1, 2, \ldots, n-3, a, c, b$.
\end{enumerate}
$A$ and $B$ are both MSTs on $G$ so we have $a(A), a(B) \in X$. Thus we must also have $a(A) - a(B) \in X \Rightarrow v_{a, b} - v_{a, c} \in X$ as desired.

\textit{Step 2.} $v_{a, b} - v_{c, d} \in X$ whenever $a \neq b$ and $c \neq d$.

If $\left\{a, b\right\}$ and $\left\{c, d\right\}$ intersect then this is a special case of Step 1, so assume they are disjoint. In this case, this follows since $v_{a, b} - v_{c, d} = [v_{a, b} - v_{a, c}] + [v_{a, c} - v_{c, d}] \in X$ (the term in each bracket is in $X$ by Step 1).

\textit{Step 3.} $v_{a, b} \in X$ whenever $a \neq b$.

Take an arbitrary MST $\mathcal{T}$. Then we have
$$a(\mathcal{T}) = \sum_{(c, d) \in \mathcal{T}} v_{c, d} \in X.$$
We also have by Step 2 that
$$\sum_{(c, d) \in \mathcal{T}} [v_{a, b} - v_{c, d}] \in X.$$
Adding these gives:
$$\sum_{(c, d) \in \mathcal{T}} v_{a, b} = (n-1)v_{a, b} \in X.$$
\end{proof}
We clearly have $n-1 \neq 0$ so dividing by this scalar yields that $v_{a, b} \in X$ as desired.

\section{Lower Bound for Randomized $(1/2 + \epsilon)$-approximation}\label{sec:randhalfhardness} 

We now prove a lower bound for randomized algorithms that achieve a $c$-approximation for $c > 1/2$. This differs from the deterministic lower bound presented in Section \ref{sec:dethalfhardness} in two ways: first, the query complexity shown is $\Omega(n/\log n)$ rather than $\Omega(n)$, and secondly, this lower bound argument only works for the cut finding setting.

\begin{theorem} \label{thm:randhalfhardness}
For $c > 1/2$, a randomized algorithm that finds a $c$-approximate max cut with constant probability $p > 0$ requires at least $\frac{n}{\log n}(\frac{p(2c-1)}{16c+72} - o(1))$ queries.

This implies that the query complexity for a randomized algorithm to achieve a $c$-approximation with $\Omega(1)$ probability for global max-cut on an \textbf{unweighted} (or weighted) undirected graph \textbf{in the cut finding setting} is $\Omega(n/\log n)$.
\end{theorem}

\begin{proof}
It suffices to show for any $\epsilon > 0$ satisfying $\epsilon^2 < \frac{2c-1}{8c + 36}$ that at least $\frac{n}{\log n}(\frac{p\epsilon^2}{2} - o(1))$ queries are necessary.

The following is our hard distribution: let $A$ be a uniformly random subset of $[n]$ then take the complete bipartite graph on $(A, [n] \backslash A)$ with all edge weights 1. The intuition is that in order to find a $c$-approximate max cut, the algorithm must learn $\Omega(n)$ bits of information about $A$. Since it only receives $O(\log n)$ bits of information per query, this will yield the desired lower bound. We make this rigorous with an information theoretic argument.

Denote by $E(A)$ the set of edges included in this graph, for any set $A$. By Yao's minimax principle, it suffices to show the lower bound for a deterministic algorithm against this distribution. Moreover, this lower bound will hold even if we restrict attention to unweighted graphs, since all graphs in our hard distribution are unweighted (all their edge weights are 0 or 1). Suppose the algorithm uses $q$ queries. The key lemma is as follows:

\begin{lemma}\label{lemma:concentration}
Let $B$ be any subset of $[n]$ and $A$ a uniformly random subset of $[n]$. Then we have:

$$\Pr[\frac{|E(A) \cap E(B)|}{|E(A)|} \geq c] \leq 4\exp(-\epsilon^2n)$$.

where the randomness is over $A$.
\end{lemma}

\begin{proof}
This essentially follows from the fact that $|A|$ concentrates around $n/2$ and $|A \cap B|$ concentrates around $|B|/2$ provided $B$ is reasonably large (we can assume $|B| \geq n/2$ w.l.o.g~since $E(B) = E(\overline{B})$). We provide details below.

We use the notation ``$(b \pm c)$" to denote some number in $[b-c, b+c]$. Also, since $E(B) = E(\overline{B})$, the problem is invariant under replacing $B$ with $\overline{B}$ so assume WLOG that $|B| \geq n/2$.

Note that $|A|$ is the sum of $n$ i.i.d. random 0/1 variables that are each 1 with probability 1/2. Hence by Hoeffding's inequality we have $\Pr[|A| \notin (\frac{1}{2} \pm \epsilon)n] \leq 2\exp(-2\epsilon^2 n)$. Also, every vertex in $B$ also belongs to $A$ with probability $\frac{1}{2}$ and this selection is independent. Thus we similarly have $\Pr[|A \cap B| \notin (\frac{1}{2} \pm \epsilon) |B|] \leq 2\exp(-2\epsilon^2|B|) \leq 2\exp(-\epsilon^2 n)$. Taking a union bound, we have:
\begin{align*}
\Pr[|A| \in (\frac{1}{2} \pm \epsilon)n \text{ and } |A \cap B| \in (\frac{1}{2} \pm \epsilon)|B|] &\geq 1 - 2\exp(-2\epsilon^2 n) - 2\exp(-\epsilon^2 n) \\
&\geq 1 - 4\exp(-\epsilon^2n).
\end{align*}
We next turn our attention to estimating $\frac{|E(A) \cap E(B)|}{|E(A)|}$ assuming the event inside the LHS is true. First, we have:
\begin{align*}
|E(A)| &= |A|(n - |A|) \\
&= \frac{1}{4}n^2 - (|A| - \frac{n}{2})^2 \\
&\geq (\frac{1}{4} - \epsilon^2) n^2.
\end{align*}
Now for the numerator, note that for $(i, j)$ to be in both $E(A)$ and $E(B)$ we must have either that one vertex is in $A \cap B$ and the other is in $\overline{A} \cap \overline{B}$, or that one vertex is in $A \cap \overline{B}$ and the other is in $\overline{A} \cap B$. Let $a = |A|$, $b = |B|$, and $t = |A \cap B|$. Then we have:
\begin{align*}
|A \cap \overline{B}| &= a - t \\
|\overline{A} \cap B| &= |B| - |A \cap B| \\
&= b - t \\
|\overline{A} \cap \overline{B}| &= |\overline{A}| - |\overline{A} \cap B| \\
&= (n - a) - (b - t) \\
&= n - a - b + t.
\end{align*}
Thus we have:
\begin{align*}
|E(A) \cap E(B)| &= (a-t)(b-t) + t(n-a-b+t) \\
&= (t^2 - t(a+b) + ab) + (t^2 - t(a+b) + tn) \\
&= 2t^2 + t(n - 2a - 2b) + ab \\
&= 2(t^2 + t(\frac{n}{2} - a - b) + \frac{ab}{2}) \\
&= 2((t + \frac{n}{4} - \frac{a+b}{2})^2 + \frac{ab}{2} - (\frac{n}{4} - \frac{a+b}{2})^2) \\
&= 2((t + \frac{n}{4} - \frac{a+b}{2})^2 + \frac{ab}{2} - \frac{n^2}{16} - \frac{(a+b)^2}{4} + \frac{n(a+b)}{4}) \\
&= 2((t + \frac{n}{4} - \frac{a+b}{2})^2 - \frac{n^2}{16} + \frac{a(n-a)}{4} + \frac{b(n-b)}{4}).
\end{align*}
We estimate each piece of this as follows:
\begin{align*}
\frac{a(n-a)}{4} &\leq \frac{(a + (n-a))^2}{16} \text{ (AM-GM/Cauchy-Schwarz)} \\
&= \frac{n^2}{16} \\
\frac{b(n-b)}{4} &\leq \frac{n^2}{16} \text{ (similarly)} \\
(t + \frac{n}{4} - \frac{a+b}{2})^2 &= |t + \frac{n}{4} - \frac{a+b}{2}|^2 \\
&\leq (|t - \frac{b}{2}| + \frac{1}{2} |\frac{n}{2} - a|)^2 \\
&\leq (\epsilon b + \frac{1}{2} \epsilon n)^2 \\
&\leq (\frac{3}{2}\epsilon n)^2 \\
&= \frac{9}{4} \epsilon^2 n^2.
\end{align*}
Combining these gives:
\begin{align*}
|E(A) \cap E(B)| &\leq 2(\frac{9}{4}\epsilon^2n^2 - \frac{n^2}{16} + \frac{n^2}{16} + \frac{n^2}{16}) \\
&= 2(\frac{9}{4}\epsilon^2n^2 + \frac{n^2}{16}) \\
&= (\frac{1 + 36\epsilon^2}{8})n^2.
\end{align*}
Putting this together with our estimate of $|E(A)|$ gives:
\begin{align*}
\frac{|E(A) \cap E(B)|}{|E(A)|} &\leq \frac{(\frac{1 + 36\epsilon^2}{8})n^2}{(\frac{1}{4} - \epsilon^2) n^2} \\
&= \frac{1 + 36\epsilon^2}{2 - 8\epsilon^2} \\
&< c,
\end{align*}
by our choice of $\epsilon$. Thus we have:
\begin{align*}
\Pr[\frac{|E(A) \cap E(B)|}{|E(A)|} < c] &\geq \Pr[|A| \in (\frac{1}{2} \pm \epsilon)n \text{ and } |A \cap B| \in (\frac{1}{2} \pm \epsilon)|B|] \\
&\geq 1 - 4\exp(-\epsilon^2n) \\
\Rightarrow \Pr[\frac{|E(A) \cap E(B)|}{|E(A)|} \geq c] &\leq 4\exp(-\epsilon^2n),
\end{align*}
which completes the proof of the lemma.
\end{proof}

\begin{corollary}\label{corollary:concentration}
The number of subsets $A$ of $[n]$ satisfying $|E(A) \cap E(B)| \geq c|E(A)|$ is at most $2^n \cdot 4\exp(-\epsilon^2n)$.
\end{corollary}

Given this lemma, we prove the theorem with an information theoretic argument. Specifically, it is clear that the random variable $A$ has entropy $H(A) = n$. Informally, we will show that:

\begin{enumerate}
\item each query only gives the algorithm a small amount of information, and
\item achieving a $c$-approximation with probability $p$ requires the algorithm to learn a linear amount of information,
\end{enumerate}

which will give us a query lower bound. We make this intuitive approach precise by defining the following random variables:

\begin{itemize}
\item $R$: the sequence of query results given to the algorithm
\item $C$: the cut output by the algorithm
\item $I$: an indicator that is 1 if 1 if $C$ achieves a $c$-approximation and 0 otherwise.
\end{itemize}

Step 1 is straightforward: each query's result is an integer in $[0, \binom{n}{2}]$ since all edges in our graph are 0 or 1 and hence has entropy at most $2\lg n$. Thus $H(R) \leq 2q\lg n$. Moreover since we are assuming the algorithm is deterministic, $C$ is a deterministic function of $R$ so we have $H(C) \leq H(R) \leq 2q\lg n$. This yields the following lower bound on $H(A \mid C)$:
\begin{align*}
H(A \mid C) &= H(A, C) - H(C) \\
&\geq H(A, C) - 2q\lg n \\
&= H(C \mid A) + H(A) - 2q\lg n \\
&\geq H(A) - 2q\lg n \\
&= n - 2q\lg n.
\end{align*}

It remains to formalize step 2 by upper bounding $H(A \mid C)$. We proceed very similarly to the proof of Fano's inequality \cite{coverthomas}; indeed, if the success condition for the algorithm was ``$C = A$" rather than ``$C$ approximately agrees with $A$", then Fano's inequality would be directly applicable. We apply a similar bounding argument as follows:
\begin{align*}
H(A \mid C) &= H(A \mid C) + H(I \mid A, C) \text{ ($I$ is fixed given $A$ and $C$)} \\
&= H(I, A \mid C) \\
&= H(A \mid I, C) + H(I \mid C) \\
&\leq H(A \mid I, C) + 1 \\
&= H(A \mid C, I = 1) \Pr[I = 1] + H(A \mid C, I = 0) \Pr[I = 0] + 1 \\
&\leq (H(A \mid C, I = 1) - n) \Pr[I = 1] + n + 1 \\
&\text{ ($A$ is supported on a set of size $2^n$ so its entropy conditioned on anything is $\leq n$)} \\
&\leq (H(A \mid C, I = 1) - n) p + n + 1 \text{ ($\Pr[I = 1] \geq p$, and $H(A \mid C, I = 1) \leq n$)} \\
&= H(A \mid C, I = 1)p + n(1-p) + 1.
\end{align*}
Now observe that $H(A \mid C, I = 1) = \EE_{B \sim \mu_{C \mid I = 1}} H(A \mid C = B, I = 1)$, where $\mu_{C \mid I = 1}$ denotes the marginal distribution of $C$ conditioned on $I = 1$. Temporarily fix $B$ and consider the distribution of $A$ conditioned on $C = B$ and $I = 1$. Since we are conditioning on $I = 1$, the value of the cut defined by $B$ must be at least $c|E(A)|$. But the value of this cut is exactly $|E(A) \cap E(B)|$.

Thus Corollary \ref{corollary:concentration} tells us that there are at most $2^n \cdot 4\exp(-\epsilon^2n)$ sets $A$ that can appear with positive probability in this distribution. Then it follows that:
\begin{align*}
H(A \mid C = B, I = 1) &\leq \lg(2^n \cdot 4\exp(-\epsilon^2n)) \\
&= n + 2 - \epsilon^2 (\lg e) n \\
&= (1 - \epsilon^2 (\lg e)) n + 2.
\end{align*}
Taking the expectation over $B$ again gives $H(A \mid C, I = 1) \leq (1 - \epsilon^2 (\lg e)) n + 2$. Plugging this in gives:
\begin{align*}
H(A \mid C) &\leq H(A \mid C, I = 1)p + n(1-p) + 1 \\
&\leq ((1 - \epsilon^2 (\lg e)) n + 2)p + n(1-p) + 1 \\
&= n(1-p + p(1 - \epsilon^2 \lg e)) + 2p + 1 \\
&= n(1 - p\epsilon^2 \lg e) + 2p + 1.
\end{align*}
Finally, putting together our two bounds on $H(A \mid C)$ gives:
\begin{align*}
n(1 - p\epsilon^2 \lg e) + 2p + 1 &\geq n - 2q\lg n \\
\Leftrightarrow q &\geq \frac{n \cdot p\epsilon^2 \lg e - 2p - 1}{2\lg n} \\
&= \frac{n}{\log n}(\frac{p\epsilon^2}{2} - o(1)),
\end{align*}
which proves the theorem.
\end{proof}

\section{Lower Bounds for $\epsilon$-approximation}\label{sec:dethardnesszero} 

Here we show the deterministic lower bound part of Theorem \ref{thm:informalzero}, and also show that at least $\approx \lg(1/p)-1$ queries are necessary for randomized $\epsilon$-approximation with probability $\geq 1-p$. Both these results are captured by the following theorem:

\begin{theorem}\label{thm:generalepshardness}
For any $c \in (0, 1/2)$, any randomized algorithm using at most $q$ queries will achieve a $c$-approximation (even in the value estimation setting) with probability at most $1 - \frac{n - 2^q}{2^{q+1}(n-1)}$.
\end{theorem}
\begin{proof}

We construct a hard distribution of instances. It comprises the following:

\begin{itemize}
\item With probability $1/2$: the empty graph (i.e. all edges have weight 0).
\item For each edge $(i, j)$ in the graph, with probability $\frac{1}{n(n-1)}$: the graph $G_{i, j}$ where all edges have weight 0 except for $(i, j)$ which has weight 1.
\end{itemize}

Note that we can alternately characterize this distribution as follows: we choose an edge $(i, j)$ uniformly at random. Then all edges other than $(i, j)$ are set to 0, and $(i, j)$ is set to 1 with probability $1/2$. By Yao's minimax principle, it suffices to show that a deterministic algorithm succeeds on a graph from this distribution with probability at most $1 - \frac{n - 2^q}{2^{q+1}(n-1)}$. The following is the key lemma:

\begin{lemma}\label{lemma:cauchyschwarz}
Let $Q_1, Q_2, \ldots, Q_q$ be any $q$ cuts. Then with probability at least $\frac{n - 2^q}{2^q(n-1)}$, none of these cuts will be crossed by the chosen edge $(i, j)$.
\end{lemma}
\begin{proof}
For every vertex $a \in [n]$, we can assign it a $q$-bit string as follows: the $b$th bit indicates which side of the cut $Q_b$ $a$ is placed on. We thus define a partition of the $n$ vertices into $2^q$ buckets. $(i, j)$ crosses some queried cut if and only if $i$ and $j$ have distinct bit-strings i.e. they belong to distinct buckets. Let $B = 2^q$ and let the $B$ buckets have sizes $n_1, n_2, \ldots, n_B$. Noting that $(i, j)$ is chosen uniformly at random, we have the following:
\begin{align*}
\Pr[(i, j)\text{ never crossed}] &= \frac{2}{n(n-1)} \sum_{b = 1}^B \binom{n_b}{2} \\
&= \frac{1}{n(n-1)} \sum_{b = 1}^B (n_b^2 - n_b) \\
&= \frac{1}{n(n-1)} ((\sum_{b = 1}^B n_b^2) - n) \\
&\geq \frac{1}{n(n-1)} (\frac{1}{B} (\sum_{b = 1}^B n_b)^2 - n) (\text{Cauchy-Schwarz}) \\
&= \frac{1}{n(n-1)} (\frac{n^2}{B} - n) \\
&= \frac{1}{n-1} (\frac{n}{B} - 1) \\
&= \frac{1}{n-1} (\frac{n}{2^q} - 1) \\
&= \frac{n - 2^q}{2^q(n-1)},
\end{align*}

which completes the proof of the lemma.
\end{proof}

Given this lemma, our argument is simple: let $Q_1, \ldots, Q_q$ be the cuts the algorithm will query if it always receives 0 as an answer from the oracle (this sequence is fixed and independent of the chosen edge $(i, j)$ since we are now considering a deterministic algorithm). By the lemma, we have that with probability at least $\frac{n - 2^q}{2^q(n-1)}$, none of these cuts will be crossed by the chosen edge $(i, j)$. If this is the case, the algorithm will always receive 0 as an answer and thus follow this exact sequence of queries. Moreover, at the end, the algorithm has absolutely no information on the weight of that edge so it cannot distinguish between the two different scenarios. In one case the max cut value is 0 and in the other case the max cut value is 1 so no answer can cover both cases, so the algorithm fails with probability at least $1/2$. Hence the algorithm's overall probability of failure is at least $\frac{n - 2^q}{2^{q+1}(n-1)}$ which proves the theorem.
\end{proof}

Our lower bounds arise as natural corollaries of this theorem:

\begin{corollary}\label{cor:detzerohardness}
A deterministic algorithm that achieves a $c$-approximation for any $c > 0$ requires at least $\lg n$ queries.

This implies that the query complexity for a deterministic algorithm to achieve a $c$-approximation for global max-cut on an \textbf{unweighted} (or weighted) undirected graph in the value estimation setting is $\Omega(\log n)$.
\end{corollary}
\begin{proof}

Suppose the algorithm uses $q$ queries. Then this algorithm must achieve a $c$-approximation with probability 1, so Theorem \ref{thm:generalepshardness} tells us that we must have:
\begin{align*}
    1 &\leq 1 - \frac{n - 2^q}{2^{q+1}(n-1)} \\
    \Leftrightarrow \frac{n - 2^q}{2^{q+1}(n-1)} &\leq 0 \\
    \Rightarrow q &\geq \lg n,
\end{align*}
as desired. This lower bound holds even if we restrict attention to unweighted graphs, since all graphs in the hard distribution used in the proof of Theorem \ref{thm:generalepshardness} are unweighted (all their edge weights are 0 or 1).
\end{proof}

\begin{corollary}\label{cor:randepshardness}
For $p > 0$, a randomized algorithm that achieves a $c$-approximation in the value-estimation setting with probability at least $1-p$ for any $c > 0$ requires at least $\lg(1/p) - 1 - o(1)$ queries.
\end{corollary}

\begin{proof}


Plugging in the bound from Theorem \ref{thm:generalepshardness} gives:
\begin{align*}
    1-p &\leq 1 - \frac{n - 2^q}{2^{q+1}(n-1)} \\
    \Leftrightarrow \frac{n - 2^q}{2^{q+1}(n-1)} &\leq p \\
    \Leftrightarrow n - 2^q &\leq 2^{q+1}pn - 2^{q+1}p \\
    \Leftrightarrow n &\leq 2^q(2pn - 2p + 1) \\
    \Leftrightarrow 2^q &\geq \frac{n}{2pn - 2p + 1} \\
    &= \frac{1}{2p - \frac{2p - 1}{n}} \\
    &\geq \frac{1}{2p} - o(1) \text{ (assuming $p > 0$)} \\
    \Rightarrow q &\geq \lg(1/p) - 1 - o(1),
\end{align*}
as desired.

\end{proof}

Note that this implies that the upper bound given for randomized algorithms in Corollary \ref{corollary:randalgorepetition} is tight in terms of its dependence on the failure probability $p$, although there is a gap in terms of the dependence on $c$.


\section{Randomized Algorithm for $(1-\epsilon)$-approximation in $\widetilde{O}(n)$ Queries without Sparsification}\label{sec:rswmaxcut}

In this section, we show another proof of Corollary \ref{cor:randhalfalgo} that does not require the optimizations presented in Section \ref{sec:fws} for sparsification. The starting point is NaiveWeightedSubsample and its analysis in Section \ref{sec:rswadaptation}. To remove the dependence of the query complexity on the weights of the graph, we show that we can stop NaiveWeightedSubsample before all edges have been assigned strengths. This may no longer give us a sparsifier, but it will still suffice for max-cut.

\begin{theorem}\label{thm:earlystopnwsworks}
Let $H = \text{NaiveWeightedSubsample}(G, n^3)$. Then $H$ can be calculated in $O(n\log^3 n(\log n + \frac{1}{\epsilon^2}))$ queries, and outputing the cut $U$ maximizing $F(U; H)$ will achieve a $(1-5\epsilon)$-approximation for max-cut with probability $1 - O(n^{1-d}\log n)$.

\end{theorem}
\begin{proof}
Once again, we provide an outline and defer details to Appendices \ref{sec:nwsearlystopworks} (correctness) and~\ref{sec:sparsifierspeed} (efficiency). To see why this eliminates the query complexity's dependence on $W$, note that $\kappa$ is now only allowed to vary between $W_\text{tot}$ and $W_\text{tot}/n^3$, so the outer loop can now only proceed for $O(\log n)$ iterations, as opposed to $O(\log n + \log W)$. This is also essentially why the success probability no longer depends on $W$.

For correctness, the intuition is that the output of $\text{NaiveWeightedSubsample}(G, n^3)$ is essentially the same as taking the sparsifier output by $\text{NaiveWeightedSubsample}(G, \infty)$ and discarding sufficiently weak edges. It can be shown that in the sparsifier, these weak edges have total weight $O(W_\text{tot}/n)$, which is much less than $W_\text{tot}/2$ and hence the value of the max cut. Hence discarding these weak edges will only shift the max cut by a factor of $1 - o(1)$, so it suffices to use $\text{NaiveWeightedSubsample}(G, n^3)$.


\end{proof}

\section{Other Algorithmic Results}\label{sec:easyalgoresults}

We noted in Section \ref{sec:results} that there exists a straightforward $O(n^2)$ deterministic algorithm to learn the graph and hence find the max cut, even in the cut finding setting. This immediately addresses quite a few of the algorithmic results stated there. We have addressed the randomized upper bound in Theorem \ref{thm:informalhalf} in Section \ref{sec:randhalfalgo}. We now finally address the randomized upper bound in Theorem \ref{thm:informalzero} in Section \ref{sec:randalgozero}, and the deterministic upper bound in Theorem \ref{thm:informalzero} in Section \ref{sec:detalgozero}.

\subsection{Warm-Up: Randomized Algorithm for $(1/2 - \epsilon)$-approximation}\label{sec:randalgozero}

Here we address the randomized upper bound part of Theorem \ref{thm:informalzero}. This is subsumed by the results in \cite{feige} as they show these results for arbitrary symmetric submodular functions, but we restate the main ideas here to help motivate our deterministic algorithms in Sections \ref{sec:detalgozero} and \ref{sec:detgreedyalgo}. The primary observation enabling this algorithm is simple and well-known:

\begin{lemma}\label{lemma:randomsample}
Suppose we sample a cut $U$ uniformly at random from subsets of $[n]$. Then we have $\EE[F(U)] = \frac{1}{2} \sum_{i < j} w_{i, j} \geq \frac{1}{2} \max_S F(S)$.
\end{lemma}
\begin{proof}
Observe that this sampling is equivalent to independently assigning each vertex to either be in $U$ or not be in $U$ with probability $1/2$. Thus for any pair of distinct vertices $i, j$, they will be on opposite sides of the cut with probability exactly $1/2$. The lemma follows.
\end{proof}

This lemma implies that randomly sampling a cut achieves a $(1/2)$-approximation in expectation, and this takes just 1 query. We can strengthen this to a {non-adaptive} algorithm that succeeds with arbitrarily high probability by repeating this procedure $O(1)$ times:

\begin{corollary}\label{corollary:randalgorepetition}
For $c < 1/2$ and a probability $p > 0$, there exists a {non-adaptive} randomized algorithm using $\log(1/p)/\log(2-2c)$ queries that finds a $c$-approximate max cut with probability at least $1-p$.

This implies that the query complexity for a randomized algorithm to achieve a $c$-approximation with $\Omega(1)$ probability for global max-cut on a weighted undirected graph in the cut finding setting is $O(1)$.
\end{corollary}
\begin{proof}


Consider one uniformly random cut $U$. Let $S^\ast$ be any exact max cut. Then by Lemma \ref{lemma:randomsample} we have:
\begin{align*}
    \EE[\frac{F(U)}{F(S^\ast)}] \geq \frac{1}{2} &\Rightarrow \EE[1 - \frac{F(U)}{F(S^\ast)}] \leq \frac{1}{2}.
\end{align*}
Thus we can use Markov's inequality to estimate the probability that $U$ achieves a $c$-approximation as follows:
\begin{align*}
    \Pr[\frac{F(U)}{F(S^\ast)} \leq c] &= \Pr[1 - \frac{F(U)}{F(S^\ast)} \geq 1 - c] \\
    &\leq \frac{1}{2(1-c)}.
\end{align*}
Thus if we let $q = \log(1/p)/\log(2-2c)$ and sample cuts $Q_1, Q_2, \ldots, Q_q$ independently and uniformly at random, the probability that none of them achieve a $c$-approximation is at most $(2(1-c))^{-q} = p$. Thus an algorithm can simply query each of $Q_1, Q_2, \ldots, Q_q$ and output whichever of these has the highest value.
\end{proof}

Thus a randomized algorithm can achieve a $c$-approximation in $O(1)$ queries by simply querying $O(1)$ random cuts and outputting the largest one. We will now see that a similar idea can be made to work for deterministic algorithms as well.

\subsection{Deterministic Algorithm for $(1/2 - \epsilon)$-approximation}\label{sec:detalgozero} 

Here we show the deterministic upper bound part of Theorem \ref{thm:informalzero}. As just mentioned, we would like to draw on the intuition from Section \ref{sec:randalgozero} that sampling random cuts suffices since these achieve a $(1/2)$-approximation in expectation. Of course, doing this directly faces the obstacle that we cannot randomly sample cuts. However, it turns out that it suffices to find a deterministic set of cuts that is ``sufficiently random". This is made precise by the following lemma:

\begin{lemma}\label{lemma:pseudorandom}
Let $q = \frac{4}{(1 - 2c)^2} \log n$. Then there exist cuts $Q_1, Q_2, \ldots, Q_q \subseteq [n]$ such that the following condition holds: every edge $(i, j)$ for distinct $i, j \in [n]$ appears in at least a $c$ fraction of these cuts.
\end{lemma}

\begin{proof}


Sample $q$ cuts $Q_1, Q_2, \ldots, Q_q$ as independent and uniformly random subsets of $[n]$. First fix an edge $(i, j)$. The probability that it appears in $Q_k$ for any fixed $k$ is $1/2$. So the fraction of the $q$ cuts that it appears in should concentrate around $q/2$ by independence. We make this precise using Hoeffding's inequality:
\begin{align*}
\Pr[(i, j)\text{ appears in fewer than }cq\text{ cuts}] &= \Pr[\frac{1}{q} \sum_{k = 1}^q \mathbf{1}[(i, j) \in Q_k] < c] \\
&\leq \exp(-2q(1/2 - c)^2).
\end{align*}
So taking a union bound over all $\binom{n}{2}$ edges, the probability that there exists an edge appearing in fewer than $cq$ cuts is at most:
\begin{align*}
\binom{n}{2} \exp(-2q(1/2 - c)^2) &< \frac{n^2}{2} \exp(-2q(1/2 - c)^2) \\
&= \frac{n^2}{2} \exp(-2 \frac{4}{(1 - 2c)^2} \log n (1/2 - c)^2) \\
&= \frac{1}{2},
\end{align*}
so it follows that there exists a collection of $q$ cuts such that every edge appears in at least $cq$ of them as desired.
\end{proof}

Now we will show that it essentially just suffices to query the $q$ cuts provided by Lemma \ref{lemma:pseudorandom} and output the largest one, mimicking the repetition procedure used in the proof of Corollary \ref{corollary:randalgorepetition}.

\begin{theorem}\label{thm:detlesshalfalgorithm}
For any $c < 1/2$, a deterministic algorithm can {non-adaptively} find a $c$-approximate max cut in $\frac{4}{(1 - 2c)^2} \log n$ queries.

This implies that the query complexity for a {non-adaptive} deterministic algorithm to achieve a $c$-approximation for global max-cut on a weighted undirected graph in the cut finding setting is $O(\log n)$.
\end{theorem}

\begin{proof} Since Lemma \ref{lemma:pseudorandom} is completely independent of the particular graph the algorithm is dealing with, the algorithm can as a first step find $q$ such cuts $Q_1, \ldots, Q_q$ without making any queries at all. (Even though we showed the existence of these cuts using the probabilistic method, such a set of cuts can be found deterministically e.g. with brute force search.) Once the algorithm has done this, it can query each of $Q_1, \ldots, Q_q$ and output whichever cut has the largest value. We verify that this algorithm works below.


Let $a_1, a_2, \ldots, a_q$ respectively be the values of each of the cuts $Q_1, Q_2, \ldots, Q_q$. By the condition of Lemma \ref{lemma:pseudorandom}, every edge appears in at least $cq$ of these cuts, thus we have:
$$\frac{1}{q} \sum_{k = 1}^q a_k \geq c \sum_{i \neq j} w_{(i, j)},$$
where $w_{(i, j)}$ denotes the weight of edge $(i, j)$. Now to finish note that the LHS is at most $\max_{k \in [q]} a_k$ which is the value of the cut the algorithm outputs, and the RHS is $c$ times the total weight of the graph which is at least $c$ times the max cut. Hence the algorithm's output is always at least $c$ times the max cut, completing the proof of the theorem.
\end{proof}

\subsection{Deterministic Algorithm for $(1/2)$-approximation}\label{sec:detgreedyalgo}

Here we show that $O(n)$ queries suffice to achieve a $(1/2)$-approximation. There is a well-known greedy algorithm for this in the classical model, but we reproduce it here for completeness and to highlight that it can be readily adapted to the cut query model.

\begin{lemma}\label{lemma:detgreedyhalf}
A deterministic algorithm can adaptively find a $(1/2)$-approximate max cut with $5n$ queries.

This implies that the query complexity for an adaptive deterministic algorithm to achieve a $(1/2)$-approximation for global max-cut on a weighted undirected graph in the cut finding setting is $O(n)$.
\end{lemma}
\begin{proof}
We use a greedy algorithm that in some sense derandomizes the naive random sampling we saw in Lemma \ref{lemma:randomsample}.


We first describe our greedy algorithm. Initialize $S, T$ to both be the empty set. Now for each vertex $i = 1, 2, \ldots, n$ in that order, find the total weight of edges between $i$ and $S$, and the total weight of edges between $i$ and $T$. If the former is larger, add $i$ to $T$. Else add $i$ to $S$. At the end, this yields a partition of $[n]$ into $S$ and $T$, and we claim that this corresponds to a $(1/2)$-approximate max cut.

This algorithm can be run in the cut query model since the total weight of edges between $i$ and $S$ is equal to $\frac{1}{2}(F(S) + F(\left\{i\right\}) - F(S \cup \left\{i\right\}))$, and similarly for $i$ and $T$. Thus each iteration of the algorithm can be carried out with queries to $\left\{i\right\}, S, T, S \cup \left\{i\right\}$, and $T \cup \left\{i\right\}$, yielding 5 queries per iteration for a total of $5n$ queries.

It now remains to show that this algorithm achieves a $(1/2)$-approximation. To see this, partition the edges into $n-1$ layers $U_2, U_3, \ldots, U_n$, where an edge between vertices $i, j$ is assigned to $U_{\max(i, j)}$. The reason for this layering is that $U_i$ contains all the edges that will be determined in step $i$ to either cross or not cross the output cut. For $i = 2, \ldots, n$, let $W_i$ be the total weight of all edges in $U_i$.

So now let us revisit step $i$ of the algorithm for any $i \geq 2$. Each edge in $U_i$ has one endpoint at $i$ and the other at $j$, where $j$ has already been placed in $S$ or $T$. Thus the edges between $i, S$ and the edges between $i, T$ partition the edges in $U_i$. Let the edges between $i, S$ have total weight $A$ and the edges between $i, T$ has total weight $B$. Then we have $A+B = W_i$, and our algorithm places $i$ in such a way that $U_i$ contributes $\max(A, B) \geq \frac{A+B}{2} = \frac{W_i}{2}$ to the eventual output cut.

Thus at step $i$, we add a total weight of at least $\frac{W_i}{2}$ to the value of the output cut, so the value of the output cut is at least $\frac{1}{2}(\sum_{i = 2}^n W_i) = \frac{1}{2} \sum_{i < j} w_{i, j}$. As in Lemma \ref{lemma:randomsample}, this is at least $\frac{1}{2}$ times the max cut. This completes our proof.

\end{proof}





\section{Lemmas from Section \ref{sec:dethalfhardness}}

\subsection{Proof of Corollary \ref{cor:dethalfhardness}}\label{lemmaproof:cordethalfhardness}

This is equivalent to showing that for any positive $\alpha < \frac{(\sqrt{c} - \sqrt{1-c})^4}{108c(2c-1)}$ and $n$ sufficiently large (where the threshold for being ``sufficiently large" might depend on $\alpha$), $\alpha n$ queries do not suffice to deterministically achieve a $c$-approximation. By Theorem \ref{thm:dethalfhardness}, it suffices to show that there exists $\epsilon \in (0, 1)$ such that $c > \frac{1}{1+\epsilon^2}$ and $\alpha < \frac{(1 - \epsilon)^3}{108(1+\epsilon)}$.

To do this, let $\mu = \sqrt{\frac{1-c}{c}} \in (0, 1)$, so that $c = \frac{1}{1+\mu^2}$. Then we have:
\begin{align*}
    \frac{(1 - \mu)^3}{108(1+\mu)} &= \frac{(1 - \sqrt{\frac{1-c}{c}})^3}{108(1+\sqrt{\frac{1-c}{c}})} \\
    &= \frac{(\sqrt{c} - \sqrt{1-c})^3}{108c(\sqrt{c} + \sqrt{1-c})} \\
    &= \frac{(\sqrt{c} - \sqrt{1-c})^4}{108c(2c-1)} \\
    &> \alpha.
\end{align*}
By continuity, there exists an open interval $I$ containing $\mu$ such that for all $t \in I$ we have $t \in (0, 1)$ and $\frac{(1 - t)^3}{108(1+t)} > \alpha$. Take $\epsilon \in I$ such that $\epsilon > \mu$ (this exists since $I$ is open and $\mu \in I$). Then we have $\epsilon \in (0, 1)$ and $\frac{(1 - \epsilon)^3}{108(1+\epsilon)} > \alpha$, and we have $\frac{1}{1+\epsilon^2} < \frac{1}{1+\mu^2} = c$, so this $\epsilon$ satisfies all conditions. This completes the proof of the corollary.

\subsection{Proof of Lemma \ref{lemma:cutduality}}\label{lemmaproof:cutduality}

It remains to show that it suffices to find a near-max cut $C$ and a perturbation $z$ such that in addition to the aforementioned constraints we have:
\begin{align}\label{eqn:perturbconstraint}
    z^T v_C \geq \frac{\epsilon^2 n^2}{4}.
\end{align}
To see this, let us analyse the max cut values of the graphs corresponding to $\one$ and $\one + z$. The graph corresponding to $\one$ is $K_n$ and has max cut value $\lfloor \frac{1}{4}n^2 \rfloor \leq \frac{1}{4}n^2$. On the other hand, the graph corresponding to $\one + z$ has max cut value greater than or equal to the value of the cut $C$, which is
\begin{align*}
    (\one + z)^T v_C &= \one^T v_C + z^T v_C \\
    &= |C|(n - |C|) + z^T v_C \\
    &\geq (\frac{1}{4}n^2 - n \log n) + \frac{\epsilon^2 n^2}{4} \\
    &\sim (\frac{1+\epsilon^2}{4}) n^2.
\end{align*}
It follows that the approximation ratio achieved by the algorithm is at most
\begin{align*}
    \frac{\frac{1}{4}n^2}{(\frac{1+\epsilon^2}{4} - o(1))n^2} &= \frac{1}{1+\epsilon^2 - o(1)} \\
    &= \frac{1}{1+\epsilon^2} + o(1) \\
    &< c,
\end{align*}
for $n$ sufficiently large, as desired.

\subsection{Proof of Lemma \ref{lemma:step2}}\label{lemmaproof:step2}

First, suppose that LP1 has value $\geq \epsilon^2n^2/4$ and let $z$ be a feasible point attaining the optimum. Then for the matrix LP, define $Z$ as follows:
\begin{align*}
    (Z)_{i, j} &= \begin{cases}
       z_{i, j},&  i \neq j,\\
       0,& \text{ otherwise}.
    \end{cases}
\end{align*}
We claim that $Z$ satisfies is feasible for the matrix LP and attains value $\geq \epsilon^2n^2/2$. It is clear that $Z \geq -\one$ since any nonzero entry of $Z$ is also an entry of $z$. Now, we have for any vertex set $S$ that:
\begin{align*}
    \langle Z, M_{S} \rangle &= \sum_{i_1, i_2} (Z)_{i_1, i_2} (M_{S})_{i_1, i_2} \\
    &= \sum_{i_1 \neq i_2} (z)_{i_1, i_2} (v_{S})_{i_1, i_2} \text{ (since $M_{S}$'s diagonal entries are 0)} \\
    &= \sum_{i_1 < i_2} 2 (z)_{i_1, i_2} (v_{S})_{i_1, i_2} \\
    &= 2 z^T v_{S}.
\end{align*}
Thus we have $\langle Z, M_{Q_j} \rangle = 2z^T v_{Q_j} = 0$ for all $j$, so $Z$ satisfies the orthogonality constraints. Moreover, we have $\langle Z, M_C \rangle = 2z^Tv_C \geq \epsilon^2n^2/2$, thus the matrix LP has value $\geq \epsilon^2n^2/2$.

Conversely (note that this is the direction that is necessary for the proof of Theorem \ref{thm:dethalfhardness}), suppose that the matrix LP has value $\geq \epsilon^2n^2/2$ and let $Z$ be a feasible point attaining the optimum. Then define $z$ by $z_{i, j} = \frac{1}{2}(Z_{i, j} + Z_{j, i})$. First for feasibility: note that $z_{i, j} \geq \frac{1}{2}((-1) + (-1)) = -1$. Now, we have for any vertex set $S$ that:
\begin{align*}
    z^T v_{S} &= \sum_{i_1 < i_2} (z)_{i_1, i_2} (v_{S})_{i_1, i_2} \\
    &= \frac{1}{2} \sum_{i_1 < i_2} ((Z)_{i_1, i_2} + (Z)_{i_2, i_1}) (M_{S})_{i_1, i_2} \\
    &= \frac{1}{2} \sum_{i_1 < i_2} ((Z)_{i_1, i_2} (M_{S})_{i_1, i_2} + (Z)_{i_2, i_1} (M_{S})_{i_2, i_1}) \text{ (since $M_{S}$ is symmetric)} \\
    &= \frac{1}{2} \sum_{i_1 \neq i_2} (Z)_{i_1, i_2} (M_{S})_{i_1, i_2} \\
    &= \frac{1}{2} \sum_{i_1, i_2} (Z)_{i_1, i_2} (M_{S})_{i_1, i_2} \text{ (since $M_{S}$'s diagonal entries are 0)} \\
    &= \frac{1}{2} \langle Z, M_{S} \rangle.
\end{align*}
Thus we have $z^T v_{Q_j} = \frac{1}{2} \langle Z, M_{Q_j} \rangle = 0$ for all $j$, so $z$ satisfies the orthogonality constraints. Finally, we similarly have that $z^T v_C = \frac{1}{2} \langle Z, M_C \rangle \geq \epsilon^2n^2/4$, thus LP1 has value $\geq \epsilon^2n^2/4$. This completes the proof of the lemma.

\subsection{Proof of Lemma \ref{lemma:step3}}\label{lemmaproof:step3}

First we address the feasibility of $Z = -yy^T$. For all $i, j$ we have $|Z|_{i, j} = |y_i| \cdot |y_j| \leq 1$ so in particular we have $Z \geq -\one$. Now for the remaining parts, observe that for any vectors $u, v \in \RR^n$ we have $\langle uu^T, vv^T \rangle = \tr(uu^Tvv^T) = \tr(u^Tvv^Tu) = (u^Tv)^2$. Thus for any vertex set $S$ we have:
\begin{align*}
    \langle Z, M_S \rangle &= \langle -yy^T, \frac{\one\one^T - u_Su_S^T}{2} \rangle \\
    &= -\frac{1}{2} \langle yy^T, \one\one^T \rangle + \frac{1}{2} \langle yy^T, u_Su_S^T \rangle \\
    &= -\frac{1}{2} (y^T\one)^2 + \frac{1}{2} (y^Tu_S)^2 \\
    &= \frac{1}{2} (y^Tu_S)^2
\end{align*}
Thus the orthogonality requirement follows, since $\langle Z, M_{Q_j} \rangle = \frac{1}{2} (y^T u_{Q_j})^2 = 0$. Thus $Z$ is feasible, and its value in the matrix LP is $\langle Z, M_C \rangle = \frac{1}{2} (y^T u_C)^2 \geq \epsilon^2n^2/2$. This completes the proof.

\subsection{Proof of Lemma \ref{lemma:duality}}\label{lemmaproof:duality}

We go ahead and take the dual of the stated LP. For the constraint $-\one \leq z$ take the dual vector $l \geq 0$, for the constraint $z^Tw_i = 0$ take the dual variable $\lambda_i \in \RR$, and for the constraint $z \leq \one$ take the dual vector $r \geq 0$. The constraints thus obtained are:
\begin{align*}
l^T(-\one - z) + r^T(z - \one) + z^T\sum_{i = 1}^k \lambda_i w_i &\leq 0 \\
\Leftrightarrow z^T(\sum_{i = 1}^k \lambda_i w_i + r - l) &\leq \one^T(r+l).
\end{align*}
So the dual LP is to minimize $\one^T(r + l)$ subject to $r, l \geq 0$ and $\sum_{i = 1}^k \lambda_i w_i + r - l = w$.

Before continuing, we make a brief comment on notation: to avoid confusion with the vectors $w_1, \ldots, w_k$, we label components of vectors differently here. Namely, for $j \in [d]$ and a vector $u \in \RR^d$, we let $u(j)$ denote the $j$th entry of $u$.

Now let $v = \sum_{i = 1}^k \lambda_i w_i$. $v$ ranges over $W$ as the $\lambda_i$'s range over $\RR$. Now temporarily fix the $\lambda_i$'s and hence $v$. Then we require $l - r = v - w$ and we want to minimize $\one^T(r+l)$. This is an independent problem for each coordinate: in entries $j$ where $v(j) \geq w(j)$ it is optimal to take $r(j)= 0$ and $l(j) = v(j) - w(j)$, and in entries where $v(j) < w(j)$ it is optimal to take $l(j) = 0$ and $r(j) = w(j) - v(j)$. Either way we have $l(j) + r(j) = |v(j) - w(j)|$, so we will have $\one^T(r+l) = \sum_j |v(j) - w(j)| = ||v - w||_1$. Thus the dual LP is equivalent to minimizing $||v-w||_1$ subject to $v \in W$.

We now finish using the strong LP duality theorem \cite{guptalp}. The dual LP is clearly bounded since its objective function is at least 0, and moreover it is feasible since $v = 0 \in W$ is a feasible point for example. Thus the dual LP has finite minimum which is $\min_{v \in W} ||v - w||_1$, so strong duality tells us that the primal LP also has the same value. This completes the proof of the lemma.

\subsection{Proof of Lemma \ref{lemma:ell1}}\label{lemmaproof:ell1}

We apply a volume argument. Specifically, we estimate the size of $(D + B_{\epsilon n}) \cap \left\{-1, 1\right\}^n$ using an $\ell_1$ $\epsilon$-net argument and show that this must be much less than $2^n$. Since almost all $p \in \left\{-1, 1\right\}^n$ satisfy $|\one^Tp| \leq 2\sqrt{n \log n}$ by concentration, some such $p$ will not be covered by $D + B_{\epsilon n}$, as desired.

We now make this precise. Let $\lambda = \frac{1 - \epsilon}{3} > 0$ and let $r = \lambda n$. Observe that $\epsilon + 2\lambda < 1$. We will prove the lemma with an $\epsilon$-net argument. Let $H' = \left\{p \in \left\{-1, 1\right\}^n: |\one^Tp| \leq 2\sqrt{n \log n}\right\}$ and $H = [-1, 1]^n$. Note that clearly $H' \subseteq H$. Also $D$ being a linear subspace is convex and $B_{\epsilon n}$ being an $\ell_1$ ball is convex.

Suppose for the sake of contradiction that $H' \subseteq D + B_{\epsilon n}$. Let $S = (H + B_{\epsilon n}) \cap D$. Note that $H' \subseteq S + B_{\epsilon n}$. This is because any $h \in H'$ can be expressed as $h = d + b$ where $d \in D$ and $b \in B_{\epsilon n}$. Then $d = h - b = h + (-b) \in H' + B_{\epsilon n} \subseteq H + B_{\epsilon n}$. But also we have $d \in D$ so taking the intersection gives $d \in S$. Thus $h = d + b \in S + B_{\epsilon n}$. This is true for all $h \in H'$ so we have $H' \subseteq S + B_{\epsilon n}$ as claimed. The point of introducing $S$ is that it is a finite piece of $D$ so now we can bound its volume (within $D$) from above. Note also that $H, B_{\epsilon n}$ are convex so $H + B_{\epsilon n}$ is convex, and since $D$ is also convex we have that $S$ is convex.

We bound $S$ above with a simpler set in order to estimate its $d$-dimensional volume. For any point $x \in S$, we can write it as $h + b$ with $h \in H$ and $b \in B_{\epsilon n}$. So we have $||x||_1 \leq ||h||_1 + ||b||_1 \leq n + \epsilon n = (1 + \epsilon) n \Rightarrow x \in B_{(1 + \epsilon) n}$. We also have $S \subseteq D$ so we have $x \in D$. This implies that $S \subseteq B_{(1+\epsilon)n} \cap D$. Call this larger set $T$. We hence have $H' \subseteq S + B_{\epsilon n} \subseteq T + B_{\epsilon n}$.

Now we apply the $\epsilon$-net argument. Let $x_1, x_2, \ldots, x_k$ be a maximal subset of $T$ such that $||x_i - x_j||_1 \geq 2r$ for all $i \neq j$. We first show that $k$ cannot be too large, and then use this to upper bound $|(D + B_{\epsilon n}) \cap \left\{-1, 1\right\}^n|$. For the upper bound on $k$, we have the following:

\begin{lemma}
We have $k \leq \exp(\alpha' \log(2(1+\epsilon)/\lambda)n)$.
\end{lemma}

\begin{proof}
Note that the condition means that the $x_i + B_r = x_i + B_{\lambda n}$ are pairwise almost disjoint. Hence we have:
\begin{align*}
    |T| &\geq \sum_{i = 1}^k |(x_i + B_{\lambda n}) \cap T|,
\end{align*}
where $|\cdot|$ denotes volume in $D$. We next show for each $i$ that there exists $y_i \in D$ such that $(x_i + B_{\lambda n}) \cap T \supseteq (y_i + B_{\lambda n/2}) \cap D$. Note that the LHS is $(x_i + B_{\lambda n}) \cap B_{(1+\epsilon)n} \cap D$ so it suffices to show that $(x_i + B_{\lambda n}) \cap B_{(1+\epsilon)n} \supseteq y_i + B_{\lambda n/2}$. This is intuitively clear but below is a formal proof:

Define the point $y_i$ as follows: if $||x_i||_1 \leq (1 + \epsilon - \frac{\lambda}{2})n$ then let $y_i = x_i$. Otherwise let $y_i = (1 + \epsilon - \frac{\lambda}{2})n \frac{x_i}{||x_i||_1}$. In other words, $y_i$ is the same as $x_i$ but scaled down if necessary to have $\ell_1$ norm at most $(1 + \epsilon - \frac{\lambda}{2})n$. Note that this is all well-defined since $1 + \epsilon \geq 1 \geq 2\lambda > \frac{\lambda}{2}$. We check the three conditions we need $y_i$ to satisfy:

\begin{itemize}
\item $y_i$ is a scalar multiple of $x_i$, so since $x_i \in T \subseteq D$ and $D$ is a linear subspace we also have $y_i \in D$.
\item $y_i + B_{\lambda n/2} \subseteq B_{(1+\epsilon)n}$ is immediate from the triangle inequality since $||y_i||_1 \leq (1 + \epsilon - \frac{\lambda}{2})n$.
\item $y_i + B_{\lambda n/2} \subseteq x_i + B_{\lambda n}$ would also follow from the triangle inequality if we can show that $||y_i - x_i||_1 \leq \lambda n/2$. This is immediate in the case that $y_i = x_i$, and in the other case we have $||y_i - x_i||_1 = ||x_i||_1 - ||y_i||_1 \leq (1+\epsilon)n - (1 + \epsilon - \frac{\lambda}{2})n = \lambda n/2$ as desired.
\end{itemize}

Thus our claim is proven. The next observation is that since $y_i \in D$ and $D$ is a linear subspace, we have $(y_i + B_{\lambda n/2}) \cap D = y_i + (B_{\lambda n/2} \cap D)$. Finally, note that for any radius $t$, again by linearity of $D$ we have $B_t \cap D = t(B_1 \cap D) \Rightarrow |B_t \cap D| = t^d |B_1 \cap D|$. Putting this all together gives:
\begin{align*}
((1+\epsilon)n)^d |B_1 \cap D| &= |B_{(1+\epsilon)n} \cap D| \\
&= |T| \\
&\geq \sum_{i = 1}^k |(x_i + B_{\lambda n}) \cap T| \\
&\geq \sum_{i = 1}^k |(y_i + B_{\lambda n/2}) \cap D| \\
&= \sum_{i = 1}^k |y_i + (B_{\lambda n/2} \cap D)| \\
&= k |B_{\lambda n/2} \cap D| \\
&= k (\lambda n/2)^d |B_1 \cap D| \\
\Leftrightarrow k &\leq (\frac{2(1+\epsilon)}{\lambda})^d \\
&= \exp(d\log(2(1+\epsilon)/\lambda)) \\
&\leq \exp(\alpha'\log(2(1+\epsilon)/\lambda) n),
\end{align*}
which proves the lemma.
\end{proof}

Now we complete our $\epsilon$-net argument using this lemma. Note that by maximality, for any $x \in T$ there exists $i \in [k]$ such that $||x - x_i||_1 \leq 2r \Leftrightarrow x \in x_i + B_{2r}$. Hence $T \subseteq \bigcup_{i = 1}^k (x_i + B_{2r})$. From here, $| \cdot |$ denotes cardinality rather than any kind of volume since it will only be applied to finite and discrete sets. Thus we have:
\begin{align*}
H' &\subseteq T + B_{\epsilon n} \\
&\subseteq (\bigcup_{i = 1}^k (x_i + B_{2r})) + B_{\epsilon n} \\
&\subseteq \bigcup_{i = 1}^k (x_i + B_{2r + \epsilon n}) \\
&= \bigcup_{i = 1}^k (x_i + B_{(\epsilon + 2\lambda) n}) \\
\Rightarrow H' &\subseteq \bigcup_{i = 1}^k ((x_i + B_{(\epsilon + 2\lambda) n}) \cap H') \\
\Rightarrow |H'| &\leq \sum_{i = 1}^k |(x_i + B_{(\epsilon + 2\lambda) n}) \cap H'| \\
\Leftrightarrow \frac{|H'|}{2^n} &\leq \sum_{i = 1}^k \frac{|(x_i + B_{(\epsilon + 2\lambda) n}) \cap H'|}{2^n} \\
&\leq \sum_{i = 1}^k \frac{|(x_i + B_{(\epsilon + 2\lambda) n}) \cap \left\{-1, 1\right\}^n|}{2^n}. \numberthis\label{eqn:probabilities}
\end{align*}
The final step is to estimate both sides of \eqref{eqn:probabilities} to get a contradiction. First we address the LHS:

\begin{lemma}\label{lemma:LHSbound}
We have $\frac{|H'|}{2^n} \geq 1 - 2/n^2$.
\end{lemma}
\begin{proof}
$\frac{|H'|}{2^n}$ is exactly the probability that for $p$ chosen uniformly at random in $\left\{-1, 1\right\}^n$, we have $|\one^Tp| \leq 2\sqrt{n \log n}$. By Hoeffding's inequality this is:

\begin{align*}
1 - \Pr[|\one^Tp| > 2\sqrt{n\log n}] &= 1 - \Pr[|\sum_{j = 1}^n p(j)| > 2\sqrt{n\log n}] \\
&= 1 - \Pr[|\frac{1}{n}\sum_{j = 1}^n p(j)| > 2\sqrt{\frac{\log n}{n}}] \\
&\geq 1 - 2\exp(-2n \cdot (4\log n/n)/4) \\
&= 1 - 2\exp(-2\log n) \\
&= 1 - 2/n^2.
\end{align*}
\end{proof}
Next we address the RHS of \eqref{eqn:probabilities}:

\begin{lemma}\label{lemma:RHSbound}
For any $x \in \RR^n$, $\frac{|(x + B_{(\epsilon + 2\lambda) n}) \cap \left\{-1, 1\right\}^n|}{2^n} \leq \exp(-n(1-\epsilon-2\lambda)^2/2)$.
\end{lemma}
\begin{proof}
We make a brief comment on notation first: to avoid confusion with the vectors $x_1, \ldots, x_k$, we label components of vectors differently here. Namely, for $j \in [n]$ and a vector $u \in \RR^n$, we let $u(j)$ denote the $j$th entry of $u$.

Now we turn to the lemma. $\frac{|(x + B_{(\epsilon + 2\lambda) n}) \cap \left\{-1, 1\right\}^n|}{2^n}$ is the probability that for $p$ chosen uniformly at random in $\left\{-1, 1\right\}^n$, we have $||p - x||_1 \leq (\epsilon + 2\lambda)n$.

For each $j$, let $D_j$ be the distribution of $|p(j) - x(j)|$. The $D_j$'s are all independent. Also define the distribution $D'_j$ as follows: if $x(j) \in [-1, 1]$ let $D'_j = D_j$, else let $D'_j$ be the uniform distribution on $\left\{0, 2\right\}$. We claim that $D'_j$ satisfies two properties:

\begin{enumerate}
\item $D_j$ and $D'_j$ can be coupled such that a sample from $D_j$ is greater than or equal to its coupled sample from $D'_j$.
\item $D'_j$ is bounded in $[0, 2]$ and has expectation 1.
\end{enumerate}

First we check property 1. This is straightforward when $x(j) \in [-1, 1]$ since then samples can just be coupled to themselves. Next if $x(j) > 1$ or $x(j) < -1$ - assume WLOG that $x(j) > 1$ by symmetry about 0 - $D_j$ is the uniform distribution on $\left\{x(j) - 1, x(j) + 1\right\}$. We can couple $x(j) + 1$ from $D_j$ with $2$ from $D'_j$ and $x(j) - 1$ from $D_j$ with $0$ from $D'_j$ so this works.

Next we check property 2. This is clear when $|x(j)| > 1$. In the other case, $D'_j$ is the uniform distribution on $\left\{1 - x(j), 1 + x(j)\right\}$. Indeed we have $0 \leq 1 - |x(j)| \leq |1 \pm x(j)| \leq 1 + |x(j)| \leq 2$. Clearly $D'_j$ has expectation 1. This establishes our claim.

Thus we can proceed as follows, once again applying Hoeffding's inequality:
\begin{align*}
\Pr[||p - x||_1 \leq (\epsilon + 2\lambda)n] &= \Pr[\frac{1}{n} \sum_{j = 1}^n |p(j) - x(j)| \leq \epsilon + 2\lambda] \\
&= \Pr_{a_j \sim D_j\text{, indep.}} [\frac{1}{n} \sum_{j = 1}^n a_j \leq \epsilon + 2\lambda] \\
&\leq \Pr_{a'_j \sim D'_j\text{, indep.}} [\frac{1}{n} \sum_{j = 1}^n a'_j \leq \epsilon + 2\lambda] \\
&\leq \exp(-2n \cdot (1 - \epsilon - 2\lambda)^2/4) \\
&= \exp(-n (1 - \epsilon - 2\lambda)^2/2),
\end{align*}
which yields the lemma.
\end{proof}

With the estimates from Lemmas \ref{lemma:LHSbound} and \ref{lemma:RHSbound}, we can finish as follows:
\begin{align*}
1 - 2/n^2 &\leq \frac{|H'|}{2^n} \\
&\leq \sum_{i = 1}^k \frac{|(x_i + B_{(\epsilon + 2\lambda) n}) \cap \left\{-1, 1\right\}^n|}{2^n} \\
&\leq \sum_{i = 1}^k \exp(-n (1 - \epsilon - 2\lambda)^2/2) \\
&= k \exp(-n (1 - \epsilon - 2\lambda)^2/2) \\
&\leq \exp(\alpha'\log(2(1+\epsilon)/\lambda) n -n (1 - \epsilon - 2\lambda)^2/2) \\
&= \exp((\alpha'\log(\frac{2(1+\epsilon)}{\lambda}) - \frac{(1 - \epsilon - 2\lambda)^2}{2})n).
\end{align*}
For this to hold for any $n$ sufficiently large, we must have $\alpha'\log(\frac{2(1+\epsilon)}{\lambda}) - \frac{(1 - \epsilon - 2\lambda)^2}{2} \geq 0$ so it suffices to show that this isn't the case. We do this as follows:
\begin{align*}
\alpha'\log(\frac{2(1+\epsilon)}{\lambda}) - \frac{(1 - \epsilon - 2\lambda)^2}{2} &\leq \alpha' (\frac{2(1+\epsilon)}{\lambda}) - \frac{(1 - \epsilon - 2\lambda)^2}{2} \\
&= \alpha' (\frac{6(1+\epsilon)}{1-\epsilon}) - \frac{(1 - \epsilon)^2}{18} \\
&< 0,
\end{align*}
by definition of $\alpha'$. This contradiction completes the proof of Lemma \ref{lemma:ell1}.

\subsection{Proof of Corollary \ref{cor:endofstep4}}\label{lemmaproof:corendofstep4}

We have $\dim V \leq q+1 = \alpha n + 1 < \alpha'n$ for any $\alpha' \in (\alpha, \frac{(1-\epsilon)^3}{108(1+\epsilon)})$ and $n$ sufficiently large. Then all conditions of Lemma \ref{lemma:ell1} apply so the lemma tells us that there exists $p \in \left\{-1, 1\right\}^n$ such that $p \notin V + B_{\epsilon n}$ and $|\one^Tp| \leq 2\sqrt{n \log n}$.

Since $p \in \left\{-1, 1\right\}^n$, let $C \subseteq [n]$ be the cut such that $u_C = p$ (i.e. $C$ is the set of indices where $p$ has an entry of 1). Now note that
\begin{align*}
    \one^Tu_C &= |C| - (n - |C|) \\
    &= 2|C| - n.
\end{align*}
So we have:
\begin{align*}
    |\one^T u_C| &\leq 2\sqrt{n \log n} \\
    \Leftrightarrow |2|C| -n| &\leq 2\sqrt{n \log n} \\
    \Leftrightarrow ||C| - \frac{n}{2}| &\leq \sqrt{n \log n} \\
    \Leftrightarrow |C| &\in [\frac{n}{2} - \sqrt{n\log n}, \frac{n}{2} + \sqrt{n \log n}].
\end{align*}
Thus $C$ is a near-max cut. Moreover we have $u_C \notin V + B_{\epsilon n}$ so this completes the proof of this corollary and thus Theorem \ref{thm:dethalfhardness}.

\section{Lemmas from Section \ref{sec:randhalfalgo} and Appendix \ref{sec:rswmaxcut}}

\subsection{Proof of Lemma \ref{lemma:allalgoprimitives}}\label{lemmaproof:allalgoprimitives}

We start with a simple lemma:

\begin{lemma}\label{lemma:partialcut}
For any disjoint vertex sets $S, T$, it is possible to find $F(S, T) = \sum_{u \in S, v \in T} w_{u, v}$ in $O(1)$ queries.
\end{lemma}
\begin{proof}
Observe that $F(S, T) = \frac{1}{2}(F(S) + F(T) - F(S \cup T))$ (where here $F(S) = F(S, [n] \backslash S)$ is the standard cut function. Each of these three terms can be found in 1 query, so this completes the proof.
\end{proof}

Now we describe our algorithm. Note that we distinguish between supernodes and individual vertices throughout this description.

At any point, we want to keep track of the degrees of all supernodes. To do this, in \textbf{InitializeDS} we query every singleton set to find the degree of its vertex. This takes $O(n)$ queries.

The only operation that changes the graph is \textbf{Contract}, so we need to maintain our list of degrees when this happens. Let $S$ be the set of vertices in the newly formed supernode after the call to Contract. This contraction does not affect the degrees of any supernodes other than $S$. So we can just query $F(S)$ to find the degree of the supernode $S$ (and delete the degrees of any destroyed supernodes from our list). This only takes 1 query.

Next we address \textbf{GetTotalWeight}. Let the queried set $S$ comprise supernodes $S_1, S_2, \ldots, S_k$. Then observe that:
$$W(S) = \frac{1}{2}((\sum_{i = 1}^k F(S_i)) - F(\bigcup_{i = 1}^k S_i)).$$

We already know each $F(S_i)$ since we are keeping track of supernode degrees, so we only need to find $F(\bigcup_{i = 1}^k S_i)$ which is 1 query.

For \textbf{GetEdge}, we proceed using a ``greedy bisection" procedure:

\begin{enumerate}
    \item Choose the supernode $S$ such that $\frac{F(S)}{|S|(n - |S|)}$ is maximal.
    \item Initialize $T = [n] \backslash S$. While $T$ is not a singleton:
    \begin{enumerate}
        \item Partition $T$ into two equally sized \underline{vertex} sets $T_1$ and $T_2$. Find $F(S, T_i)$ for $i = 1, 2$ using Lemma \ref{lemma:partialcut}.
        \item Update $T \gets T_j$, where $j = \argmax_i \frac{F(S, T_i)}{|S| \cdot |T_i|}$.
    \end{enumerate}
    \item Now let $t$ be the single vertex in $T$. Initialize $U = S$ and repeat the same bisection procedure while $U$ is not a singleton:
    \begin{enumerate}
        \item Partition $U$ into equally sized $U_1$ and $U_2$.
        \item Update $U \gets U_j$, where $j = \argmax_i \frac{F(t, U_i)}{|U_i|}$.
    \end{enumerate}
    \item Let $s$ be the single vertex in $U$. Return the edge between $S$ and $T$.
\end{enumerate}

It is clear that this runs in $O(\log n)$ queries. This procedure can be thought of as follows: we initialize $E$ to be the set of all possible edges, which has average weight at least $\frac{2W(H)}{n^2}$ (the number of available edges will always be at most $n^2/2$). At each point, we then split $E$ into some number of choices and update $E$ to be the one with highest average weight. So at every stage, the average weight of $E$ will be at least $\frac{2W(H)}{n^2}$. In particular, the edge we obtain at the end will have weight at least $\frac{2W(H)}{n^2}$.

Finally, we address \textbf{Sample}. Let the queried set $S$ comprise supernodes $S_1, S_2, \ldots, S_k$. Note that $k \leq n$. We use a similar bisection procedure, but we need to take extra care because our data structure does not track the degree of each $S_i$ within $S$ (it only tracks the degree of each $S_i$ with respect to the entire graph).

\begin{enumerate}
    \item Initialize $T = S$.
    \item Repeat the below indefinitely (it will move to step 3 at some point):
    \begin{enumerate}
        \item Partition $T$ into two equally sized \underline{supernode} sets $T_1, T_2$. (By ``equally sized", we mean containing an equal number of supernodes.)
        \item Let $E_1$ be the set of edges with both endpoints in $T_1$, $E_2$ respectively for $T_2$, and $E_b$ the set of edges with one endpoint in $T_1$ and one endpoint in $T_2$.
        \item Find $W(E_b) = F(T_1, T_2)$ using Lemma \ref{lemma:partialcut} and $W(E_i) = \text{GetTotalWeight}(T_i)$ for $i = 1, 2$.
        \item Choose one of $E_1, E_2, E_b$ with probability proportionate to its total weight. If we chose $E_i$ for $i = 1$ or $2$, update $T \gets T_i$ and repeat step 2. Else, move to step 3.
    \end{enumerate}
    \item Let $T_1, T_2$ be the sets we obtained from step 2. While $T_2$ is not a singleton:
    \begin{enumerate}
        \item Partition $T_2$ into two equally sized \underline{vertex} sets $T_{2, 1}$ and $T_{2, 2}$. (Now by ``equally sized", we mean containing an equal number of vertices.)
        \item Find $F(T_1, T_{2, 1})$ and $F(T_1, T_{2, 2})$ using Lemma \ref{lemma:partialcut}.
        \item Update $T_2 \gets T_{2, 1}$ or $T_{2, 2}$ with probability proportionate to $F(T_1, T_{2, 1})$ and $F(T_1, T_{2, 2})$.
    \end{enumerate}
    \item Now let $t$ be the single vertex in $T_2$. Repeat the same bisection procedure as step 3 while $T_1$ is not a singleton:
    \begin{enumerate}
        \item Partition $T_1$ into two equally sized \underline{vertex} sets $T_{1, 1}$ and $T_{1, 2}$.
        \item Find $F(t, T_{1, 1})$ and $F(t, T_{1, 2})$ using Lemma \ref{lemma:partialcut}.
        \item Update $T_1 \gets T_{1, 1}$ or $T_{1, 2}$ with probability proportionate to $F(t, T_{1, 1})$ and $F(t, T_{1, 2})$.
    \end{enumerate}
    \item Let $s$ be the single vertex in $T_1$. Return the edge between $s$ and $t$.
\end{enumerate}

First we check the efficiency of this.
\begin{enumerate}
    \item Step 2 starts with $T$ being a collection of $k \leq n$ supernodes and each time the number of supernodes gets halved, thus it runs for at most $O(\log n)$ iterations. Each iteration just relies on $O(1)$ calls to Lemma \ref{lemma:partialcut} and GetTotalWeight, which amounts to $O(1)$ queries per iteration and hence $O(\log n)$ queries total.
    \item Step 3 starts with $T_2$ being a set of at most $n$ vertices and once again the number of vertices is getting halved every time. Clearly each iteration uses $O(1)$ queries, so this makes for $O(\log n)$ queries total.
    \item A similar argument applies for Step 4.
\end{enumerate}

It hence remains to check correctness. This procedure can be thought of as initializing $E$ to be the set of all edges with both endpoints in $S$. We then repeatedly split $E$ into two or three subsets $E_i$, and we choose to update $E \gets E_i$ for some $i$ selected with probability proportionate to the total weight of $E_i$. Given this interpretation, it is clear that firstly we will exit Step 2 at some point; $T$ could never become a single supernode because then there would be no edges available with both endpoints in $T$. Then it is also clear that when $E$ is finally a single edge, this edge will have been sampled with probability proportionate to its weight. This completes the proof.

\subsection{Key Concentration Results}

\subsubsection{Setup and Generic Results}\label{sec:sparsifierconcentration}

We first outline our approach. Throughout this section, we allow the graph $G$ to have parallel edges. Benczur and Karger \cite{benczurkarger} use the following concentration result to prove the correctness of their algorithm for constructing sparsifiers in the classical model:

\begin{theorem}\label{theorem:compressionbk}
(\cite{benczurkarger}) Let $\rho_\epsilon = \frac{3(d+3)\log n}{\epsilon^2}$. Let $G$ be a graph with edge weights $w_e$ and strengths $k_e$, and suppose we have edge strength estimates $k'_e$ such that $k_e/r \leq k'_e \leq k_e$ for all $e$ (where $r$ is a constant). Then for each edge $e$, let $p_e = \min(\rho_e w_e/k_e, 1)$ and let $H$ be the graph defined by sampling each edge independently with probability $p_e$. If an edge is chosen, we add it with weight $w_e/p_e$.

Then with probability $1 - O(n^{-d})$, every cut in $H$ has value $1 \pm \epsilon$ times its value in $G$.
\end{theorem}

The reason we cannot use this theorem directly to justify our algorithm is that in our algorithm the sampling is not independent. Our task here is to generalize this theorem to our situation where the sampling is dependent. The main property we exploit is that our algorithm's sampling is negatively correlated; given the knowledge that we sampled some edge $e$, the probability that we also sampled some other edge $e'$ goes down.

\begin{definition}
We say Bernoulli random variables $X_1, \ldots, X_k$ are \emph{negatively correlated} if for every subset $I \subseteq [k]$, for all $j \in I$, and for any $b \in \left\{0, 1\right\}$ we have:
$$\Pr[X_j = b \mid \bigcap_{i \in I, i \neq j} \left\{X_i = b\right\}] \leq \Pr[X_j = b]$$
\end{definition}

The following lemma says that the usual Chernoff bound also holds for negatively correlated random variables:
\begin{lemma}\label{lemma:negcorrchernoff}
(\cite{DD98}) Let $a_1, a_2, \ldots, a_k \in [0, 1]$ be constants and $X_1, \ldots, X_k$ be negatively correlated Bernoulli random variables with probabilities $p_1, \ldots, p_k$. Let $X = \sum_{i = 1}^k a_i X_i$ and $\mu = \EE[X] = \sum_{i = 1}^k a_i p_i$. Then we have:
$$\Pr[X \notin [(1-\epsilon)\mu, (1+\epsilon)\mu]] \leq 2e^{-\epsilon^2 \mu/3}$$
\end{lemma}

\begin{corollary}\label{cor:negcorrchernoffwithdet}
Let $a_1, \ldots, a_k \in [0, 1]$ be constants and $X_1, \ldots, X_k$ be negatively correlated Bernoulli random variables with probabilities $p_1, \ldots, p_k$. Fix some constant $a \geq 0$. Let $X = \sum_{i = 1}^k a_iX_i + a$ and $\mu = \EE[X] = \sum_{i = 1}^k a_ip_i + a$. Then we have:
$$\Pr[X \notin [(1-\epsilon)\mu, (1+\epsilon)\mu]] \leq 2e^{-\epsilon^2\mu/3}$$.
Informally, the Chernoff bound still holds if we add a ``non-integer number" of deterministic random variables that are all equal to 1.
\end{corollary}
\begin{proof}
Define random variables $Y_1, \ldots, Y_{\lfloor a \rfloor + 1}$ to all be deterministically equal to 1 (or equivalently, Bernoulli variables with probability 1). Let $b_1, \ldots, b_{\lfloor a \rfloor} = 1$ and $b_{\lfloor a \rfloor + 1} = a - \lfloor a \rfloor$. The conclusion follows by applying Lemma \ref{lemma:negcorrchernoff} to the variables $X_1, \ldots, X_k, Y_1, \ldots, Y_{\lfloor a \rfloor + 1}$ with constant weights $a_1, \ldots, a_k, b_1, \ldots, b_{\lfloor a \rfloor + 1}$. It is clear that the $X$'s and $Y$'s are collectively negatively correlated.
\end{proof}

With this lemma, we can now adapt the proof of Theorem \ref{theorem:compressionbk} to our needs. We first show that negatively correlated subsampling accurately estimates cut values in graphs with bounded edge weights:

\begin{theorem}\label{theorem:boundedconcentration}
Let $G$ be a random graph on $n$ vertices such that:
\begin{enumerate}
    \item The weight of every edge $e$ is $a_eX_e$ where $a_e \in [0, 1]$ is constant and $X_e$ is Bernoulli with probability $p_e$.
    \item The $X_e$'s are negatively correlated.
    \item The expected weight of every cut in $G$ exceeds $\rho_\epsilon = 3(d+2) \frac{\log n}{\epsilon^2}$.
\end{enumerate}
Then, with probability $1 - O(n^{-d})$, every cut in $G$ has value within $1 \pm \epsilon$ of its expectation.
\end{theorem}
\begin{proof}
By Lemma \ref{lemma:negcorrchernoff}, we know that for every cut $C$, if we let $\mu = \sum_{e \in C} a_e p_e$ we have:
$$\Pr[\sum_{e \in C} a_e X_e \notin [(1-\epsilon)\mu, (1+\epsilon)\mu]] \leq 2e^{-\epsilon^2\mu/3}.$$

Now we can proceed exactly as in \cite{kargersparsifier} to take a union bound over all cuts $C$ to show that the probability that at least one cut deviates from its expectation by more than $\epsilon$ is $O(n^{-d})$.
\end{proof}

Now we can follow the approach of \cite{benczurkarger} to extend this to arbitrarily weighted graphs using the notion of $c$-smoothness:

\begin{definition}
(\cite{benczurkarger}) Let $G$ be a random graph where the weight $U_e$ of edge $e$ is a random variable in the range $[0, m_e]$ with expectation $w_e$. Let $k_e$ be the strength of $e$ in the expected graph $\EE[G]$ where each edge $e$ gets weight $w_e$. Then we say that $G$ is \emph{$c$-smooth} if $cm_e \leq k_e$ for all $e$.
\end{definition}

\begin{lemma} (Decomposition Lemma, due to \cite{benczurkarger}) Any $c$-smooth graph $G$ on $n$ vertices can be decomposed as a positive-weighted sum of at most $n-1$ (possibly dependent) random graphs $F_i$, each with maximum edge weight at most 1 and minimum expected cut at least $c$. Moreover, if $X_e$ is the Bernoulli random variable determining whether $e$ is nonzero in the original graph, the weight of $e$ in $F_i$ will also be a scalar multiple of $X_e$.
\end{lemma}

\begin{theorem}\label{theorem:smoothconcentration}
Let $G$ be a random graph such that:
\begin{enumerate}
    \item The weight of edge $e$ is $m_eX_e$, where $m_e \geq 0$ is constant and $X_e$ is Bernoulli with probability $p_e$.
    \item The $X_e$'s are negatively correlated.
    \item $G$ is $c$-smooth for $c = 3(d+3)\frac{\log n}{\epsilon^2}$.
\end{enumerate}
Then, with probability $1 - O(n^{-d})$, every cut in $G$ has value within $1 \pm \epsilon$ times its expectation.
\end{theorem}
\begin{proof}
Using the Decomposition Lemma, write $G = \sum_{i = 1}^k g_i F_i$, where $g_i \geq 0$ and $k < n$. It is clear that $\EE[G] = \sum_{i = 1}^{k} g_i \EE[F_i]$. Each $F_i$ meets all three conditions for Theorem \ref{theorem:boundedconcentration} with $d+1$ rather than $d$, so the probability that any cut deviates from its expectation in $F_i$ by more than a factor of $1 \pm \epsilon$ is $O(n^{-(d+1)})$. Taking a union bound over all $i$ tells us that the probability any cut deviates from its expectation in any $F_i$ is $O(n^{-d})$.

Now denote $v_H(C)$ to be the value of the cut $C$ in the graph $H$. Then we have with probability $1 - O(n^{-d})$ that for all cuts $C$:
\begin{align*}
    v_{F_i}(C) &\in [(1-\epsilon) \EE[v_{F_i}(C)], (1+\epsilon)\EE[v_{F_i}(C)]] \forall i \\
    \Rightarrow \sum_i g_i v_{F_i}(C) &\in [(1-\epsilon) \sum_i g_i \EE[v_{F_i}(C)], (1+\epsilon) \sum_i g_i \EE[v_{F_i}(C)]] \\
    \Rightarrow v_G(C) &\in [(1-\epsilon) \EE[v_G(C)], (1+\epsilon)\EE[v_G(C)]].
\end{align*}
\end{proof}

This is the core concentration result that we need. We state a more special case of this theorem that we will be able to apply directly in later sections. In this theorem, the $\Lambda_i$'s can be thought of as sets that we contract during Fast Weighted Subsampling and the edge set $\widetilde{E}(\Lambda_i)$ is the set of edges in $\Lambda_i$ that were not already contracted before $\Lambda_i$ was contracted.

\begin{lemma}\label{lemma:specificconcentration}
Fix $\Lambda_1, \ldots, \Lambda_\gamma$ to be a laminar family of distinct subsets of $[n]$ and integers $r_1, r_2, \ldots, r_\gamma > 0$. Also, suppose that each connected component of $G$ is contained in at least one $\Lambda_i$. For each $i$, define $E(\Lambda_i)$ and $\widetilde{E}(\Lambda_i)$ as in Lemma \ref{lemma:sparsifierconcentration}.

Now suppose we construct a random graph $H$ from $G$ as follows: for each $j \in [\gamma]$, independently sample $r_j$ edges from $\widetilde{E}(\Lambda_j)$ proportionally to their weights, with replacement. For each edge $e$, let $p_e = 1 - (1 - \frac{w_e}{W(\widetilde{E}(\Lambda_j))})^{r_j}$ be the probability that it gets selected at least once. If $e$ is selected at least once, add it to $H$ with weight $w_e/p_e$. Else, do not include it in $H$ (equivalently, add it with weight 0).

Then \textbf{if $H$ is $c$-smooth} for $c = 3(d+3)\frac{\log n}{\epsilon^2}$, every cut in $H$ is within $1 \pm \epsilon$ of the corresponding cut in $G$ with probability $1 - O(n^{-d})$.
\end{lemma}
\begin{proof}
Since the lemma supposes $c$-smoothness, we just need to verify conditions 1 and 2 of Theorem \ref{theorem:smoothconcentration} and also that $G = \EE[H]$.
\begin{enumerate}
    \item Note that edge $e$ appears in exactly one $\widetilde{E}(\Lambda_j)$; these edge collections are disjoint by construction and they must also cover all edges since each connected component is contained in some $\Lambda_j$.
    \item Then edge $e$ has weight $\frac{w_e}{p_e} X_e$ in $H$, where $X_e$ is Bernoulli with probability $p_e$. It is clear from this that $\EE[H] = G$ and that condition 1 of the theorem applies.
    \item To check negative correlation, consider any collection $I$ of edges, an edge $e^\ast \in I$, and any $b \in \left\{0, 1\right\}$. Let $i$ be such that $e^\ast \in \widetilde{E}(\Lambda_i)$. Then we have:
    \begin{align*}
        \Pr[X_{e^\ast} = b \mid \bigcap_{e \in I, e \neq e^\ast} \left\{X_e = b\right\}] &= \Pr[X_{e^\ast} = b \mid \bigcap_{e \in I \cap \widetilde{E}(\Lambda_i), e \neq e^\ast} \left\{X_e = b\right\}] \\
        &\leq \Pr[X_{e^\ast} = b].
    \end{align*}
    The first step comes from the fact that edges in different $\widetilde{E}(\Lambda_j)$ are independent. For the second step, any of the remaining edges we condition on are sampled in the same ``batch" as $e^\ast$. If $b = 1$, then we are conditioning that each edge in this collection was sampled at least once; this makes it less likely that $e^\ast$ was sampled at least once since there are fewer rounds available for $e^\ast$ to be sampled.
    
    If $b = 0$, then we are conditioning that no edge in this collection was ever sampled; this makes it more likely that $e^\ast$ was sampled at least once since the total weight of the edges in contention to be sampled is smaller while the weight of $e^\ast$ has not changed, making the probability of selecting $e^\ast$ in each round higher.
\end{enumerate}
\end{proof}

\subsubsection{Proof of Lemma \ref{lemma:sparsifierconcentration}}\label{lemmaproof:sparsifierconcentration}

Let $H = \text{ConstructSparsifier}(G, X)$ and $c = 3(d+3)\frac{\log n}{\epsilon^2}$. Note once again that we are conditioning on fixed values of the $C_i$'s, $\beta_i$'s, and $\mu_i$'s, even though these might be random. We are only concerned here with the independent randomness used in ConstructSparsifier.

We essentially proceed by applying Lemma \ref{lemma:specificconcentration}, but we need to take some extra care in case $H$ is not $c$-smooth.

\begin{lemma}\label{lemma:sparsifiersmoothnesseasy}
Suppose $cw_e/k_e \leq 1/2$ for all edges $e \in G$. Then $H$ is $c$-smooth, and it approximates the cuts of $G$ within a factor of $1 \pm \epsilon$ with probability $1 - O(n^{-d})$.
\end{lemma}
\begin{proof}
We apply Lemma \ref{lemma:specificconcentration}, where $\gamma = r$, $\Lambda_i = C_i$, and $r_i = \mu_i = \frac{c_1 \log^2n}{\epsilon^2\beta_i} W(\widetilde{E}(C_i))$. It is clear that the logic of ConstructSparsifier matches the construction described in the statement of Lemma \ref{lemma:specificconcentration}.

Now we check whether we can use Lemma \ref{lemma:specificconcentration}. The statement of Lemma \ref{lemma:sparsifierconcentration} explicitly supposes that every connected component of $G$ is contained in at least one $C_i$, so that condition is met. It remains to check that $H$ is $c$-smooth.

Consider any edge $e \in G$. Let $i \in [r]$ be such that $e \in \widetilde{E}(C_i)$ and then let $\beta = \beta_i$ and $r = \mu_i$. We wish to show that $\frac{cw_e}{p_e} \leq k_e \Leftrightarrow p_e \geq \frac{cw_e}{k_e}$. Then we have:
\begin{align*}
    p_e &= 1 - (1 - \frac{w_e}{W(\widetilde{E}(C_i))})^r \\
    &\geq 1 - \exp(-\frac{rw_e}{W(\widetilde{E}(C_i))}) \\
    &= 1 - \exp(-\frac{c_1 \log^2n w_e}{\epsilon^2\beta}) \\
    &\geq 1 - \exp(-\frac{c_1 \log^2n w_e}{\epsilon^2k_e}) \\
    &= 1 - \exp(-\alpha x),
\end{align*}
where we let $x = \frac{cw_e}{k_e}$ and $\alpha = \frac{c_1\log n}{3(d+3)}$. We wish to show that $1 - x - \exp(-\alpha x) \geq 0$. Letting $f(x) = 1 - x - e^{-\alpha x}$, we see that $f'(x) = -1 + \alpha e^{-\alpha x}$ which is decreasing. Thus $f$ is concave, so since $x \in [0, 1/2]$ by assumption it suffices to check that $f(0) \geq 0$ and $f(1/2) \geq 0$. It is clear that $f(0) = 0$, and $f(1/2) = 1/2 - e^{-\alpha/2} > 0$ if we set $c_1$ to be large enough. Thus $H$ is $c$-smooth and this completes our proof.
\end{proof}

\begin{lemma}\label{lemma:sparsifiersmoothnesshard}
$H$ approximates the cuts of $G$ within a factor of $1 \pm 2\epsilon$ with probability at least $1 - O(n^{-d})$. (Here we allow for edges where $cw_e/k_e > 1/2$, thus this completes the proof of Lemma \ref{lemma:sparsifierconcentration}.)

Moreover, for any constant $c' < 1$ and $n$ sufficiently large, $H$ is $c'$-smooth. (This is not relevant to the proof of Lemma \ref{lemma:sparsifierconcentration}, but we will need this for Section \ref{sec:nwsearlystopworks}.)
\end{lemma}
\begin{proof}
First we show that $H$ approximates the cuts of $G$. We define a new graph $G_0$ as follows. For any edge $e$ such that $cw_e/k_e \leq 1/2$, copy $e$ from $G$ to $G_0$. If $cw_e/k_e > 1/2$, subdivide it into $O(\log n)$ parallel sub-edges $e'$ of some other weight $w$ such that $cw/k_e \in [1/4, 1/2]$. Call such an edge $e$ \emph{heavy}. We make some observations:
\begin{enumerate}
    \item It is possible to do this with $O(\log n)$ parallel edges because we have $w_e \leq k_e \Rightarrow \frac{cw_e}{k_e} \leq c = O(\log n)$.
    \item Such a subdivision preserves the value of every cut in $G$ and thus the strength of each sub-edge $e'$ will be $k_e$. So for each $e'$ we have $cw_{e'}/k_{e'} = cw/k_e \in [1/4, 1/2]$. Also, the strength of all other edges will be preserved.
\end{enumerate}
Now let $H_0 = \text{ConstructSparsifier}(G_0, X)$. By construction, $G_0$ meets the conditions of Lemma \ref{lemma:sparsifiersmoothnesseasy} so it follows that with probability $1 - O(n^{-d})$, $H_0$ will approximate the cuts of $G_0$ and hence $G$ within a factor of $1 \pm \epsilon$.

We now use a coupling argument to show that this means $H$ also approximates the cuts of $G$ with high probability. We can alternatively characterize the distribution of $H_0$ as being sampled in the following way: take the same collection of edges we sampled from $G$ to construct $H$. Every time we see a heavy edge in this collection, replace it with one of its sub-edges uniformly at random. Once again, include any (possibly sub-)edge that we sampled at least once with weight $w_e/p_e$. If $e$ is a sub-edge, $p_e$ is the probability that that sub-edge was chosen at least once.

We next claim that all sub-edges of heavy edges will be included in $H_0$ with high probability. To see this, consider any such edge. It will satisfy the condition that $cw_e/k_e \geq 1/4$. So we can proceed as in the proof of Lemma \ref{lemma:sparsifiersmoothnesseasy} to find that:
\begin{align*}
    p_e &\geq 1 - \exp(-\alpha x) \\
    &\geq 1 - \exp(-\alpha/4) \\
    &= 1 - n^{-\frac{c_1}{12(d+3)}}
\end{align*}
There are $O(n^2\log n)$ edges that we want to demonstrate this for, so a union bound tells us that with probability $1 - n^{2-\frac{c_1}{12(d+3)}}\log n$, all heavy edges in $G$ and all sub-edges in $G_0$ will be taken. For $c_1$ sufficiently large, this is $1 - O(n^{-d})$.

Thus by a union bound, we have both of the following with probability $1 - O(n^{-d})$:
\begin{enumerate}
    \item $H_0$ approximates the cuts of $G$ within a factor of $1 \pm \epsilon$.
    \item $H_0$ includes all sub-edges in $G_0$ (and hence $H$ also includes all heavy edges in $G$).
\end{enumerate}

So now suppose these are both true. Then we claim that every edge in $H$ has total weight within $1 \pm O(n^{-d})$ of its corresponding total weight in $H_0$. Firstly, if the edge is not heavy then it will have exactly the same weight in both graphs. So now suppose $e$ is heavy with weight $w_e$. Let $A_e$ be the set of $e$'s sub-edges in $H_0$. Let $p_e$ be the probability that $e$ gets included in $H$ and $p'_e$ the probability that a particular one of $e$'s sub-edges will be included in $H_0$. We know that $p_e$ and $p'_e$ are both $1 - O(n^{-d})$. Then the total weight of $e$ in $H$ will be $\frac{w_e}{p_e}$ while its total weight in $H_0$ will be $\frac{w_e}{p'_e}$. The ratio between these is:
\begin{align*}
    \frac{p'_e}{p_e} &= \frac{1 - O(n^{-d})}{1 - O(n^{-d})} \\
    &= (1 - O(n^{-d}))(1 + O(n^{-d})) \\
    &= 1 \pm O(n^{-d})
\end{align*}
so this proves our claim.

Finally, with this claim it follows that $G_0$ approximates the cuts of $H_0$ within a factor of $1 \pm O(n^{-d})$. Moreover, $H_0$ approximates the cuts of $G$ within a factor of $1 \pm \epsilon$. It follows that $G_0$ approximates the cuts of $G$ within a factor of $(1 \pm O(n^{-d}))(1 \pm \epsilon) \in (1 \pm 2\epsilon)$. This occurs with probability $1 - O(n^{-d})$ so this completes the proof of the first part of the lemma.

For the second part of the lemma, we want to show that $\frac{c'w_e}{p_e} \leq k_e$ for any edge $e$. We have two cases:
\begin{itemize}
    \item If $cw_e/k_e \leq 1/2$, then we can proceed as in the proof of Lemma \ref{lemma:sparsifiersmoothnesseasy} to show that $k_e \geq \frac{cw_e}{p_e} \geq \frac{c'w_e}{p_e}$ since $c = \Theta(\log n)$ and $c' < 1$.
    \item If $cw_e/k_e > 1/2$, then we just saw that $p_e = 1 - O(n^{-d})$. So we have that $k_e \geq w_e = \frac{p_ew_e}{p_e} \geq \frac{c'w_e}{p_e}$.
\end{itemize}
This completes the proof of the second part and hence the lemma.



\end{proof}

\subsection{Structure Theorems for Strong Components and Edges}

Here we provide some results due to \cite{benczurkarger} on edge strengths that will be useful in later proofs.

\begin{definition}
Define a vertex set $S$ to be a \emph{$\kappa$-strong set} if $K(G[S]) \geq \kappa$. Moreover, we say $S$ is a \emph{$\kappa$-strong component} if $S$ is $\kappa$-strong and no strict superset of $S$ is.
\end{definition}

\begin{lemma}\label{lemma:strongunion}
If $S$ and $T$ are distinct $\kappa$-strong sets with nonempty intersection, then $S \cup T$ is also $\kappa$-strong.
\end{lemma}
\begin{proof}
Consider any nontrivial cut on $S \cup T$. It must induce a nontrivial cut on at least one of $S$ and $T$, thus it must have value at least $\min(K(G[S]), K(G[T])) \geq \kappa$. It follows that $K(G[S \cup T]) \geq \kappa$ also.
\end{proof}

\begin{corollary}\label{lemma:strongstructure}
The collection of $\kappa$-strong components is a partition of the vertex set of $G$. Moreover, if $\kappa' > \kappa$, then the partition into $\kappa'$-strong components is a refinement of the partition into $\kappa$-strong components.
\end{corollary}
\begin{proof}
For the first property, clearly any vertex belongs to some $\kappa$-strong component and by Lemma \ref{lemma:strongunion} no two distinct $\kappa$-strong components can have nonempty intersection.

For the second property, note that any $\kappa'$-strong component is also a $\kappa$-strong set so it is contained in some $\kappa$-strong component.
\end{proof}

\begin{lemma}\label{lemma:strongcontract}
Suppose we have a graph $G$ and a disjoint collection $\mathcal{C}$ of sets that have been contracted. Let $G'$ be $G$ with all sets in $\mathcal{C}$ contracted. Let $S$ be a $\kappa$-strong component of $G$ that has not been contracted. Then define $S'$ to be the collection of all supernodes from $G'$ that have nonempty intersection with $S$. Then $S'$ is $\kappa$-strong in $G'$.
\end{lemma}
\begin{proof}
Let $T' \subset S'$ be any nontrivial sub-collection of supernodes, and let $C$ be the cut between $T'$ and $S' \backslash T'$ in $G'$. We wish to show that $C$ has value $\geq \kappa$ for any $T'$. Define:
\begin{align*}
    T &= \left\{v \in S: v\text{ belongs to some supernode in }T'\right\} \\
    U &= \left\{v \in S: v\text{ belongs to some supernode in }S' \backslash T'\right\}.
\end{align*}
We claim that $T, U$ nontrivially partition $S$. We check three things to verify this:
\begin{enumerate}
    \item $T, U$ are disjoint because $T'$ and $S' \backslash T'$ are disjoint.
    \item Every vertex in $S$ belongs to some supernode in $S'$, so it will belong to either $T$ or $U$. Thus $T$ and $U$ partition $S$.
    \item By construction, every supernode of $S'$ contains at least one element from $S$. Since $T'$ and $S' \backslash T'$ each must have at least one supernode from $S'$, it follows that $T$ and $U$ must both be nonempty.
\end{enumerate}
To finish, note that any edge from $G$ between a vertex in $T$ and a vertex in $U$ cannot have been contracted yet since its endpoints belong to distinct supernodes. Any such edge also appears in $C$ by definition. It follows that the value of $C$ is at least the value of the cut between $T$ and $U$ in $G$. But this latter cut is a cut of $S$ and therefore has value at least $\kappa$. This completes the proof.
\end{proof}

\begin{corollary}\label{cor:strongcontract}
Suppose we have a graph $G$ and a disjoint collection $\mathcal{C}$ of sets that have been contracted. Let $G'$ be $G$ with all sets in $\mathcal{C}$ contracted. Then for any edge $e \in G'$, its strength in $G'$ is at least its strength in $G$.
\end{corollary}

\begin{lemma}\label{lemma:weakcontract}
Suppose we have a graph $G$ and a disjoint collection $\mathcal{C}$ of sets that have been contracted. Let $G'$ be $G$ with all sets in $\mathcal{C}$ contracted. Suppose $e$ is an edge of strength $< \kappa$ (in $G$) that has not been contracted. Moreover, suppose that every contracted edge in $\mathcal{C}$ is $\kappa$-strong. Then the strength of $e$ in $G'$ is also $< \kappa$.
\end{lemma}
\begin{proof}
Define $G''$ to be $G$ after contracting all $\kappa$-strong components. It is clear that $G''$ can also be obtained from $G'$ via a series of further contractions.

We first show that $e$ has strength $< \kappa$ in $G''$. Let $S''$ be any collection of supernodes of $G''$ containing both endpoints of $e$. We wish to show that $S''$ has a nontrivial cut of value $< \kappa$. To do this, let $S$ be the vertex set of $G$ comprising the union of all supernodes in $S''$. $S$ clearly contains both endpoints of $e$ as well, so by assumption $S$ has a nontrivial cut $C$ of value $< \kappa$. Now we claim that $C$ cannot nontrivially partition the vertices of any supernode in $S''$. If it did, then it would induce a nontrivial cut on that supernode, which is a $\kappa$-strong component. Then $C$ itself would have value $\geq \kappa$, which is a contradiction. This proves our claim.

Thus every supernode in $S''$ is either entirely included in or entirely excluded from $C$. This implies that $C$ also corresponds to a valid and non-trivial cut on $S''$, and this cut has value $< \kappa$ as claimed.

Thus the strength of $e$ in $G''$ is $< \kappa$. Since $G''$ can be obtained from $G'$ by contractions, it follows by Corollary \ref{cor:strongcontract} that the strength if $e$ in $G'$ is also $< \kappa$, as desired.
\end{proof}

\subsection{Properties of Estimate and Contract}\label{lemmaproofs:subsampleandcontract}

For any edge $e$, as usual we let $k_e$ denote its edge strength in $G$. For convenience, we also let $l_e$ denote its edge strength in the contracted graph $G'$. Note that we only define $l_e$ if $e$ is not already contracted in $G'$.

\subsubsection{Proof of Lemma \ref{lemma:claim1-2}}\label{lemmaproof:claim1-2}

We want to show that all $\kappa$-strong components in $G$ get contracted if they are not already contracted. For clarity, let $G''$ denote the subsampled graph after step 1 of EstimateAndContract and $G''_0$ denote the the $G''$ that results from the iterative deletion procedure in step 2.

For the next few steps, fix one $\kappa$-strong component $S$ that is not yet contracted. Let $G'_S$ denote the vertex-induced subgraph of $G'$ on all supernodes of $G'$ that have nonempty intersection with $S$, and similarly define $G''_S$. By Lemma \ref{lemma:strongcontract}, $G'_S$ is also $\kappa$-strong.

We will show using an extremely similar approach to the proofs in Appendix \ref{lemmaproof:sparsifierconcentration} that $G''_S$ approximates the cuts of $G'_S$ within a factor of $1 \pm 2c_0^{-1/3}$ with high probability.

\begin{lemma}\label{lemma:1-2smoothnesseasy}
If $cw_e/l_e \leq 1/2$ for all edges $e \in G'_S$, then $G''_S$ approximates the cuts of $G'_S$ within a factor of $1 \pm c_0^{-1/3}$ with probability $1 - O(n^{-d})$.
\end{lemma}
\begin{proof}
Let $\epsilon_0 = c_0^{-1/3}$ (note that this is different from the $\epsilon$ referenced in Algorithm \ref{algo:sparsify}). We start by applying Lemma \ref{lemma:specificconcentration} to $G'$. We set $\gamma = 1$ and $\Lambda_1 = [n]$ and $r_1 = \lambda$. It is clear that these parameters correspond to how $G''$ is constructed from $G'$. Each connected component of $G'$ is trivially contained in $\Lambda_1$ so that condition is met. So it follows from the proof of Lemma \ref{lemma:specificconcentration} that $\EE[G''_S] = G_S$, that each edge of $G''_S$ is a scalar multiple of a Bernoulli random variables, and that these Bernoulli variables are negatively correlated. It only remains to check that $G''_S$ is $c$-smooth.

So consider any edge $e$ in $G'_S$. Since $G'_S$ is $\kappa$-strong, we have that $l_e \geq \kappa$. The maximum value $m_e$ of $e$ in $G''_S$ is $w_e/p_e$, so we wish to show that $cm_e \leq l_e \Leftrightarrow cw_e/p_e \leq l_e \Leftrightarrow p_e \geq \frac{cw_e}{l_e}$. We have:
\begin{align*}
    p_e &= 1 - (1 - \frac{w_e}{W(G')})^\lambda \\
    &\geq 1 - \exp(-\frac{\lambda w_e}{W(G')}) \\
    &= 1 - \exp(-\frac{c_0\log^2 n w_e}{\kappa}) \\
    &\geq 1 - \exp(-\frac{c_0\log^2 n w_e}{l_e}) \\
    &= 1 - \exp(-\alpha x)
\end{align*}
where we let $x = \frac{cw_e}{l_e}$ and $\alpha = \frac{c_0\epsilon_0^2\log n}{3(d+3)} = \frac{c_0^{1/3}\log n}{3(d+3)}$. We wish to show that $1 - \exp(-\alpha x) \geq x \Leftrightarrow 1 - x - \exp(-\alpha x) \geq 0$. Letting $f(x) = 1 - x - e^{-\alpha x}$, we see that $f'(x) = -1 + \alpha e^{-\alpha x}$ which is decreasing. Thus $f$ is concave, so since $x \in [0, 1/2]$ by assumption it suffices to check that $f(0) \geq 0$ and $f(1/2) \geq 0$. It is clear that $f(0) = 0$, and $f(1/2) = 1/2 - e^{-\alpha/2} > 0$ if we set $c_0$ to be large enough. Thus $G''_S$ is $c$-smooth and this completes our proof.
\end{proof}

\begin{lemma}\label{lemma:1-2smoothnesshard}
$G''_S$ approximates the cuts of $G'_S$ within a factor of $1 \pm 2c_0^{-1/3}$ with probability $1 - O(n^{-d})$. (Here, we allow for edges where $cw_e/l_e > 1/2$.)
\end{lemma}
\begin{proof}
We use exactly the same parallel edge trick as in Lemma \ref{lemma:sparsifiersmoothnesshard}. Define a new graph $H'$ as follows. For any edge $e$ such that $cw_e/l_e \leq 1/2$, copy $e$ from $G'$ to $H'$. If $cw_e/l_e > 1/2$, then subdivide it into $O(c_0\log n)$ parallel sub-edges $e'$ of some other weight $w$ such that $cw/l_e \in [1/4, 1/2]$. (The total weight of these sub-edges should equal the weight of the original edge.) Call such an edge $e$ \emph{heavy}. We make some observations:
\begin{enumerate}
    \item It is possible to do this with only $O(c_0\log n)$ parallel edges because we have $w_e \leq l_e \Rightarrow \frac{cw_e}{l_e} \leq c = O(c_0\log n)$.
    \item Such a subdivision preserves the value of every cut in $G'$ and thus the strength of each sub-edge $e'$ will be $l_e$. It follows that for each $e'$ we have $cw_{e'}/l_{e'} = cw/l_e \in [1/4, 1/2]$ in $H'$. Also, the strength of all other edges will be preserved.
\end{enumerate}
Now suppose we run the same subsampling procedure that we did to obtain $G''_S$, but now we run it on $H'$ to obtain $H''_S$. By construction, $H'$ meets the conditions of Lemma \ref{lemma:1-2smoothnesseasy} so it follows that with probability $1 - O(n^{-d})$, $H''_S$ will approximate the cuts of $H'_S$ and hence $G'_S$ within a factor of $1 \pm c_0^{-1/3}$.

We now use a coupling argument to show that this means $G''_S$ also approximates the cuts of $G'_S$ with high probability. We can alternatively characterize the distribution of $H''_S$ as being sampled in the following way: take the same collection of edges we sampled from $G'_S$ to construct $G''_S$. Every time we see a heavy edge in this collection, replace it with one of its sub-edges uniformly at random. Once again, include any (possibly sub-)edge that we sampled at least once with weight $w_e/p_e$. If $e$ is a sub-edge, $p_e$ is the probability that the sub-edge was chosen at least once.

We next claim that all sub-edges of heavy edges will be included in $H''_S$ with high probability. To see this, consider any such edge. It will satisfy the condition that $cw_e/l_e \geq 1/4$. So we can proceed as in the proof of Lemma \ref{lemma:1-2smoothnesseasy} to find that:
\begin{align*}
    p_e &\geq 1 - \exp(-\alpha x) \\
    &\geq 1 - \exp(-\alpha/4) \\
    &= 1 - n^{-\frac{c_0^{1/3}}{12(d+3)}}
\end{align*}
There are $O(c_0n^2\log n)$ edges that we want to demonstrate this for ($O(n^2)$ edges in the original graph and then $O(\log n)$ sub-edges for each original edge), so a union bound tells us that with probability $1 - c_0n^{2-\frac{c_0^{1/3}}{12(d+3)}}\log n$, all heavy edges in $G'_S$ and all sub-edges in $H'_S$ will be taken. For $c_0$ sufficiently large, this is $1 - O(n^{-d})$.

Thus by a union bound, we have both of the following with probability $1 - O(n^{-d})$:
\begin{enumerate}
    \item $H''_S$ approximates the cuts of $G'_S$ within a factor of $1 \pm c_0^{-1/3}$.
    \item $H''_S$ includes all sub-edges in $H'_S$ (and hence $G''_S$ also includes all heavy edges in $G'_S$).
\end{enumerate}
So now suppose these are both true. Then we claim that every edge in $G''_S$ has total weight within $1 \pm O(n^{-d})$ of its corresponding total weight in $H''_S$. Firstly, if the edge is not heavy then it will have exactly the same weight in both graphs. So now suppose $e$ is heavy with weight $w_e$. Let $A_e$ be the set of $e$'s sub-edges in $H'_S$. Let $p_e$ be the probability that $e$ gets included in $G''_S$ and $p'_e$ the probability that a particular one of $e$'s sub-edges will be included in $H''_S$. We know that $p_e$ and $p'_e$ are both $1 - O(n^{-d})$. Then the total weight of $e$ in $G''_S$ will be $\frac{w_e}{p_e}$ while the total weight of $e$ in $H''_S$ will be $\frac{w_e}{p'_e}$. The ratio between these is:
\begin{align*}
    \frac{p'_e}{p_e} &= \frac{1 - O(n^{-d})}{1 - O(n^{-d})} \\
    &= (1 - O(n^{-d}))(1 + O(n^{-d})) \\
    &= 1 \pm O(n^{-d})
\end{align*}
so this proves our claim.

Finally, with this claim it follows that $G''_S$ approximates the cuts of $H''_S$ within a factor of $1 \pm O(n^{-d})$. Moreover, $H''_S$ approximates the cuts of $G'_S$ within a factor of $1 \pm c_0^{-1/3}$. It follows that $G''_S$ approximates the cuts of $G'_S$ within a factor of $(1 \pm O(n^{-d}))(1 \pm c_0^{-1/3}) \in (1 \pm 2c_0^{-1/3})$. This occurs with probability $1 - O(n^{-d})$ so this completes the proof.




\end{proof}

\begin{lemma}\label{lemma:1-2single}
With probability $1 - O(n^{-d})$, $S$ will not be contracted.
\end{lemma}
\begin{proof}
By Lemma \ref{lemma:1-2smoothnesshard} and the observation that $G'_S$ is also $\kappa$-strong, we have that with probability $1 - O(n^{-d})$, $G''_S$ will be $(1 - 2c_0^{-1/3})\kappa$-strong. For $c_0$ sufficiently large, this means that every cut in $G''_S$ will have value strictly greater than $(1-\delta)\kappa$.

This implies that $S$ has to stay connected throughout Step 2 of Algorithm \ref{algo:contract}. Suppose not, then consider the first time a cut containing edges from $S$ gets deleted. This cut would be on an entire (currently) connected component of $G''$ so its vertex set would contain $S$ and thus it would induce a nontrivial cut on $G''_S$. This nontrivial cut on $G''_S$ has value $> (1-\delta)\kappa$ so the original cut would also have value $> (1-\delta)\kappa$ so it could not be deleted. This is a contradiction.

Thus $S$ stays connected throughout step 2, so $S$ remains connected in $G''_0$ and thus it will be contracted with probability $1 - O(n^{-d})$.
\end{proof}

\begin{lemma}
With probability $1 - O(n^{1-d})$, no $\kappa$-strong component in $G$ will be contracted. This will complete the proof of Lemma \ref{lemma:claim1-2}.
\end{lemma}
\begin{proof}
The $\kappa$-strong components of $G$ partition $[n]$ so there are at most $n$ of them. The conclusion follows by applying Lemma \ref{lemma:1-2single} and taking a union bound over all the components.
\end{proof}

\subsubsection{Proof of Lemma \ref{lemma:claim1-3}}\label{lemmaproof:claim1-3}

As in the previous section, we let $G''$ denote the subsampled graph after step 1 of EstimateAndContract and $G''_0$ denote the $G''$ that results from the iterative deletion procedure. We want to show that with high probability, $G''_0$ does not contain any edges of strength $< \kappa/2$. We begin with a simple lemma that says that a cut that is small in $G'$ will typically not be large in $G''$:

\begin{lemma}\label{lemma:weakcutconcentration}
Let $S$ be any collection of edges in $G'$ with total weight $< \kappa/2$ in $G'$. Then, with probability $1 - O(n^{-d})$ the total weight of $S$ in $G''$ is $< (1-\delta)\kappa$.
\end{lemma}
\begin{proof}
Adopting the usual notation, the probability we wish to bound is:
$$\Pr[\sum_{e \in S} \frac{w_e}{p_e} X_e \geq (1-\delta)\kappa].$$ Observe that:
\begin{align*}
    \sum_{e \in S} \frac{w_e}{p_e} X_e &= \sum_{e \in S} \frac{w_e}{1 - (1 - \frac{w_e}{W(G')})^\lambda} X_e \\
    &\leq \sum_{e \in S} \frac{w_e}{1 - \exp(-\lambda w_e/W(G'))} X_e \\
    &= \sum_{e \in S} \frac{w_e}{1 - \exp(-\frac{c_0\log^2n}{\kappa} w_e)} X_e
\end{align*}
Now let $f(x) = \frac{1 - e^{-x}}{x}$ for $x \geq 0$. We have $f'(x) = \frac{e^{-x}}{x^2}(1+x-e^x) \leq 0$. It follows that for $x \leq 1$ we have $f(x) \geq f(1) = 1 - 1/e \Rightarrow 1 - e^{-x} \geq (1 - 1/e)x$, and for $x \geq 1$ we have $1 - e^{-x} \geq 1 - 1/e$. So let $S_1 = \left\{e \in S: \frac{c_0\log^2n}{\kappa} w_e \leq 1\right\}$ and $S_2 = S \backslash S_1$. Applying both of our estimates gives:
\begin{align*}
    \sum_{e \in S} \frac{w_e}{p_e} X_e &\leq \sum_{e \in S} \frac{w_e}{1 - \exp(-\frac{c_0\log^2n}{\kappa} w_e)} X_e \\
    &\leq \sum_{e \in S_1} \frac{1}{(1 - 1/e)\frac{c_0\log^2n}{\kappa}} X_e + \sum_{e \in S_2} \frac{w_e}{1 - 1/e} X_e \\
    &\leq \frac{\kappa}{(1 - 1/e)c_0\log^2n} (\sum_{e \in S_1} X_e + \frac{c_0\log^2n}{\kappa} \sum_{e \in S_2} w_e).
\end{align*}
Thus the probability we want to bound is at most the probability that:
\begin{align*}
    \frac{\kappa}{(1 - 1/e)c_0\log^2n} (\sum_{e \in S_1} X_e + \frac{c_0\log^2n}{\kappa} \sum_{e \in S_2} w_e) &\geq (1-\delta)\kappa \\
    \Leftrightarrow \sum_{e \in S_1} X_e + \frac{c_0\log^2n}{\kappa} \sum_{e \in S_2} w_e &\geq (1 - \delta)(1 - 1/e)c_0\log^2 n
\end{align*}
Now we have:
\begin{align*}
    \EE[\sum_{e \in S_1} X_e + \frac{c_0\log^2n}{\kappa} \sum_{e \in S_2} w_e] &= \sum_{e \in S_1} p_e + \frac{c_0 \log^2n}{\kappa} \sum_{e \in S_2} w_e \\
    &= \sum_{e \in S_1} (1 - (1 - \frac{w_e}{W(G')})^\lambda) + \frac{c_0\log^2n}{\kappa} \sum_{e \in S_2} w_e \\
    &\leq \sum_{e \in S_1} \frac{\lambda w_e}{W(G')} + \frac{c_0\log^2n}{\kappa} \sum_{e \in S_2} w_e \text{ (Bernoulli's inequality)} \\
    &= \frac{c_0\log^2n}{\kappa} \sum_{e \in S} w_e \\
    &\leq \frac{c_0\log^2 n}{2}
\end{align*}
Now define $t = \frac{c_0\log^2n}{2} - \EE[\sum_{e \in S_1} X_e + \frac{c_0\log^2n}{\kappa} \sum_{e \in S_2} w_e] \geq 0$. Note that $t$ is deterministic. Now define the random variable $X$ to be $\sum_{e \in S_1} X_e + (\frac{c_0\log^2n}{\kappa} \sum_{e \in S_2} w_e + t)$. We have that:
\begin{align*}
    \Pr[\sum_{e \in S_1} X_e + \frac{c_0\log^2n}{\kappa} \sum_{e \in S_2} w_e &\geq (1-\delta)(1-1/e)c_0\log^2n] \\
    \leq \Pr[X &\geq (1-\delta)(1-1/e)c_0\log^2n] \\
    \leq \Pr[X &\geq \frac{3}{5}c_0\log^2 n]
\end{align*}
for $\delta$ sufficiently close to 1. Moreover, $X$ meets all the necessary conditions for us to apply Corollary \ref{cor:negcorrchernoffwithdet} (as in the proof of Lemma \ref{lemma:1-2smoothnesseasy}, it follows from Lemma \ref{lemma:specificconcentration} that the $X_e$'s are negatively correlated). By construction, we have $\EE[X] = \frac{c_0 \log^2n}{2}$. So applying the corollary gives:
\begin{align*}
    \Pr[X \geq \frac{3}{5}c_0\log^2n] &= \Pr[X \geq \frac{6}{5} \EE[X]] \\
    &\leq 2\exp(-\frac{c_0\log^2 n}{150}) \\
    &= O(n^{-d}),
\end{align*}
for $n$ sufficiently large, as desired.
\end{proof}

We next show that there is a sequence of $\leq n$ cuts we can make in $G'$, all with total weight $< \kappa/2$, such that at the end we have divided $G'$ into its $(\kappa/2)$-strong components.

\begin{lemma}\label{lemma:cutsequence}
There exists an integer $k \leq n$ and a sequence of partitions $\mathcal{P}_1, \ldots, \mathcal{P}_k$ of the supernode set of $G'$ satisfying the following conditions:
\begin{enumerate}
    \item $\mathcal{P}_1$ consists of just 1 set which is all of $G'$.
    \item $\mathcal{P}_k$ consists of exactly the $(\kappa/2)$-strong components of $G'$.
    \item $\mathcal{P}_i$ is obtained from $\mathcal{P}_{i-1}$ by selecting some set within $\mathcal{P}_{i-1}$ and splitting it into two using a cut of value less than $\kappa/2$ in $G'$.
\end{enumerate}
\end{lemma}
\begin{proof}
Initialize $\mathcal{P} = \mathcal{P}_1$. At any point, we will ensure that all sets of $\mathcal{P}$ are the union of some subcollection of the $(\kappa/2)$-strong components of $G'$. While there is some set $S$ in $\mathcal{P}$ that contains multiple $(\kappa/2)$-strong components of $G'$, proceed as follows. By assumption, $S$ itself is not $(\kappa/2)$-strong, so we can find a cut in $S$ of value $< \kappa/2$ splitting it into $S_1$ and $S_2$. Replace $S$ with $S_1$ and $S_2$ in $\mathcal{P}$. Our sequence $\mathcal{P}_i$ is just the sequence of states $\mathcal{P}$ attains throughout the course of this process. We need to check a few things to complete the proof of the lemma:

\begin{enumerate}
    \item To see that $S_1$ and $S_2$ are each still the union of some subcollection of $(\kappa/2)$-strong components: if they were not, there would exist some $(\kappa/2)$-strong component with some vertices in each of $S_1$ and $S_2$. But then our chosen cut would induce a cut on this component and thus have value at least $\kappa/2$, a contradiction.
    \item By definition, this procedure terminates when every set in $\mathcal{P}$ is a $(\kappa/2)$-strong component. This ensures that $\mathcal{P}_k$ has the required structure.
    \item Finally, $k \leq n$ because it is clear by induction that $\mathcal{P}_i$ consists of $i$ sets, and there are at most $n$ supernodes to partition.
\end{enumerate}
\end{proof}

Next, we use this to show that with high probability, Step 2 of Algorithm \ref{algo:contract} must delete all edges that are of strength $\kappa/2$ in $G'$. (Note that this does not yet complete the proof of the lemma, as strength in $G'$ need not be the same as strength in $G$.)

\begin{lemma}\label{lemma:1-3easy}
Let $\mathcal{P}_k$ be the partition of the supernodes of $G'$ into $(\kappa/2)$-strong components. Then with probability $1 - O(n^{1-d})$, the partition $\mathcal{P}$ defined by $G''_0$'s connected components will be a refinement of $\mathcal{P}_k$.
\end{lemma}
\begin{proof}
Let $\mathcal{P}_1, \ldots, \mathcal{P}_k$ be defined as in Lemma \ref{lemma:cutsequence}. It trivially holds that $\mathcal{P}$ is a refinement of $\mathcal{P}_1$. We show for each $i < k$ that $\Pr[\mathcal{P}\text{ is a refinement of }\mathcal{P}_i\text{ but not of }\mathcal{P}_{i+1}]$ is $O(n^{-d})$. Once we show this, it will follow by a union bound over indices $i$ that $\Pr[\mathcal{P}\text{ is not a refinement of }\mathcal{P}_k]$ is $O(n^{1-d})$, and this will complete the proof of the lemma.

So now fix an index $i$. Let $S$ be the set in $\mathcal{P}_i$ that gets split into $S_1$ and $S_2$ in $\mathcal{P}_{i+1}$. If $\mathcal{P}$ is a refinement of $\mathcal{P}_i$ but not of $\mathcal{P}_{i+1}$, then the cut on $S$ splitting it into $S_1$ and $S_2$ must induce a nontrivial cut on some set in $\mathcal{P}$. But every set in $\mathcal{P}$ is $(1-\delta)\kappa$-strong in $G''$ by construction. Thus the cut on $S$ splitting it into $S_1$ and $S_2$ must have value at least $(1-\delta)\kappa$ in $G''$ despite having value $< \kappa/2$ in $G'$ by definition. By Lemma \ref{lemma:weakcutconcentration}, this occurs with probability $O(n^{-d})$. This proves our claim and hence the lemma.
\end{proof}

Finally, we tie this back to edges of strength $< \kappa/2$ in $G$:

\begin{lemma}
With probability $1 - O(n^{1-d})$, no edge $e$ of strength $< \kappa/2$ (that has not already been contracted) will be included in $G''_0$. This completes the proof of Lemma \ref{lemma:claim1-3}.
\end{lemma}
\begin{proof}
By the assumption of Lemma \ref{lemma:claim1-3}, we may apply Lemma \ref{lemma:weakcontract} to find that any edge $e$ we are concerned with will also have strength $< \kappa/2$ in $G'$. Now by Lemma \ref{lemma:1-3easy}, with probability $1 - O(n^{1-d})$, the partition $\mathcal{P}$ defined by $G''_0$'s connected components will be a refinement of $\mathcal{P}_k$.

If this is the case, then since $e$ has strength $< \kappa/2$ in $G'$, its endpoints belong to distinct $(\kappa/2)$-strong components in $G'$. Then $e$'s endpoints must also belong to distinct sets in $\mathcal{P}$ i.e. they are no longer connected in $G''_0$. This must mean that $e$ itself was not included in $G''_0$. This is true for any edge $e$ of strength $< \kappa/2$ in $G$ (that has not been contracted), so the conclusion follows.
\end{proof}






\subsection{Correctness of Naive Weighted Subsampling}\label{lemmaproofs:nwscorrect}

Let us first show that every call made to EstimateAndContract behaves in alignment with Lemmas \ref{lemma:claim1-2} and \ref{lemma:claim1-3} with high probability. In subsequent sections, we will assume that all calls to EstimateAndContract run successfully.

\begin{lemma}\label{lemma:nwsiterationcount}
If every call to EstimateAndContract in NaiveWeightedSubsample behaves in alignment with Lemmas \ref{lemma:claim1-2} and \ref{lemma:claim1-3}, then NaiveWeightedSubsample will call EstimateAndContract $O(\min (\log n + \log W, \log T))$ times.
\end{lemma}
\begin{proof}
Note firstly that $\kappa$ is initialized to be $< 2W_\text{tot}$ and it is halved after each call to EstimateAndContract. The loop in NaiveWeightedSubsample will terminate once $\kappa \leq W_\text{tot}/T$, so this allows for at most $O(\log T)$ iterations. It remains to show that the number of iterations must also be at most $O(\log n + \log W)$.

Let $w_\text{max}$ and $w_\text{min}$ be the maximum and minimum nonzero edge weights in $G$. Then note that any edge in $G$ must have strength $\geq w_\text{min}$ (since we must have $k_e \geq w_e \geq w_\text{min}$.) Since we are assuming here that EstimateAndContract always runs correctly, Lemma \ref{lemma:claim1-2} tells us that all edges of $G$ should be contracted once we have run EstimateAndContract with value $\kappa \leq w_\text{min}$. Thus the number of iterations is at most:
\begin{align*}
    O(\log \frac{W_\text{tot}}{w_\text{min}}) &= O(\log \frac{n^2 w_\text{max}}{w_\text{min}}) \\
    &= O(\log (n^2W)) \\
    &= O(\log n + \log W),
\end{align*}
which completes the proof of the lemma.
\end{proof}
Taking a union bound over the first $O(\min (\log n + \log W, \log T))$ iterations of EstimateAndContract and applying the probability estimates from Lemmas \ref{lemma:claim1-2} and \ref{lemma:claim1-3} yields the following:

\begin{corollary}\label{cor:contractunionboundnws}
With probability $1 - O(n^{1-d} \cdot \min(\log n + \log W, \log T))$, every call to EstimateAndContract in NaiveWeightedSubsample will behave in alignment with Lemmas \ref{lemma:claim1-2} and \ref{lemma:claim1-3} (i.e. the events that the lemmas claim will happen with high probability will all happen).
\end{corollary}

\subsubsection{Estimate and Contract Gives Good Strength Estimates}

\begin{definition}\label{def:labelled}
We say an edge $e$ is \emph{labeled} in a particular run of EstimateAndContract if in that run $e$ is contracted. Moreover, we say $e$ is \emph{correctly labeled} if $e$ was given an edge strength estimate $k'_e \in [k_e/4, k_e]$. We say $e$ is \emph{incorrectly labeled} if $e$ was given an edge strength estimate outside this range.
\end{definition}

Note firstly that $\kappa$ will always be a power of 2 in every call to EstimateAndContract. Moreover, it decreases by a factor of 2 after each call. What we want to show is that all edges of sufficiently high strength will be correctly labeled at the end of Naive Weighted Subsampling. The following lemma is the key step for this:

\begin{lemma}\label{lemma:nws1round}
We have the following for each round of the loop in Algorithm \ref{algo:nws}. Let $\kappa$ denote the value at the start of the round (so before it gets halved).
\begin{enumerate}
    \item After the call to EstimateAndContract, every edge with strength $\geq \kappa$ has been correctly labeled.
    \item After the call to EstimateAndContract, no edge with strength $< \kappa/2$ has been labeled.
    \item After the call to EstimateAndContract, any edge with strength in $[\kappa/2, \kappa)$ will either not be labeled, or will be correctly labeled.
\end{enumerate}
\end{lemma}
\begin{proof}
We proceed by induction on the number of rounds that have elapsed. The base case is vacuously true as at the beginning no edges have been labeled or contracted at all. So assume that all previous rounds have proceeded according to our claim and now we focus on the current one.

We first show property 1. Consider any edge $e$ of strength $\geq \kappa$ that has not yet been contracted. Lemma \ref{lemma:claim1-2} tells us that $e$ will be contracted and labeled with an estimate of $\kappa/2$. This will be a correct labeling if $k_e \leq 2\kappa$. So we need to show that any edge of strength $> 2\kappa$ has already been contracted. We have two cases:
\begin{enumerate}
    \item If $\kappa$ is not the first value of $\kappa$ used, this is clear by induction as we previously ran EstimateAndContract with $2\kappa$ in the place of $\kappa$.
    \item If $\kappa$ is the first value, then by definition we have $\kappa \geq W_\text{tot} \geq k_e$ for any edge $e$, so there cannot be any edge of strength $> 2\kappa$.
\end{enumerate}
Next we show property 2. Consider any edge of strength $< \kappa/2$ in $G$. By induction, $e$ cannot have been labeled before this iteration as previous iterations were all run with larger values of $\kappa$. So the property will follow from Lemma \ref{lemma:claim1-3} if we can check that every edge $e$ already contracted in $\mathcal{C}$ has strength $\geq \kappa/2$. We will show more strongly that it must have strength $\geq \kappa$. Indeed, any such edge must have been contracted in some previous iteration where $\kappa$ was at least twice as large, so by induction its strength must be at least $2\kappa/2 = \kappa$. This completes the proof of property 2.

Finally, we show property 3. Consider an edge $e$ with strength in $[\kappa/2, \kappa)$. Any previous iteration would have used a value $\geq 2\kappa$, so as we just noted no edge of strength $< \kappa$ could have already been contracted. So one of the following must happen:
\begin{itemize}
    \item $e$ is not labeled.
    \item $e$ is labeled with $k'_e = \kappa/2$. We have $k_e/2 < k'_e \leq k_e$ so in this case $e$ will be correctly labeled.
\end{itemize}
This completes the proof of the lemma.
\end{proof}

\begin{corollary}\label{cor:nwsestimate}
At the end of NaiveWeightedSubsample, every edge with strength $\geq 2W_\text{tot}/T$ will be correctly labeled. Moreover, no edges will be incorrectly labeled.
\end{corollary}
\begin{proof}
It is clear from Lemma \ref{lemma:nws1round} that no edge in $G$ will ever be incorrectly labeled, so it suffices to show that all edges with strength $\geq 2W_\text{tot}/T$ will be labeled. To see this, we have two cases depending on how the loop in NaiveWeightedSubsample terminates:
\begin{itemize}
    \item If it terminates because $\text{GetTotalWeight}(G) = 0$, this means that all edges have been labeled and contracted.
    \item If it terminates because $\kappa \leq W_\text{tot}/T$, then the final call to EstimateAndContract ran with twice this value of $\kappa$, which is $\leq 2W_\text{tot}/T$. In this case, the conclusion follows from property 1 of Lemma \ref{lemma:nws1round} that any edge with strength $\geq 2W_\text{tot}/T$ has been correctly labeled.
\end{itemize}
\end{proof}

\subsubsection{$\text{NaiveWeightedSubsample}(G, \infty)$ Constructs a Sparsifier}\label{sec:nwsinfcorrect}

If every call to EstimateAndContract runs correctly, then Corollary \ref{cor:nwsestimate} tells us that every edge of strength $\geq 2W_\text{tot}/\infty = 0$ is correctly labeled.

Thus we may apply Lemma \ref{lemma:sparsifierconcentration} to find that in this case, $H = \text{ConstructSparsifier}(G, X)$ will be a $(2\epsilon)$-sparsifier of $G$ with probability $1 - O(n^{-d})$.

In summary: for $H$ to be a $(2\epsilon)$-sparsifier, we need all calls to EstimateAndContract to run correctly, and we also need ConstructSparsifier to succeed. By Corollary \ref{cor:contractunionboundnws}, the former occurs with probability $1 - O(n^{1-d} \cdot \min(\log n + \log W, \log \infty)) = 1 - O(n^{1-d}(\log n + \log W))$ and the latter occurs with probability $1 - O(n^{-d})$. Thus a union bound tells us that $H$ is a $(2\epsilon)$-sparsifier of $G$ with probability $1 - O(n^{1-d}(\log n + \log W))$, as desired.

\subsubsection{$\text{NaiveWeightedSubsample}(G, n^3)$ Approximates Max-Cut}\label{sec:nwsearlystopworks}

Here we address the correctness part of Theorem \ref{thm:earlystopnwsworks}. As mentioned in the proof outline, the intuition is to couple the outputs of $\text{NaiveWeightedSubsample}(G, n^3)$ and $\text{NaiveWeightedSubsample}(G, \infty)$, but we need to do this carefully to eliminate the dependence of the success probability on $W$. We begin by setting up some auxiliary variables.

Let $M$ be the number of calls made to EstimateAndContract in $\text{NaiveWeightedSubsample}(G, \infty)$, and enumerate these calls $1, 2, \ldots, M$. For each $i$, let $\kappa_i$ be the value of $\kappa$ used in the $i$th call. The $i$th call will add the tuple $(C_i, \kappa_i/2)$ to the list $X$ of edge strength estimates. Also, let $A$ be maximal such that $\kappa_A > W_\text{tot}/n^3$.

Now define iid Bernoulli variables $S_i$ for $i \in [M]$ that are 1 with probability $1 - O(n^{1-d})$, such that if $S_i = 1$, then the $i$th call to EstimateAndContract is consistent with Lemmas \ref{lemma:claim1-2} and \ref{lemma:claim1-3}. (We do not make any assumptions about EstimateAndContract's behaviour if $S_i = 0$.) It is possible to construct such random variables, because the random bit strings used in different calls to EstimateAndContract are mutually independent.

Now, for each $i \in [M]$, let $H_i$ be the set of edges sampled in the $i$th stage of ConstructSparsifier (as run by $\text{NaiveWeightedSubsample}(G, \infty)$). Then we have $\text{NaiveWeightedSubsample}(G, \infty) = \bigcup_{i = 1}^M H_i$ and moreover it is clear that $\text{NaiveWeightedSubsample}(G, n^3)$ is identically distributed to $\bigcup_{i = 1}^A H_i$.

Note that $H_i$ only depends on the first $i$ entries of $X$. In particular, $H_i$ is mutually independent with $\left\{S_j: j > i\right\}$. We now prove a sequence of results, the last of which will be exactly what we want.

\begin{lemma}\label{lemma:earlynws1}
$\Pr[\bigcup_{i = 1}^M H_i\text{ is a }(2\epsilon)\text{-sparsifier} \mid S_i = 1 \forall i \leq M] = 1 - O(n^{-d}).$
\end{lemma}
\begin{proof}
This follows directly from the first part of the argument in Section \ref{sec:nwsinfcorrect}.
\end{proof}
The next lemma is the core ``coupling" part of this argument:
\begin{lemma}\label{lemma:earlynws2}
If $S_i = 1 \forall i \leq M$ and $\bigcup_{i = 1}^M H_i$ is a $(2\epsilon)$-sparsifier, then $\argmax_U F(U; \bigcup_{i = 1}^A H_i)$ achieves a $(1-5\epsilon)$-approximation for max-cut.
\end{lemma}
\begin{proof}
By Corollary \ref{cor:nwsestimate} (which we can apply since we are supposing $S_i = 1 \forall i$), $\bigcup_{i = 1}^A H_i$ will already include an edge (likely with different weight) in the place of any edge with strength $\geq 2W_\text{tot}/n^3$ in the original graph. So for any $e \in \bigcup_{i = A+1}^M H_i$, we must have $k_e < 2W_\text{tot}/n^3$.

Now by Lemma \ref{lemma:sparsifiersmoothnesshard}, $\bigcup_{i = 1}^M H_i$ is $(1/2)$-smooth for sufficiently large $n$. So if we use $w(e; H)$ to denote the weight of an edge $e$ in some graph $H$, it follows that $\frac{1}{2} w(e; \bigcup_{i = 1}^M H_i) \leq k_e$ for all edges $e$. In particular, if $e \in \bigcup_{i = A+1}^M H_i$, we will have $w(e; \bigcup_{i = 1}^M H_i) \leq 2k_e < 4W_\text{tot}/n^3$. Finally, there are at most $n^2/2$ edges in $\bigcup_{i = A+1}^M H_i$, so we have:
$$\sum_{e \in \bigcup_{i = A+1}^M H_i} w(e; \bigcup_{i = 1}^M H_i) < \frac{2W_\text{tot}}{n}.$$
In other words, the total weight of $\bigcup_{i = 1}^M H_i$ that $\bigcup_{i = 1}^A H_i$ is missing is $< 2W_\text{tot}/n$. Now, let $U$ be a maximizer of $F(U; \bigcup_{i = 1}^A H_i)$ and let $S$ be a maximizer of $F(S; G)$. Note that a uniform random cut in $G$ will have expected value $W_\text{tot}/2$, so we have $F(S; G) \geq W_\text{tot}/2$. Putting all of this together, we have:
\begin{align*}
    F(U; G) &\geq \frac{1}{1+2\epsilon} F(U; \bigcup_{i = 1}^M H_i) \text{ (since $\bigcup_{i = 1}^M H_i$ is a sparsifier)} \\
    &\geq \frac{1}{1+2\epsilon} F(U; \bigcup_{i = 1}^A H_i) \\
    &\geq \frac{1}{1+2\epsilon} F(S; \bigcup_{i = 1}^A H_i) \\
    &\geq \frac{1}{1+2\epsilon} (F(S; \bigcup_{i = 1}^M H_i) - \frac{2W_\text{tot}}{n}) \\
    &\geq \frac{1}{1+2\epsilon} ((1-2\epsilon)F(S; G) - \frac{2W_\text{tot}}{n}) \\
    &\geq \frac{1}{1+2\epsilon} ((1-2\epsilon)F(S; G) - \frac{4}{n} F(S; G)) \\
    &\geq (1-5\epsilon)F(S; G),
\end{align*}
which implies the lemma.
\end{proof}

\begin{corollary}\label{cor:earlynws3}
Let $I$ be $0/1$ indicator that is 1 if and only if $\argmax_U F(U; \bigcup_{i = 1}^A H_i)$ achieves a $(1-5\epsilon)$-approximation for max-cut. Then $\Pr[I=1 \mid S_i = 1 \forall i \leq M] = 1 - O(n^{-d})$.
\end{corollary}
\begin{proof}
This directly follows from combining Lemmas \ref{lemma:earlynws1} and \ref{lemma:earlynws2}.
\end{proof}

\begin{corollary}\label{cor:earlynws4}
$\Pr[I=1 \mid S_i = 1 \forall i \leq A] = 1 - O(n^{-d})$.
\end{corollary}
\begin{proof}
This follows from the fact that $I$ only depends on $H_1, \ldots, H_A$. Hence $I, S_1, \ldots, S_A$ are independent from $S_{A+1}, \ldots, S_M$. It follows that we can remove the conditioning on these latter variables while preserving the probability.
\end{proof}

\begin{corollary}\label{cor:earlynws5}
$\Pr[I=1] = 1 - O(n^{1-d}\log n).$ This implies the correctness part of Theorem \ref{thm:earlystopnwsworks}.
\end{corollary}
\begin{proof}
By a union bound, we have $\Pr[S_i = 1 \forall i \leq A] \geq 1 - O(n^{1-d} \cdot |A|)$. We have $\kappa_1 = O(W_\text{tot})$, $\kappa_{i+1} = \kappa_i/2$, and $\kappa_A = \Omega(W_\text{tot})/n^3$, so it follows that $|A| = O(\log n)$. Combining this union bound with the probability from Corollary \ref{cor:earlynws4} yields the desired result.
\end{proof}

\subsection{MSF-based Edge Strength Estimates}

\subsubsection{Proofs of Lemma \ref{lemma:apxkruskal} and Corollary \ref{cor:strengthestimate}}\label{lemmaproof:apxkruskal}

First we address Lemma \ref{lemma:apxkruskal}. Fix a pair of vertices $i, j$. First we show that $\tilde{d}_{i, j} \leq d_{i, j}$. Suppose for the sake of contradiction that $\tilde{d}_{i, j} > d_{i, j}$. Then every edge on the path $P$ from $i$ to $j$ in $\widetilde{\mathcal{T}}$ has weight $> d_{i, j}$. Now take an edge $e$ of weight $d_{i, j}$ along the path from $i$ to $j$ in $\mathcal{T}$. The graph $\mathcal{T} \backslash \left\{e\right\}$ will have the same connected components as $\mathcal{T}$ (and hence $G$) except that one of these components has now been split into two pieces with $i$ in one and $j$ in the other. But now some edge $e'$ along the path from $i$ to $j$ in $\widetilde{\mathcal{T}}$ must connect these two components. So $\mathcal{T} \cup \left\{e'\right\} \backslash \left\{e\right\}$ is also a spanning forest. Moreover, we have $w_{e'} \geq \tilde{d}_{i, j} > d_{i, j} = w_e$ so this spanning forest has total weight strictly greater than $\mathcal{T}$. This is a contradiction.

It remains to show that $\tilde{d}_{i, j} \geq \frac{2}{n^2} d_{i, j}$. Consider any edge $e$ along the path between $i$ and $j$ in $\widetilde{\mathcal{T}}$ and consider the moment right before we contract it in Approximate Kruskal. At this point, $i$ and $j$ are still in distinct supernodes so there exists some edge $e'$ in $\mathcal{T}$ along the path from $i$ to $j$ that has not yet been contracted. Let $W$ denote the total weight of all edges that have not yet been contracted. By definition of the GetEdge operation, we have:
\begin{align*}
    w_e &\geq \frac{2W}{n^2} \\
    &\geq \frac{2w_{e'}}{n^2} \\
    &\geq \frac{2}{n^2} d_{i, j}
\end{align*}
This is true for any $e$ along the path between $i$ and $j$ in $\widetilde{\mathcal{T}}$ so taking the min over all such $e$ yields the desired inequality, completing the proof of the lemma.

Now we finish by addressing Corollary \ref{cor:strengthestimate}. First, if $i, j$ are connected in $G$ then we have:
\begin{align*}
    k_{i, j} &\in [d_{i, j}, n^2d_{i, j}] \\
    &\subseteq [\tilde{d}_{i, j}, \frac{n^4}{2}\tilde{d}_{i, j}],
\end{align*}
as desired. Secondly, if $i, j$ are not connected in $G$ then we have $\tilde{d}_{i, j} = 0 = k_{i, j}$ so the claim trivially follows. (To see the latter equality, note that any vertex-induced subgraph of $G$ containing $i$ and $j$ cannot have a path from $i$ to $j$, so it will have some cut disconnecting $i$ and $j$ and this cut will have value 0.) This completes our proof.

\subsubsection{Additional Structure Theorem for Strength Estimates}

We also take this opportunity to make a straightforward but important observation about these strength estimates:

\begin{lemma}\label{lemma:deletiongivescliques}
Let $M > 0$ be some threshold and let $C = \left\{(i, j): \tilde{d}_{i, j} < M\right\}$. Then $K_n \backslash C$ is a union of disjoint cliques.
\end{lemma}
\begin{proof}
We only need the fact that $\mathcal{\widetilde{T}}$ is a spanning forest for this to work; its approximate maximality is irrelevant. Define $\mathcal{\widetilde{T'}}$ to be the result of removing all edges with weight $< M$ from $\mathcal{\widetilde{T}}$. We will show that $K_n \backslash C$ is just the result of connecting any two vertices from the same connected component in $\mathcal{\widetilde{T'}}$. This would clearly imply the lemma.

First, consider two vertices $i, j$ from distinct connected components of $\mathcal{\widetilde{T'}}$. We have two cases:
\begin{itemize}
    \item If $i, j$ are disconnected in $\mathcal{\widetilde{T}}$, then we have $\tilde{d}_{i, j} = 0 \Rightarrow (i, j) \in C \Rightarrow (i, j) \notin K_n \backslash C$ as desired.
    \item If $i, j$ are connected in $\mathcal{\widetilde{T}}$, then since they are disconnected in $\mathcal{\widetilde{T'}}$, the path between them in $\mathcal{\widetilde{T}}$ must include some edge of weight $< M$. Thus the minimum edge weight on this path must also be $< M$, so this implies that $\tilde{d}_{i, j} < M \Rightarrow (i, j) \in C \Rightarrow (i, j) \notin K_n \backslash C$.
\end{itemize}
It remains to address the case when $i, j$ are in the same connected component of $\mathcal{\widetilde{T'}}$. Since $\mathcal{\widetilde{T'}}$ is a sub-forest of the forest $\mathcal{\widetilde{T}}$, it follows that there is a path from $i$ to $j$ in $\mathcal{\widetilde{T}}$ and moreover every edge on this path has weight $\geq M$. By definition, it follows that $\tilde{d}_{i, j} \geq M \Rightarrow (i, j) \notin C \Rightarrow (i, j) \in K_n \backslash C$. This proves the lemma.
\end{proof}

\subsection{Correctness of Fast Weighted Sampling}\label{lemmaproofs:fwscorrect}

Before we proceed, we establish some terminology. By one ``round" of Algorithm \ref{algo:fws}, we refer to one pass of the outermost loop of the algorithm. By one ``batch" for a given value of $\kappa$, we refer to the $r$ calls made to EstimateAndContract in the innermost loop of the algorithm with that particular value of $\kappa$.

Let us first show that every call made to EstimateAndContract behaves in alignment with Lemmas \ref{lemma:claim1-2} and \ref{lemma:claim1-3} with high probability. In subsequent sections, we will assume that all calls to EstimateAndContract run successfully.

\begin{lemma}\label{lemma:contractunionbound}
With probability $1 - O(n^{5-d})$, every call to EstimateAndContract in FastWeightedSubsample will behave in alignment with Lemmas \ref{lemma:claim1-2} and \ref{lemma:claim1-3} (i.e. the events that the lemmas claim will happen with high probability will all happen).
\end{lemma}
\begin{proof}
Let us crudely bound the number of calls we need to make to EstimateAndContract assuming that they all run correctly. We can then take a union bound over this number of calls.

For each value of $\kappa$, there are at most $n$ calls made to EstimateAndContract, since the $S_i$'s partition $[n]$. For each iteration of the outermost loop, $\kappa$ runs through $O(\log n)$ values since it is iteratively halved from $n^4\tilde{D}/2$ to $\tilde{D}/(2n)$. Moreover, the outermost loops runs at most $O(n^2)$ times, as we will see in Section \ref{sec:fwsestimates} that $L$ will decrease in size by at least 1 each round if EstimateAndContract runs correctly.

Thus in total, we need the first $n \cdot O(\log n) \cdot O(n^2) = O(n^3 \log n)$ calls to EstimateAndContract to all be successful. By a union bound and using the probability estimates in Lemmas \ref{lemma:claim1-2} and \ref{lemma:claim1-3}, this occurs with probability $1 - O(n^{4-d}\log n) = 1 - O(n^{5-d})$, as desired.
\end{proof}

\subsubsection{Estimate and Contract Gives Good Strength Estimates}\label{sec:fwsestimates}

We use the same notions of correctly or incorrectly labeled edges as in Definition \ref{def:labelled}.

Note firstly that $\kappa$ will always be a power of 2 in every call to EstimateAndContract. Moreover, it strictly decreases after each batch of $r$ calls to EstimateAndContract. We want to show that all edges $e$ will be correctly labeled at the end of Fast Weighted Subsampling. To set this up, we make an observation first:

\begin{lemma}\label{lemma:deletionworkswell}
For any round of the outermost loop in Algorithm \ref{algo:fws} and any index $i \in [r]$, $G[S_i] \cap C = \emptyset$. In other words, no edge from $C$ will have both endpoints in the same $S_i$. 
\end{lemma}
\begin{proof}
By Lemma \ref{lemma:deletiongivescliques}, $K_n \backslash C$ will be a union of disjoint cliques. $S_1, \ldots, S_r$ will hence be the vertex sets of these cliques. So any edge with both endpoints in the same $S_i$ will hence belong to $K_n \backslash C$. The conclusion follows.
\end{proof}

The key step to showing that FastWeightedSubsample will label all edges correctly is the following lemma. It is similar to Lemma \ref{lemma:nws1round}.

\begin{lemma}\label{lemma:fws1round}
Consider any round of the outermost loop in Algorithm \ref{algo:fws}. Then in that round, each $S_i$ is the union of some collection of supernodes i.e. no edge between $S_i$ and $S_j$ for distinct $i$ and $j$ has been contracted by $\mathcal{C}$.

Moreover, we have the following:
\begin{enumerate}
    \item After each batch of EstimateAndContract operations with a value $\kappa$, every edge with strength $\geq \kappa$ has been correctly labeled. This implies that at the end of the round, every edge with strength $\geq \tilde{D}/n$ is correctly labeled.
    \item After each batch of EstimateAndContract operations with a value $\kappa$, no edge with strength $< \kappa/2$ has been labeled. This implies that at the end of the round, no edge with strength $< \tilde{D}/(4n)$ is labeled.
    \item After each batch of EstimateAndContract operations with a value $\kappa$, any edge with strength in $[\kappa/2, \kappa)$ will either not be labeled, or will be correctly labeled.
\end{enumerate}
\end{lemma}
\begin{proof}
We proceed by induction on the number of rounds that have elapsed. The base case is vacuously true as at the beginning no edges have been labeled or contracted at all. So assume that all previous rounds have proceeded according to our claim and now we focus on the current one. Define $\tilde{D}$ and $C$ as in the algorithm. Observe by Corollary \ref{cor:strengthestimate} that for any edge $e \in C$ we have:
\begin{align*}
    k_e &\leq \frac{n^4}{2} \tilde{d}_e \\
    &< \frac{\tilde{D}}{4n}
\end{align*}
This implies that each $(\frac{\tilde{D}}{4n})$-strong component of $G$ will be contained in the same connected component $S_i$. Note that this also means that $k_e < \kappa/2$ for any $e \in C$ so $e$ will not already have been contracted. Any edge between distinct $S_i$ and $S_j$ must clearly be in $C$, so this addresses the claim that $S_1, \ldots, S_r$ are each a union of supernodes. It remains to demonstrate properties 1-3.

Throughout this proof, we will disregard contractions when considering $G[S_i]$; we will just view it as a vertex-induced subgraph. So by the strength of an edge in $G[S_i]$, we will mean its strength in the vertex-induced subgraph of $G$ on $S_i$ before any contractions have been applied. 

With this in mind, let us first show property 1. Consider any edge $e$ of strength $\geq \kappa$ that has not yet been contracted. Now $e$ belongs to some $\kappa$-strong component of $G$, and since $\kappa > \frac{\tilde{D}}{4n}$, this component must be entirely contained in one connected component $S_i$. This implies that the strength of $e$ in $G[S_i]$ is also $\geq \kappa$.

Then Lemma \ref{lemma:claim1-2} tells us that $e$ will be contracted and labeled with an estimate of $\kappa/2$. This will be a correct labeling if $k_e \leq 2\kappa$. So we need to show that any edge of strength $\geq 2\kappa$ has already been contracted. Let $\kappa'$ be as referenced and used in Algorithm \ref{algo:fws}. We have two cases:

\begin{enumerate}
    \item If $\kappa$ is not the value used in the first batch of EstimateAndContract operations, this is clear by induction as we previously ran a batch of EstimateAndContract operations with $2\kappa$ in the place of $\kappa$.
    \item If $\kappa$ is from the first batch, then we have two sub-cases depending on the value of $\kappa$:
    \begin{enumerate}
        \item If $\kappa \neq \kappa'$, then this means that in some previous round the algorithm ran a batch of calls to EstimateAndContract with $2\kappa$ and then halved it at the end of the batch to obtain $\kappa$.
        \item If $\kappa = \kappa'$, then we have $\kappa \geq n^4\tilde{D}/2 \geq n^4\tilde{d}_e/2 \geq k_e \geq \kappa \Rightarrow k_e = \kappa \leq 2\kappa$, where we have used the maximality of $\tilde{D}$ and Corollary \ref{cor:strengthestimate}.
    \end{enumerate}
\end{enumerate}
Next we show property 2. Consider any edge $e$ of strength $< \kappa/2$ in $G$. By induction, $e$ cannot have been labeled before this batch as those batches were all run with larger values of $\kappa$. If $e \in C$ then we are immediately done as Lemma \ref{lemma:deletionworkswell} tells us that $e$ will not appear in any connected component $S_i$. If $e \notin C$, then let $S_i$ be the connected component containing $e$. Any vertex-induced subgraph of $S_i$ containing $e$ is also a vertex-induced subgraph of $G$ containing $e$, so it will have a cut of value $< \kappa/2$. This implies that $e$ will also have strength $< \kappa/2$ in $S_i$.

Hence the property will follow from Lemma \ref{lemma:claim1-3} if we can check that every edge $e$ contracted in $\mathcal{C}$ (and in particular contained in $S_i$) has strength $\geq \kappa/2$ in $S_i$. We will show more strongly that it has strength $\geq \kappa$ in $S_i$. Indeed, we know by induction that $e$ must have strength $\geq \kappa$ in $G$. Then as argued earlier, the $\kappa$-strong component of $G$ containing $e$ must be entirely contained in $S_i$, so that the strength of $e$ in $S_i$ is also $\geq \kappa$. This completes the proof of property 2. 





Finally, we show property 3. Consider an edge $e$ with strength in $[\kappa/2, \kappa)$. Any previous batch would have used a value $\geq 2\kappa$, so by the induction hypothesis no edge of strength $< \kappa$ could have been contracted yet. So one of the following must happen:
\begin{enumerate}
    \item $e$ is not labeled.
    \item $e$ is labeled with $k'_e = \kappa/2$. We have $k_e/2 < k'_e \leq k_e$ so in this case $e$ will be correctly labeled.
\end{enumerate}
This completes the proof of the lemma.
\end{proof}

\begin{corollary}\label{lemma:fwsestimate}
Each round of the outermost loop in Algorithm \ref{algo:fws} reduces the size of $L$ by at least 1. Moreover, at the end of Algorithm \ref{algo:fws}, every edge in $G$ will be correctly labeled.
\end{corollary}
\begin{proof}
We start with the first claim. Let $e$ be an edge attaining the maximum edge strength estimate in $L$ i.e. such that $\tilde{d}_e = \tilde{D}$. By Corollary \ref{cor:strengthestimate}, we have $k_e \geq \tilde{d}_e = \tilde{D}$. So in the final batch of calls to EstimateAndContract, we will have $\kappa \leq \tilde{D}/n \leq \tilde{D} \leq k_e$. Thus by Lemma \ref{lemma:fws1round}, $e$ will be correctly labeled in this batch if not an earlier one. This proves the first claim.

Now for the second claim: it is clear from Lemma \ref{lemma:fws1round} that no edge in $G$ will ever be incorrectly labeled, so it suffices to show that every edge in $G$ will be labeled. But this is also clear since the outermost loop will continue as long as there are unlabeled edges, and each iteration will label at least one edge. The conclusion follows.
\end{proof}

\subsubsection{Fast Weighted Subsampling Constructs a Sparsifier}

If every call to EstimateAndContract runs correctly, then we claim that we can apply Lemma \ref{lemma:sparsifierconcentration} for any such list $X$. Indeed, Corollary \ref{lemma:fwsestimate} tells us the following:
\begin{enumerate}
    \item Every edge is assigned a label by some round of EstimateAndContract, thus it is contracted. Therefore each connected component of $G$ is contained in some $C_i$.
    \item Each edge is \emph{correctly} labeled, thus we also have that for any edge $e$ its strength estimate is in $[k_e/4, k_e]$.
\end{enumerate}
Thus we may apply Lemma \ref{lemma:sparsifierconcentration} to find that in this case, $H = \text{ConstructSparsifier}(G, X)$ will be a $(2\epsilon)$-sparsifier of $G$ with probability $1 - O(n^{-d})$.

In summary: for $H$ to be a $(2\epsilon)$-sparsifier, we need all calls to EstimateAndContract to run correctly, and we also need ConstructSparsifier to succeed. By Lemma \ref{lemma:contractunionbound}, the former occurs with probability $1 - O(n^{5-d})$ and the latter occurs with probability $1 - O(n^{-d})$ (assuming the former ran successfully). Thus a union bound tells us that $H$ is a $(2\epsilon)$-sparsifier of $G$ with probability $1 - O(n^{5-d})$, as desired.


\subsection{Efficiency of Naive and Fast Weighted Subsampling}\label{sec:sparsifierspeed}

Here we address the efficiency parts of Theorems \ref{thm:nwsworks}, \ref{thm:earlystopnwsworks}, and \ref{thm:fwsworks}. The analysis is mostly simultaneous for all three theorems, with some delineations at the end.

Note firstly that the initial Approximate Kruskal step in FastWeightedSubsample runs in $O(n \log n)$ queries and the first call to InitializeDS in any of the algorithms runs in $O(n)$ queries so we need only address the EstimateAndContract operations and the final call to ConstructSparsifier. A key lemma we use is the following inequality due to \cite{benczurkarger}:
\begin{lemma}\label{lemma:bkinequality}
(\cite{benczurkarger}) In an undirected weighted graph $G$, we have $\sum_e \frac{w_e}{k_e} \leq |V(G)|-1$.
\end{lemma}

First we address ConstructSparsifier. The following lemma accounts for the $O(n\log^3n/\epsilon^2)$ term in all three theorems.
\begin{lemma}\label{lemma:constructsparsifierspeed}
Suppose that in the input $X$ to ConstructSparsifier, any labeled edges are labeled correctly i.e. for all $i$ and edges $e \in \widetilde{E}(C_i)$ ($\widetilde{E}(C_i)$ is defined as in Lemma \ref{lemma:sparsifierconcentration}), we have $\beta_i \in [k_e/4, k_e]$. (Note that we do not require that every edge is labeled, only that no edges are labeled incorrectly.)

Then ConstructSparsifier runs in $O(n \log^3n/\epsilon^2)$ queries.
\end{lemma}
\begin{proof}
First we make some straightforward observations:
\begin{enumerate}
    \item The initial call to InitializeDS uses $O(n)$ queries.
    \item Each contraction reduces the number of supernodes by at least 1, thus there are $O(n)$ contractions in total. Each contraction takes $O(1)$ queries, thus all the contractions also take $O(n)$ queries total.
    \item There is one call to GetTotalWeight per set that we want to contract. There are $O(n)$ such sets by the above, and each call to GetTotalWeight takes $O(1)$ queries, so this makes for $O(n)$ queries total.
\end{enumerate}
It remains to address the query complexity of the sampling that we need to do. The total number of edges we sample is $\sum_{i = 1}^r \mu_i$ and we use $O(\log n)$ queries to sample each of these edges, thus the total query complexity of this is at most a constant multiple of:
\begin{align*}
    \log n \cdot \sum_{i = 1}^r \mu_i &= \log n \cdot \sum_{i = 1}^r \frac{c_1 \log^2n}{\epsilon^2 \beta_i} W(\widetilde{E}(C_i)) \\
    &= \frac{c_1 \log^3n}{\epsilon^2} \sum_{i = 1}^r \sum_{e \in \widetilde{E}(C_i)} \frac{w_e}{\beta_i} \\
    &\leq \frac{4c_1 \log^3n}{\epsilon^2} \sum_{i = 1}^r \sum_{e \in \widetilde{E}(C_i)} \frac{w_e}{k_e} \text{ (since any labeled edges are labeled correctly)} \\
    &= \frac{4c_1 \log^3n}{\epsilon^2} \sum_{\text{labeled }e} \frac{w_e}{k_e} \\
    &\leq \frac{4c_1 \log^3n}{\epsilon^2} (n-1) \text{ (using Lemma \ref{lemma:bkinequality})} \\
    &= O(\frac{n \log^3 n}{\epsilon^2}).
\end{align*}
Thus the total query complexity of ConstructSparsifier is indeed $O(n\log^3n/\epsilon^2)$, as desired.
\end{proof}

We now address the calls to EstimateAndContract generically before analyzing each of the three algorithms individually:
\begin{lemma}\label{lemma:estimateandcontractgenericefficiency}
Assume that for either NaiveWeightedSubsample or FastWeightedSubsample, all calls to EstimateAndContract run correctly. Enumerate all the calls we ever make to EstimateAndContract $1, 2, \ldots, M$, and for $p \in [M]$ let $G'_p$ and $\kappa_p$ be the graph $G'$ and the value of $\kappa$ used in call $p$. (Here, we assume that any contractions are already included in $G'_p$, so that $G'_p$ may include some supernodes.) Then the total query complexity of all calls to EstimateAndContract is:
$$O(n\log^3n \cdot \max_{e \in G} |\left\{p: e \in G'_p\right\}|).$$
\end{lemma}
\begin{proof}
First we make some simple observations:
\begin{enumerate}
\item Each nontrivial contraction reduces the number of supernodes by at least 1, thus there are $O(n)$ such contractions in total. Each of these takes $O(1)$ queries, so all contractions take $O(n)$ queries total.
\item GetTotalWeight needs to be called at most once for each nontrivial contraction since we are just concerned with $W(G')$, the total weight of all non-contracted edges. Thus there are also $O(n)$ calls to GetTotalWeight that take $O(1)$ queries each, for a total of $O(n)$ queries here as well.
\end{enumerate}
It remains to address the query complexity of the sampling that we need to do. The total number of edges we sample is:
$$c_0\log^2n \cdot \sum_{p = 1}^M \frac{W(G'_p)}{\kappa_p}.$$
We use $O(\log n)$ queries to sample each of these edges, so the total query complexity of this is at most a constant multiple of:
\begin{align*}
    c_0 \log^3n \cdot \sum_{p = 1}^M \frac{W(G'_p)}{\kappa_p} &= c_0 \log^3n \cdot \sum_{p = 1}^M \sum_{e \in G'_p} \frac{w_e}{\kappa_p} \\
    &\leq 2c_0 \log^3n \cdot \sum_{p = 1}^M \sum_{e \in G'_p} \frac{w_e}{k_e} \\
    &= 2c_0\log^3n \cdot \sum_{e \in G} (\frac{w_e}{k_e} \cdot |\left\{p: e \in G'_p\right\}|) \\
    &\leq 2c_0\log^3n \cdot (\sum_{e \in G} \frac{w_e}{k_e}) \cdot \max_{e \in G} |\left\{p: e \in G'_p\right\}| \\
    &\leq 2c_0(n-1)\log^3n \cdot \max_{e \in G} |\left\{p: e \in G'_p\right\}| \text{ (using Lemma \ref{lemma:bkinequality})} \\
    &= O(n\log^3n \cdot \max_{e \in G} |\left\{p: e \in G'_p\right\}|).
\end{align*}
To justify the first inequality step, we know from the proofs of Lemmas \ref{lemma:nws1round} (for NaiveWeightedSubsample) and \ref{lemma:fws1round} (for FastWeightedSubsample) that if $k_e \geq 2\kappa_p$, $e$ will already have been contracted in a previous batch and thus not included in $G'_p$. This completes the proof of the lemma.
\end{proof}

\subsubsection{Efficiency of Naive Weighted Subsampling}

For NaiveWeightedSubsample, we proceed in a straightforward manner as follows:

\begin{lemma}\label{lemma:nwsspeedgeneric}
The total query complexity of EstimateAndContract across all calls in NaiveWeightedSubsample is
$$O(n\log^3n \cdot \min(\log n + \log W, \log T)).$$
\end{lemma}
\begin{proof}
We naively bound $\max_{e \in G} |\left\{p: e \in G'_p\right\}|$ by $M$, the total number of calls to EstimateAndContract. By Lemma \ref{lemma:nwsiterationcount}, $M = O(\min(\log n + \log W, \log T))$. The conclusion follows.
\end{proof}

\begin{corollary}
The total query complexity of EstimateAndContract throughout $\text{NaiveWeightedSubsample}(G, \infty)$ is $O(n\log^3n(\log n + \log W))$.
\end{corollary}

\begin{corollary}
The total query complexity of EstimateAndContract throughout $\text{NaiveWeightedSubsample}(G, n^3)$ is $O(n\log^4n)$.
\end{corollary}

\subsubsection{Efficiency of Fast Weighted Subsampling}

We finish the analysis of FastWeightedSubsample with the following lemma:

\begin{lemma}
For FastWeightedSubsample, $\max_{e \in G} |\left\{p: e \in G'_p\right\}| = O(\log n)$.
\end{lemma}
\begin{proof}
Fix a particular edge $e$ and index $p$ such that $e \in G'_p$. We have both of the following:
\begin{enumerate}
    \item Let $\tilde{D}$ be the value used in the round where call $p$ is made. Note that $\tilde{D} \geq \kappa_p/n^4$. So if $k_e < \kappa_p/(2n^9)$ then we will have $\tilde{d}_e \leq k_e < \kappa_p/(2n^9) \leq \tilde{D}/(2n^5)$. In this case we would have $e \in C$. Since in this case $G'_p$ would be $G[S_i]$ for some $i \in [r]$ (possibly with some contractions), Lemma \ref{lemma:deletionworkswell} tells us that $e \notin G'_p$, which is a contradiction. Thus we must have $k_e \geq \kappa_p/(2n^9) \Leftrightarrow \kappa_p \leq 2n^9k_e$.
    \item On the other hand, as in the proof of Lemma \ref{lemma:estimateandcontractgenericefficiency}, we know from Lemma \ref{lemma:fws1round} that we must have $k_e \leq 2\kappa_p \Leftrightarrow \kappa_p \geq k_e/2$ (else $e$ would already have been contracted).
\end{enumerate}
Thus $\kappa_p \in [k_e/2, 2n^9k_e]$, so there are at most $O(\log n)$ values that $\kappa_p$ can take since $\kappa_p$ will always be a power of 2. Moreover, each value of $\kappa$ is used in at most one batch in the algorithm and within each batch the calls are made on disjoint subsets of $G$. Hence for each possible value of $\kappa_p$, we will get a contribution of at most 1 to the set $\left\{p: e \in G'_p\right\}$. The lemma follows.
\end{proof}

\begin{corollary}
The total query complexity of EstimateAndContract across all calls in FastWeightedSubsample is $O(n\log^4 n)$.
\end{corollary}

\end{document}